\documentclass[12pt]{article}
\usepackage{amsmath}
\usepackage{amsthm}
\usepackage{graphicx}
\usepackage{enumerate}
\usepackage{natbib}
\usepackage{url} 
\usepackage{subcaption}
\usepackage{float}
\usepackage{caption}
\usepackage{amssymb}
\usepackage{amsfonts}
\usepackage{multirow}
\usepackage{mathtools}
\usepackage{setspace}
\usepackage{bm}
\usepackage{titlesec}
\usepackage{enumitem}
\usepackage{algorithmicx}
\usepackage{adjustbox} 
\usepackage{algorithm}
\usepackage{array}
\usepackage{chngcntr}
\usepackage{dcolumn} 
\usepackage{booktabs} 
\usepackage{caption} 
\usepackage{algpseudocode}
\usepackage{mathrsfs}
\usepackage{xcolor}
\newtheorem{theorem}{Theorem} 

\usepackage[colorlinks=true, urlcolor=black, linkcolor=black, citecolor = blue]{hyperref}

\newcommand{\blind}{1}

\newtheoremstyle{named}{}{}{\itshape}{}{\bfseries}{.}{.5em}{\thmnote{#3}#1}
\theoremstyle{named}
\newtheorem*{namedtheorem}{}

\begin{document}

\bibliographystyle{apalike} 
\def\spacingset#1{\renewcommand{\baselinestretch}%
{#1}\small\normalsize} \spacingset{1}


\if1\blind
{
  \title{\bf Exact Bayesian Inference for Multivariate Spatial Data of Any Size with Application to Air Pollution Monitoring}
  \author{Madelyn Clinch$^{*}$ and Jonathan R. Bradley\thanks{
This work is partially supported by the National Science Foundation (NSF)'s Division of Mathematical Sciences (DMS) through the NSF grant DMS-\(2310756\). \textit{}}\hspace{.2cm} \\
     Department of Statistics, Florida State University}
     \date{}
  \maketitle
} \fi

\if0\blind
{
  \bigskip
  \bigskip
  \bigskip
  \begin{center}
    {\bf Exact Bayesian Inference for Multivariate Spatial Data of Any Size with Application to Air Pollution Monitoring}
\end{center}
     \date{}
} \fi
\vspace{-10pt}
\begin{abstract}
\noindent Fine particulate matter and aerosol optical thickness are of interest to atmospheric scientists for understanding air quality and its various health/environmental impacts. The available data are extremely large, making uncertainty quantification in a fully Bayesian framework quite difficult, as traditional implementations do not scale reasonably to the size of the data. We specifically consider roughly 8 million observations obtained from NASA's Moderate Resolution Imaging Spectroradiometer (MODIS) instrument. To analyze data on this scale, we introduce Scalable Multivariate Exact Posterior Regression (SM-EPR) which combines the recently introduced data subset approach and Exact Posterior Regression (EPR). EPR is a new Bayesian hierarchical model where it is possible to sample independent replicates of fixed and random effects directly from the posterior without the use of Markov chain Monte Carlo (MCMC). We extend EPR to the multivariate spatial context, where the multiple variables may be distributed according to different distributions. The combination of the data subset approach with EPR allows one to perform exact Bayesian inference without MCMC for effectively any sample size. Additional motivation is provided via technical results illustrating favorable Kullback-Leibler and covariance properties. We demonstrate SM-EPR using a motivating big remote sensing data application and provide several simulations.   
\end{abstract}

\noindent%
\newpage
\spacingset{1.9} 
\section{Introduction}
\label{s:intro}
\vspace{-10pt}
Air pollution poses a significant risk to public health, environmental safety, and climate stability, making it a crucial on-going area of research for scientists.
Numerous studies have emphasized the association of fine particulate matter with harmful health effects \citep[e.g., see][]{SANG} and also environmental damage such as harming vegetation and acidifying bodies of water \citep{EPA2023Particulate}. Additionally, there are impacts to climate stability by affecting cloud formation processes and influencing radiative forcing \citep[e.g., see][]{li2022scattering}. Accurate global fine particulate matter with a diameter of \(\leq2.5\) microns \((\text{PM}_{2.5})\) measurements are essential for health advisories and environmental agencies when developing policies to address health, environmental, and climate concerns \citep[e.g.,][]{amann2020reducing}. 

Challenges arise when collecting measurements of fine particulate matter. The absence of ground monitoring stations in certain regions of the world creates gaps in the global \(\text{PM}_{2.5}\) data \citep{NASA2015Particulate}. Remote sensing data offers a solution to this issue, but comes with its own set of challenges. For example, cloud coverage can interfere with the ability to collect measurements \citep[e.g., see][]{zhang2021satellite}. Also, satellites observe particles vertically making it difficult to distinguish between high-altitude and ground-level particles \citep{regmi2023analysis}.

Aside from the challenges in collecting the data, there are non-trivial challenges in their analysis. In particular, remote sensing \(\text{PM}_{2.5}\) data is extremely large leading to natural computational difficulties. {The size of the data continues to grow if one considers the variability of covariates, which is often ignored.} For example, aerosol optical thickness (AOT) is a covariate that is known to be useful for predicting \(\text{PM}_{2.5}\) \citep{van2010global}; however, the variability of AOT is not explicitly modeled and the conditional relationship is often assumed to be linear. As a result, predictions are often restricted to \(\text{PM}_{2.5}\) when predictions of AOT are also of interest \citep[e.g.,][]{vogel2022uncertainty}. Jointly modeling global AOT and \(\text{PM}_{2.5}\) data may provide a more comprehensive understanding of air quality and improve predictions.

In general, joint modeling, of variables such as \(\text{PM}_{2.5}\) and AOT, is natural in a Bayesian paradigm where sources of variability can immediately be accounted for. Many studies implement Bayesian approaches to make predictions of \(\text{PM}_{2.5}\) focusing on specific geographic regions such as various states in the US \citep[e.g.,][]{zhang2020application} or spatial domains in China \citep{fu2020mitigating}, while others make predictions on a continental scale \citep{shaddick2018data}. Ideally, one would jointly analyze global \(\text{PM}_{2.5}\) and AOT data on the finest resolution available. However, current fully Bayesian models are not equipped for the size of this data without introducing approximation errors. 

A common conclusion in the literature is that AOT and \(\text{PM}_{2.5}\) are correlated \citep[e.g., see][]{hutchison2005correlating}, suggesting that it is appropriate to model these two variables with a multivariate generalized linear mixed effect model (GLMM), frequently used for multivariate spatial analysis \citep{cressie2011statistics}. In contrast to typical approaches where AOT is treated as a fixed effect in modeling \(\text{PM}_{2.5}\) \citep{mirzaei2020estimation}, we consider a hierarchical framework in which \(\text{PM}_{2.5}\) is conditioned on AOT and an additional level is included that models the variability of AOT. Moreover, possible nonlinear interaction can be modeled via a basis function expansion with shared random effects \citep[e.g., see][]{cressie2022basis}. Unfortunately, for the scale of the available data, carrying out posterior inference for such a model is not practical using standard Markov Chain Monte Carlo (MCMC) techniques. There are an overwhelming number of options available to aid with large spatial data \citep[e.g., see][for a recent review]{heaton2019case}, however, nearly all strategies have an upper bound on the sample size that they can be applied to \citep{saha2023incorporating}. There is an existing method, referred to as the data subset approach, that scales to any data size \citep{SS}, but has only been implemented with computationally inefficient MCMC.

The recent development of Exact Posterior Regression (EPR) \citep{bradley2024generating} offers advantages over traditional MCMC by introducing a ``discrepancy term'' in the GLMM \citep{bradley2020a}, which is important for our application as it allows us to model feedback mechanisms (i.e., signal-to-noise dependence) that are well-known to be present in air pollution data \citep[e.g., see][among others]{liu2019two}. Additionally, when this term is given a specific improper prior, the posterior distribution for the fixed and random effects and discrepancy term takes the form of a generalized conjugate multivariate (GCM) distribution, which can be sampled from directly without MCMC. Consequently, EPR avoids computationally expensive burn-in periods, thinning, convergence checks, and tuning required by MCMC. To date, EPR has not yet been developed for multivariate spatial data. We offer a non-trivial extension to the multivariate spatial setting, allow the multiple data sources to belong to different distributions (i.e., the multi-type setting), and provide theorems that show the SM-EPR's GCM posterior distribution can be efficiently sampled from directly. 

While EPR and approximate Bayesian approaches offer computational advantages over MCMC, none of these methods are able to jointly model the entire global \(\text{PM}_{2.5}\) and AOT data. To address this challenge, we introduce Scalable Multivariate Exact Posterior Regression (SM-EPR), which combines the recently proposed data subset approach \citep{SS, saha2023incorporating} with EPR in a novel way within a multivariate spatial setting. SM-EPR enables one to independently sample replicates from the posterior distribution given a dataset of effectively any size. This novel combination overcomes the computational difficulties when analyzing large datasets, which we illustrate with our application of remote sensing data containing over eight million observations.

Subsampling is a widely recognized approach in big data analysis \citep[e.g., see][for a recent study]{heaton2024minibatch}. A predominant strategy is referred to as divide-and-conquer \citep[e.g., see][among others]{chen2014split}, which divides the larger dataset into manageable subsets, analyzes each separately and then aggregates the results. Another common strategy is to partition the spatial domain into independent regions \citep[e.g., see][]{sang2011covariance} reducing the parameter space in a way that is computationally convenient.

The data subset approach is an exciting new approach that subsets the data into training and holdout sets and writes the data model as the product of a parametric data model for the training data and the true non-parametric density for the holdout data. This split of the data is modeled directly through a prior distribution referred to as a “Subset Model.” This strategy leads to efficient sampling where we iteratively sub-sample the dataset using standard sampling techniques such as simple random sampling (SRS) and use the low-dimensional subset to sample the parameters from the posterior distribution. Our novel combination of EPR with the data subset approach is motivated by our application, as EPR cannot be computed using our multivariate spatial dataset of \(\text{PM}_{2.5}\) and AOT.

We offer technical development to investigate the statistical properties implied by including a discrepancy term and adopting a data subset model. In particular, feedback mechanisms are known to produce signal-to-noise covariances \citep{bradley2020a}, which is a source of dependence known to be present in air pollution data. We derive the posterior cross-covariance between our latent mean and the discrepancy term (an error term) to demonstrate that we are able to model this source of dependence in air pollution data. Additionally, we provide Kullback-Leibler (KL) theory \citep{bishop2006pattern} that shows that our proposed model is robust to model misspecification due to our use of the data subset approach.

The structure of the remainder of our paper is as follows: Section \ref{motivate} introduces our motivating remote sensing dataset, accompanied by a univariate sub-analysis of \(\text{PM}_{2.5}\) modeled/implemented using EPR and integrated nested Laplace approximations \citep[INLA;][]{rue2009approximate}. Section \ref{prelim} gives a review of the data subset approach, GCM, and EPR. Section \ref{method} formally defines SM-EPR, which includes theoretical development describing the computational benefits and key statistical properties. Section \ref{sim_study} presents a simulation study comparing SM-EPR, EPR, and INLA highlighting the advantages of our approach. Section \ref{joint_anal} contains the joint analysis of AOT and \(\text{PM}_{2.5}\), demonstrating significant improvement over the univariate analysis. Section \ref{s:discuss} concludes with a discussion and future directions. A notation table is provided in Supplementary Appendix \ref{appen:notation}. Additional details, simulations, and proofs are also provided in the Supplementary Appendix.
\vspace{-25pt}
\section{Motivating Data}\label{motivate}
\vspace{-10pt}
Our motivating dataset comprises global monthly average \(\text{PM}_{2.5}\) measurements from April 2023 obtained from NASA's Moderate Resolution Imaging Spectroradiometer (MODIS) instrument on the Terra satellite, downloaded from \url{https://giovanni.gsfc.nasa.gov/giovanni/}. The data, available on a \(0.5 \times 0.625\) degree grid, were rasterized, and values of cells corresponding to the coordinates of the \(0.1 \times 0.1\) degree AOT data were extracted using the \texttt{Raster} package in R. Both AOT and \(\text{PM}_{2.5}\) are observed at \(4,097,474\) locations, resulting in over \(8\) million observations. For this application, we make use of a mixed effects model with relevant covariates and a basis function expansion (see Section \ref{joint_anal} for details on model specification). 

\begin{singlespace}
\begin{table}[t!]
\caption{Evaluation metrics for each sub-region analysis using INLA.}
\label{mot_dat_inla}
\begin{center}
\begin{tabular}{cccccc}
\toprule
Sub-region & HOVE & PMCC & CRPS & WAIC & CPU Time \\ \midrule
1 & 58.3286 & 58.3289 & 3.6361 & 5.2152 & 6.25 mins \\
2 & 133.1224 & 133.1227 & 5.9907 & 6.4061 & 5.59 mins  \\ 
3 & 49.8621 & 49.8623 & 2.7436 & 4.6181 & 6.75 mins  \\ 
4 & 253.3219 & 253.3222 & 9.3251 & 6.9821 & 6.11 mins  \\
\bottomrule
\end{tabular}
\end{center}
\begin{flushleft}
\textit{Notes}: The hold out validation error (HOVE), where  20$\%$ (denoted \(0.2N\)) of the original dataset is held out and we use the posterior mean as our predictor. PMCC is the average predictive model choice criterion over the hold out locations \citep{gneiting2007strictly}. WAIC is the Wantanabe-Akaike information criterion \citep{WAIC}. CRPS is the average continuous rank probability score, scaled to favor smaller values. The last column displays the CPU time. 
\vspace{-15pt}
\end{flushleft}
\end{table} 
\end{singlespace}

\vspace{-15pt}
Several R packages are available to implement such a Bayesian model, including INLA and MCMC based implementation via Stan. Due to the extremely large size of the data, we chose to fit the model using INLA as it is known to be more computational efficiency than Stan. Unfortunately, INLA was unable to process the bivariate data on the author’s laptop computer, which is running on Windows 11 with the following specifications: 12th Gen Intel(R) CORE(TM) i7-1260P CPU with 2.1 Ghz. Numerous attempts to implement INLA yielded out of memory and hard error messages. Similarly, our attempt at fitting the multivariate version of EPR could not produce answers as our computing system produced memory errors. As a result, we explored a univariate analysis of \(\text{PM}_{2.5}\) to better assess the computational capabilities of these approaches. INLA was still unable to analyze the global dataset at the desired resolution, but it successfully processed sub-regions containing approximately 700,000 observations. In Table \ref{mot_dat_inla}, we provide several metrics (measuring predictive and computational performance) from the univariate INLA analyses over four different sub-regions. The evaluation metrics calculated from the model fitted using INLA software vary wildly from one region to another leading to inconsistent results. It is possible to analyze the entire univariate \(\text{PM}_{2.5}\) global dataset using EPR with a Central Processing Unit (CPU) time of approximately \(18.63\) hours. However, it is clear from this initial exploration that additional model development is needed to achieve our primary goal of jointly model \(\text{PM}_{2.5}\) and AOT.

\vspace{-25pt}
\section{Preliminaries}\label{prelim}
\vspace{-20pt}
\subsection{Review: The Data Subset Approach}\label{prelim_ss}
\vspace{-10pt}
Let $Z_{k,i}$ denote data at the $i$-th location distributed according to the \(k\)-th distribution with \(i = 1, \dots, N\) and \(k = 1, \dots K\), where the $NK$-dimensional vector $\textbf{z} \equiv (Z_{1,1},\ldots, Z_{K,N})^{\prime}$ and $K$ is a positive integer. In the setting, where $k = K = 1$, we obtain univariate spatial data, and when $K>1$ we have multi-type spatial data. This specification assumes that both data types are observed at the same locations as is the case for our application, however, one could easily modify this to allow for different missing data patterns by data type.

\citet{SS} and \cite{saha2023incorporating} recently developed a semi-parametric approach to Bayesian hierarchical models, which they refer to as the ``Data Subset Approach.'' In this paradigm, they assume that the data is generated from an unknown probability density function (pdf) or probability mass function (pmf) $v(\textbf{z})$. Since $v(\textbf{z})$ is non-parametric and unknown, this implies that a parametric model specification may be misspecified. The data subset approach aims to mitigate the role of model misspecification by dividing the data into training and holdout subsets, and assumes the true pdf/pmf $v$ for the holdout subset, while assuming a parametric model for the training subset. Specifically define the subset indicator $\delta_{i}$, which equals 1 or 0 with $\sum_{i = 1}^{N} \delta_{i} = n\ll N$ and $\delta_{i} = 1$ (= 0) indicates $(Z_{1,i},\ldots, Z_{K,i})^{\prime}$ is included as training (holdout). The subset indicator is then given a prior distribution denoted $f(\bm{\delta}\vert n)$, and can be chosen based on a sampling design from the survey sample literature such as a simple random sample (SRS). It is assumed that the $(N-n)K$-dimensional holdout data vector $\textbf{z}_{-\delta} = (Z_{k,i}: \delta_{i} = 0)^{\prime}$ follows $v(\textbf{z}_{-\delta}) = \int v(\textbf{z})d\textbf{z}_{\delta}$, and the $nK$-dimensional training data vector $\textbf{z}_{\delta} = (Z_{k,i}: \delta_{i} = 1)^{\prime}$ is assumed to follow a misspecified parametric model. In this paper we assume $\textbf{z}_{\delta}$ generally belongs to a member of the exponential family of distributions, where the natural parameter is modeled via a mixed effects representation $\textbf{x}_{k,i}^{\prime}\bm{\beta} + \textbf{g}_{k,i}^{\prime}\bm{\eta} + \xi_{k,i}$.

The $p$-dimensional vector $\textbf{x}_{k,i}^{\prime}$ consists of $p$ known covariates with unknown coefficients $\bm{\beta}\in \mathbb{R}^{p}$. The $r$-dimensional vector $\textbf{g}_{k,i}^{\prime}$ consists of $r$ pre-specfied basis functions evaluated at the $i$-th location, with $r$-dimensional coefficients $\bm{\eta}$ interpreted as random effects. One can use a complete class of basis functions, which are well-known to arbitrarily approximate any random function for large $r$ under reasonable conditions \citep[e.g., see][pg. 102]{obled1986some,cressie2011statistics}. Multi-type spatial covariances are induced by the basis function expansion $\textbf{g}_{k,i}^{\prime}\bm{\eta}$, since $\text{cov}(\textbf{g}_{k,i}^{\prime}\bm{\eta}, \textbf{g}_{m,j}^{\prime}\bm{\eta}) = \textbf{g}_{k,i}^{\prime} \text{cov}(\bm{\eta})\textbf{g}_{m,j}$, which is not necessarily zero for $k,m = 1,\ldots, K$ and $i,j = 1,\ldots, N$. The $NK$-dimensional random effects vector $\bm{\xi}  = (\xi_{1,1},\ldots, \xi_{K,N})^{\prime}$ models non-spatially co-varying random effects and is sometimes called the fine-scale variability term with variance referred to as a nugget. 

The data subset model can be expressed hierarchically, and is defined as the product of the following pdfs/pmfs
\vspace{-20pt}
\begin{align}
    &\text{Data Subset Model:} \hspace{2mm} v(\mathbf{z}_{-\delta})\prod_{\{i:\delta_i=1\}}\text{EF}(Z_{1,i}|\textbf{x}_{1,i}^{\prime}\bm{\beta} + \textbf{g}_{1,i}^{\prime}\bm{\eta} + \xi_{1,i},b_{i};\psi) \notag \\
        &\text{Process Model  1:} \hspace{2mm} f(\boldsymbol{\eta}|\boldsymbol{\theta}) \notag \\
    &\text{Process Model  2:} \hspace{2mm} 
    f(\boldsymbol{\xi}|\boldsymbol{\theta}) \notag \\
    &\text{Parameter Model 1:} \hspace{2mm}f(\boldsymbol{\beta}|\boldsymbol{\theta}) \notag \\
    &\text{Parameter Model 2:} \hspace{2mm} f(\boldsymbol{\theta}) \notag\\
    &\text{Subset Model:} \hspace{2mm} f(\boldsymbol{\delta}|n),
    \label{equation3}
\end{align}

\vspace{-20pt}
\noindent where $\text{EF}(Z\vert \mu,b; \psi)$ is a shorthand for the exponential family pdf/pmf with natural parameter $\mu\in \mathbb{R}$, log-partition function $b\psi(Y)$, $b \in \bm{\theta}$ may be unknown, and the distributions $f(\bm{\eta}\vert \bm{\theta})$, $f(\bm{\xi}\vert \bm{\theta})$, $f(\bm{\beta}\vert \bm{\theta})$, and $f(\bm{\theta})$ are proper. We note that \citet{SS} and \citet{saha2023incorporating} do not jointly model multiple types of data and only consider the univariate setting, where $k=K = 1$.

\citet{SS} assumes that \(\boldsymbol{\delta}\) and \(\textbf{z}\) are independent and show that there exists a \(v\) such that it is possible for \(\boldsymbol{\delta}\) and \(\textbf{z}\) to be independent. This choice is particularly important because it allows one to avoid estimating \(v\). To see this note that the joint posterior distribution can be written as: $f(\boldsymbol{\beta},\boldsymbol{\eta}, \boldsymbol{\xi}, \boldsymbol{\theta}, \boldsymbol{\delta}\vert \textbf{z}) = f(\boldsymbol{\beta}, \boldsymbol{\eta}, \boldsymbol{\xi}, \boldsymbol{\theta} \vert \boldsymbol{\delta},\textbf{z})f( \boldsymbol{\delta}\vert\textbf{z})$. The term \(v\) is a proportionality constant when deriving $f(\boldsymbol{\beta}, \boldsymbol{\eta}, \boldsymbol{\xi},\boldsymbol{\theta} \vert \boldsymbol{\delta},\textbf{z})$ and is not needed. With the added assumption of independence between \(\boldsymbol{\delta}\) and \(\textbf{z}\), we have $f( \boldsymbol{\delta}\vert\textbf{z}) = f(\boldsymbol{\delta}\vert n)$, which again does not require knowledge of \(v\).

The marginal distribution of $(\textbf{z}_{-\delta},\bm{\delta})$ from (\ref{equation3}) is easily verified to be $v(\textbf{z}_{-\delta})f(\bm{\delta}\vert n)$. Consequently, small values of $n$ imply that a larger portion of the the data (i.e., the $(N-n)$-dimensional vector $\textbf{z}_{-\delta}$) is correctly specified. This suggests small values of $n$ can aid with model robustness, which has been empirically verified in the literature \citep{saha2023incorporating}. One can however choose \(n\) to be too small, as it has been shown that smaller values of \(n\) flatten the posterior distribution leading to less precise estimates. This was shown empirically in \citet{SS} and \citet{saha2023incorporating}. Moreover, it is immediate from the iterated variance identity that one would expect posterior variances from the Data Subset Model to generally be larger than the model that sets $\delta_{i} = 1$ with probability one for all $i$, since $\text{Var}(\theta_{i}\vert \textbf{z}_{\delta})\ge E_{\textbf{z}_{-\delta}}(\text{Var}(\theta_{i}\vert \textbf{z})\vert \textbf{z}_{\delta})$, where $\theta_{i}$ is the $i$-th component of $\bm{\theta}$ and $E_{\textbf{z}_{-\delta}}(\cdot\vert \textbf{z}_{\delta})$ is the expected value with respect to $\textbf{z}_{-\delta}\vert \textbf{z}_{\delta}$.
Hence, there is a trade-off between choosing \(n\) for model robustness and statistical precision. In practice, we make use of an elbow plot of a metric to evaluate prediction to select $n$.


Obtaining multiple samples from the posterior distribution can be done using an algorithm referred to as a composite sampler \citep{SS, saha2023incorporating}. The first step of this algorithm samples $\bm{\delta}$ from $f(\bm{\delta}\vert n)$, and the second step samples from $f\left(\boldsymbol{\beta}, \boldsymbol{\eta}, \boldsymbol{\xi}, \boldsymbol{\theta}| \mathbf{z}_{\delta},  n\right)$. Sampling from $f\left(\boldsymbol{\beta}, \boldsymbol{\eta}, \boldsymbol{\xi}, \boldsymbol{\theta}| \mathbf{z}_{\delta},  n\right)$ only requires the use of \(\sum_{i=1}^N \delta_i = n << N\) observations, since $\textbf{z}_{\delta}$ is $n$-dimensional. Thus, for a given \(\boldsymbol{\delta}\), only a small subset of the data is used to sample parameters leading to computational gains, and a sampler that scales with $n$ instead of $N$. The dataset is repeatedly subsampled using this strategy. Consequently, we obtain the computational benefits of using a single subset without ignoring any data from the entire dataset provided we resample enough times so that each data point \(Z_{k,i}\) is selected.

\subsection{Review: Generalized Conjugate Multivariate Distribution}\label{prelim_gcm}
The generalized conjugate multivariate (GCM) distribution was introduced in \citep{bradley2024generating}. This distribution was developed to describe correlated non-identically distributed Diaconis-Ylvisaker (DY) \citep{DY} random variables and marginalizes over a generic \(d\)-dimensional parameter vector \(\boldsymbol{\theta}\). The univariate DY pdf is given by $DY(w_{k,i}\vert \alpha_{k, i}, \kappa_{k,i}) \propto \text{exp} \left\{\alpha_{k,i}w_{k,i} - \kappa_{k,i}\psi_{k,i}(w_{k,i})\right\}$, where \(k = 1, \dots, K\), \(i = 1, \dots, m_{k}\), \(\alpha_{k,i}\) denotes the shape parameter, \(\kappa_{k,i}\) denotes the scale parameter, and \(\psi_{k,i}\) denotes the unit log partition function for the \(i\)-th element corresponding to the \(k\)-th family (e.g., if the \(i\)-th element is a logit-beta random variable then \(\psi_{1,i} = \text{log}\{1 + \text{exp}(w_{1,i})\}\)). The GCM is a multivariate analog defined by the transformation 
\vspace{-20pt}
\begin{equation}
    \mathbf{h} = \boldsymbol{\mu} + \mathbf{V} \mathbf{D}(\boldsymbol{\theta})\mathbf{w}  
    \label{eq.gcm.transform}
\end{equation}
\noindent where \(\mathbf{h}\) is an \(M = \sum_{k = 1}^{K}m_{k}\)-dimensional vector. Similarly, \(\mathbf{w} = (w_{k,i}: k = 1,\ldots, K, i = 1,\ldots, m_{k})^{\prime}\) is an \(M\)-dimensional vector with elements \(w_{k,i} \sim \text{DY}(w_{k,i} \vert \alpha_{k,i}, \kappa_{k,i}; \psi_{k,i})\). The term \(\boldsymbol{\mu}\) is the \(M\)-dimensional location parameter vector, \(\mathbf{V}\) is a \(M\times M\) invertible matrix, \(\mathbf{D}\) is an \(M\times M\) matrix valued function of $\bm{\theta}$, and define the \(M\)-dimensional vector \(\boldsymbol{\psi} \equiv (\psi_{k,i}(\cdot): k = 1,\dots,K, i = 1, \dots, m_k)'\). 

\citet{bradley2024generating} showed that the probability density function for \(\mathbf{h}|\boldsymbol{\mu}, \boldsymbol{V}, \boldsymbol{\alpha}, \boldsymbol{\kappa}\) is equal to the following pdf: 
\vspace{-20pt}
\begin{align*}
    \text{GCM}(\mathbf{h}|\boldsymbol{\mu}, \mathbf{V}, \boldsymbol{\alpha}, \boldsymbol{\kappa}) = \int \frac{\pi(\boldsymbol{\theta})}{\mathscr{N}(\boldsymbol{\theta})} \text{exp}\left[\boldsymbol{\alpha}^{\prime}\mathbf{D}(\boldsymbol{\theta})^{-1}\mathbf{V}^{-1}(\mathbf{h}-\boldsymbol{\mu}) - \boldsymbol{\kappa}^{\prime} \boldsymbol{\psi}\{\mathbf{D}(\boldsymbol{\theta})^{-1}\mathbf{V}^{-1}(\mathbf{h}- \boldsymbol{\mu}) \} \right] d\boldsymbol{\theta},
\end{align*}

\vspace{-20pt}
\noindent where \(\pi(\boldsymbol{\theta})\) is a probability density function for \(\boldsymbol{\theta}\), \(\mathscr{N}(\boldsymbol{\theta})\) is the normalizing constant which has a known form proportional to \(\text{det}\{\mathbf{D}(\boldsymbol{\theta})\}\), and the \(i\)-th element of \(\boldsymbol{\psi}\{\mathbf{D}(\boldsymbol{\theta})^{-1}\mathbf{V}^{-1}(\mathbf{h}- \boldsymbol{\mu})\}\) is the \(i\)-th element of \(\boldsymbol{\psi}\) evaluated at the \(i\)-th element of \(\mathbf{D}(\boldsymbol{\theta})^{-1}\mathbf{V}^{-1}(\mathbf{h}- \boldsymbol{\mu})\). The development of the GCM is particularly important for the recently proposed EPR model, as its posterior takes the form of a GCM and can be sampled from directly via (\ref{eq.gcm.transform}).

Suppose \((\mathbf{h}_1^{\prime}, \mathbf{h}_2^{\prime})^{\prime}\) are jointly GCM where \(\mathbf{h}_1\) is an \(l\)-dimensional vector and \(\mathbf{h}_2\) is an \((M - l)\)-dimensional vector. Then the conditional distribution of \(\mathbf{h}_1 \vert \mathbf{h}_2\) is proportional to the following pdf: 
\begin{align*}
&\text{cGCM}(\mathbf{h}_1|\boldsymbol{\alpha}, \boldsymbol{\kappa}, \boldsymbol{\mu}^{*}, \mathbf{H}, \pi,\mathbf{D}; \boldsymbol{\psi}) \\
&\propto \int \frac{\pi(\boldsymbol{\theta})}{\mathscr{N}(\boldsymbol{\theta})} \text{exp}\left[\boldsymbol{\alpha}^{\prime}\{\mathbf{D}(\boldsymbol{\theta})^{-1}\mathbf{H}\mathbf{h}_1 -\boldsymbol{\mu}^{*}\}
- \boldsymbol{\kappa}^{\prime} \boldsymbol{\psi}\{\mathbf{D}(\boldsymbol{\theta})^{-1}\mathbf{H}\mathbf{h}_1 - \boldsymbol{\mu}^{*} \} \right] d\boldsymbol{\theta}, 
\end{align*}
\noindent where \(\boldsymbol{\mu}^{*} = \mathbf{D}(\boldsymbol{\theta})^{-1}\mathbf{V}^{-1}\boldsymbol{\mu} - \mathbf{D}(\boldsymbol{\theta})^{-1}\mathbf{Q}\mathbf{h}_2\), \(\mathbf{V}^{-1} = (\mathbf{H}, \mathbf{Q})\), \(\mathbf{H}\) is an \(M \times l\) matrix, and \(\mathbf{Q}\) is an \(M \times (M-l)\) matrix. 

\vspace{-25pt}
\subsection{Review: Exact Posterior Regression}\label{prelim_epr}
\citet{bradley2024generating} proposed the following hierarchical model
\begin{align}
    f(\mathbf{z} \vert \boldsymbol{\beta}, \boldsymbol{\eta}, \boldsymbol{\xi}, \boldsymbol{\theta}, \mathbf{q}) &= \prod_{i = 1}^N \text{EF}(Z_{1,i}\vert\textbf{x}_{1,i}^{\prime}\bm{\beta}+\textbf{g}_{1,i}^{\prime}\bm{\eta}+\xi_{1,i}-\tau_{y,1,i}, b_{i}; \psi); \hspace{2mm} i=1,\dots,N, \notag \\ f(\boldsymbol{\xi}\vert\boldsymbol{\beta},\boldsymbol{\eta},\boldsymbol{\theta},\mathbf{q}) &\propto \text{cGCM}(\boldsymbol{\xi}\vert \boldsymbol{\alpha}_{\xi}, \boldsymbol{\kappa}_{\xi}, \boldsymbol{\tau}_{\xi}^{*}, \mathbf{H}_{\xi}, \pi_{\xi}, \mathbf{D}_{\xi}; \psi_{\xi}) \notag\\
    \boldsymbol{\beta}|\boldsymbol{\theta},\mathbf{q} &\sim N(\boldsymbol{\beta} \vert \mathbf{D}_{\beta}\boldsymbol{\tau}_{\beta}, \boldsymbol{\Sigma}_{\beta}(\boldsymbol{\theta})) \notag\\
    \boldsymbol{\eta}|\boldsymbol{\theta},\mathbf{q}  &\sim N(\boldsymbol{\eta} \vert \mathbf{D}_{\eta}\boldsymbol{\tau}_{\eta}, \boldsymbol{\Sigma}_{\eta}(\boldsymbol{\theta})) \notag\\
    f(\mathbf{q}) &= 1 \notag\\
    &\pi(\boldsymbol{\theta}).
    \label{EPR.hier.mod}
\end{align}
\noindent\citet{bradley2024generating} explicitly provides formulations for the univariate scenario \(k=K =1\) with $b_{i}\psi$ defined to be the log-partition function of either the Gaussian, Poisson, or binomial distribution. In this article, we extend EPR to the multi-type scenario (i.e., \(K > 1\)). Let $N(\bm{\beta}\vert \bm{\mu},\bm{\Sigma})$ be a shorthand for the multivariate normal pdf with real vector-valued mean $\bm{\mu}$ and positive definite covariance matrix $\bm{\Sigma}$. 

The \(N\)-dimensional vector \(\boldsymbol{\tau}_{y} = (\tau_{y,1,1},\ldots, \tau_{y,1,N})^{\prime}\) is the difference/error between the traditional mixed effects model \(\mathbf{X}\boldsymbol{\beta} + \mathbf{G}\boldsymbol{\eta} + \boldsymbol{\xi}\) and the true latent process, where  the $N\times p$ matrix $\textbf{X} = (\textbf{x}_{1,1}^{\prime},\ldots, \textbf{x}_{1,N}^{\prime})^{\prime}$ and the $N\times r$ basis function matrix $\textbf{G} = (\textbf{g}_{1,1}^{\prime},\ldots, \textbf{g}_{1,N}^{\prime})^{\prime}$. The term \(\boldsymbol{\theta}\) is a \(d\)-dimensional parameter vector with prior distribution \(\pi(\boldsymbol{\theta})\). Let \(\mathbf{D}_{\beta}(\boldsymbol{\theta})\) be the matrix square root of a positive definite matrix \(\boldsymbol{\Sigma}_{\beta}(\boldsymbol{\theta}) \equiv \mathbf{D}_{\beta}(\boldsymbol{\theta})\mathbf{D}_{\beta}(\boldsymbol{\theta})'\), and \(\mathbf{D}_{\eta}(\boldsymbol{\theta})\) be the matrix square root of a positive definite matrix \(\boldsymbol{\Sigma}_{\eta}(\boldsymbol{\theta}) \equiv \mathbf{D}_{\eta}(\boldsymbol{\theta})\mathbf{D}_{\eta}(\boldsymbol{\theta})'\). The specification of the prior means \(\mathbf{D}_{\beta}\boldsymbol{\tau}_{\beta}\) for \(\boldsymbol{\beta}\) and \(\mathbf{D}_{\eta}\boldsymbol{\tau}_{\eta}\) for \(\boldsymbol{\eta}\) are important as they lead to an expression of the posterior distribution that can be sampled from directly. 

The \(2N\)-dimensional vector \(\boldsymbol{\tau}_{\xi}^{*} = (\boldsymbol{\tau}_y^{\prime} - \boldsymbol{\beta}^{\prime}\mathbf{X}^{\prime}- \boldsymbol{\eta}^{\prime}\mathbf{G}^{\prime}, \boldsymbol{\tau}_{\xi}^{\prime})^{\prime}\), the \(2N \times N\) matrix \(\mathbf{H}_{\xi} = (\sigma_{\xi}\mathbf{I}_N, \mathbf{I}_N)^{\prime}\), the \(2N \times 2N\) matrix \(\mathbf{D}_{\xi} = \sigma_{\xi}\mathbf{I}_{2N}\), \(\sigma_{\xi}^2 \in \boldsymbol{\theta}\), $\pi_{\xi}(\bm{\theta})$ is a point-mass function on the element $\bm{\theta}$, \(\boldsymbol{\tau}_{\xi}\) is a \(N\)-dimensional vector, \(\boldsymbol{\psi}_{\xi}\) is a $2N$-dimensional function where the first $N$ block consists of functions equal to the unit log-partition function of the data and the second $N$-dimensional block consists of elements equal to the unit-log partition function of the Gaussian distribution, and \(\boldsymbol{\alpha}_{\xi}\) and \(\boldsymbol{\kappa}_{\xi}\) are \(2N\)-dimensional shape parameters.

The \((2N + p + r)\)-dimensional vector  \(\boldsymbol{\tau} = (\boldsymbol{\tau}_{y}', \boldsymbol{\tau}_{\beta}', \boldsymbol{\tau}_{\eta}',\boldsymbol{\tau}_{\xi}')' = -\mathbf{D}(\boldsymbol{\theta})^{-1}\boldsymbol{Q}\mathbf{q}\) is called a discrepancy parameter, where $\textbf{q}$ is unknown, \(\mathbf{D}(\boldsymbol{\theta})^{-1} = \text{blkdiag}(\mathbf{I}_{N}, \mathbf{D}_{\beta}(\boldsymbol{\theta})^{-1},\mathbf{D}_{\eta}(\boldsymbol{\theta})^{-1}, \frac{1}{\sigma_{\xi}}\mathbf{I}_{N})\), ``blkdiag'' is the block diagonal operator, \(\mathbf{I}_{N}\) is an \(N \times N\) identity matrix,  and \(\sigma_{\xi} \in \boldsymbol{\theta}\). The \((2N + p + r) \times N\) matrix \(\mathbf{Q}\) represents the eigenvectors of the orthogonal complement of the \((2N + p + r) \times (N + p + r)\) matrix \(\mathbf{H}\) associated with non-zero eigenvalues, where $\textbf{H} = (\mathbf{I}_{N} \hspace{4pt} \mathbf{X} \hspace{4pt} \mathbf{G}: \bm{0}_{p,N} \hspace{4pt}  \mathbf{I}_p \hspace{4pt} \bm{0}_{p,r}:  \bm{0}_{r,N+p}\hspace{4pt}\textbf{I}_{r}:\mathbf{I}_{N} \hspace{2pt} \bm{0}_{N,p+r} )$, $(\textbf{A}\hspace{4pt}\textbf{B}:\textbf{C}\hspace{4pt}\textbf{D}) = \left(\begin{array}{cc}
     \textbf{A}&\textbf{B}  \\
     \textbf{C}&\textbf{D}
\end{array}
    \right)$ for generic real-valued matrics $\textbf{A},\textbf{B},\textbf{C}$ and $\textbf{D}$, and \(\mathbf{0}_{a,b}\) is an \(a  \times b\) matrix of zeros. This specification of $\bm{\tau}$ is quite complex, and \citet{bradley2024generating} primarily motivate this choice by showing that \((\boldsymbol{\xi}^{\prime},\boldsymbol{\beta}^{\prime}, \boldsymbol{\eta}^{\prime})^{\prime}\) falls in the column space of \(\mathbf{H}\). As a result, the coefficients of \(\textbf{q}\) and \((\boldsymbol{\xi}^{\prime},\boldsymbol{\beta}^{\prime}, \boldsymbol{\eta}^{\prime})^{\prime}\) are orthogonal, which is a common strategy to avoid collinearity issues \citep{reich2006effects}.

Multiplying the data, process, and parameter models in (\ref{EPR.hier.mod}), and marginalizing across \(\boldsymbol{\theta}\) results in a posterior distribution that is GCM and therefore can be sampled from directly through (\ref{eq.gcm.transform}). \citet{bradley2024generating} show independent replicates from this GCM posterior can be written as a projection onto the column space of $\textbf{H}$. This projection onto $\textbf{H}$ is referred to as EPR (see Theorem \ref{theorem1} for more details). The terms $\bm{\xi}$ and $\bm{\tau}_{y}$ are interpreted as error terms, and hence, predictions are based on posterior summaries of $Y_{k,i}\equiv \textbf{x}_{k,i}^{\prime}\bm{\beta}+\textbf{g}_{k,i}^{\prime}\bm{\eta}$. \citet{bradley2024generating} considered several settings (e.g., geostatistical models and basis function models for univariate Gaussian, Poisson, and Bernoulli data), and found (empirically) that predictions from this model are very similar to corresponding standard spatial generalized linear mixed models. However, their results suggest that EPR generally produced posterior variances that were similar to slightly larger to that of the corresponding standard spatial generalized linear mixed model.

\vspace{-25pt}
\section{Methodology}\label{method}
\vspace{-15pt}
\subsection{Scalable Multivariate Exact Posterior Regression}\label{SEPR}
\vspace{-10pt}

We present a statement of the model that combines the data subset approach \citep{SS, saha2023incorporating} with Exact Posterior Regression \citep{bradley2024generating}, and extend to the multivariate spatial context for our bivariate (i.e., $K = 2$) logit-beta and Weibull distributed data. The data, process, parameter, and subset models for SM-EPR are defined as follows: \vspace{-20pt}
\begin{align}
    f(\mathbf{z} | \boldsymbol{\beta}, \boldsymbol{\eta}, \boldsymbol{\xi}, \boldsymbol{\theta}, \mathbf{q}, \boldsymbol{\delta}) &= v(\mathbf{z}_{-\delta}) \prod_{\{i: \delta_i = 1\}} \text{Logit-Beta}(Z_{1,i} \vert \mathbf{x}_{1, i}^{\prime}\boldsymbol{\beta} + \mathbf{g}_{1, i}^{\prime}\boldsymbol{\eta} + \xi_{1,i}- \tau_{y,1,i}, \sigma_i^2, \alpha_{z}, \kappa_{z}) \notag \\
    & \times  \prod_{\{i: \delta_i = 1\}} \text{Weibull}(Z_{2,i} \vert \rho_z, \text{exp}\{ -(\mathbf{x}_{2,i}^{\prime}\boldsymbol{\beta} + \mathbf{g}_{2,i}^{\prime}\boldsymbol{\eta} + \xi_{2,i} - \tau_{y,2,i})\}) 
  \notag \\
{f(\boldsymbol{\xi}|\boldsymbol{\beta},\boldsymbol{\eta},\boldsymbol{\theta},\mathbf{q},\bm{\delta})} &{= N\left(\bm{\xi}\vert (\bm{\tau}_{\delta,\xi}^{\prime},\bm{0}_{1,2N-2n})^{\prime},\sigma_{\xi}^{2}\textbf{I}_{2N}\right)}\notag\\
    \boldsymbol{\beta}|\boldsymbol{\theta},\mathbf{q},\bm{\delta} &\sim N(\boldsymbol{\beta}\vert\mathbf{D}_{\beta}\boldsymbol{\tau}_{\beta}, \boldsymbol{\Sigma}_{\beta}(\boldsymbol{\theta})) \notag\\
    \boldsymbol{\eta}|\boldsymbol{\theta},\mathbf{q},\bm{\delta}  &\sim N(\boldsymbol{\eta}\vert\mathbf{D}_{\eta}\boldsymbol{\tau}_{\eta}, \boldsymbol{\Sigma}_{\eta}(\boldsymbol{\theta})) \notag\\
    f(\mathbf{q}\vert \boldsymbol{\delta}) &= 1 \notag\\
    &\pi(\boldsymbol{\theta}) \notag\\
    &f(\boldsymbol{\delta}|n),
    \label{hier.mod.eq}
\end{align}
\noindent where ``Weibull$(Z\vert \rho, \lambda)$'' is a shorthand for the Weibull distribution with shape parameter $\rho>0$ and scale parameter $\lambda$, ``Logit-Beta$(Z\vert \mu, \sigma^{2},\alpha,\kappa)$'' is a shorthand for the logit-beta distribution (a member of the DY family of distributions), which is equal in distribution to the transformation $\mu + \sigma \mathrm{log}\left(\frac{w}{1-w}\right)$ with $w$ distributed as beta with shape parameters $\alpha>0$ and $\kappa-\alpha>0$. When \(\delta_i = 1\) and \(k = 1\), we assume that logit(AOT) is distributed according to the logit-beta distribution with mean \(\mathbf{x}_{1,i}^{\prime}\boldsymbol{\beta} + \mathbf{g}_{1,i}^{\prime}\boldsymbol{\eta} + \xi_{1,i} - \tau_{y,1,i}\), shapes \(\alpha_{z}\) and \(\kappa_{z}\), which is reasonable considering AOT is bounded between zero and one. When \(\delta_i = 1\) and \(k =2\), we assume \(\text{PM}_{2.5}\) is Weibull distributed with shape \(\rho_z\) and scale \(\text{exp}\{-(\mathbf{x}_{2,i}^{\prime}\boldsymbol{\beta} + \mathbf{g}_{2,i}^{\prime}\boldsymbol{\eta} + \xi_{2,i} - \tau_{y,2,i})\}\) \citep{xu2023latent}. The terms \(\mathbf{D}_{\beta}(\boldsymbol{\theta})\), \(\mathbf{D}_{\eta}(\boldsymbol{\theta})\), \(\boldsymbol{\Sigma}_{\beta}(\boldsymbol{\theta})\), and \(\boldsymbol{\Sigma}_{\eta}(\boldsymbol{\theta})\) have the same definition as in Section \ref{prelim_epr}. The terms $v(\textbf{z}_{-\delta})$ and $f(\bm{\delta}\vert n)$ have the same definition as in Section~\ref{prelim_ss}, and it is again assumed that $\textbf{z}$ and $\bm{\delta}$ are independent.

The \((4n + p + r)\)-dimensional vector  \(\boldsymbol{\tau}_{\delta} = (\boldsymbol{\tau}_{{\delta},y}', \boldsymbol{\tau}_{\beta}', \boldsymbol{\tau}_{\eta}',\boldsymbol{\tau}_{{\delta},\xi}')' = -\mathbf{D}_{\delta}(\boldsymbol{\theta})^{-1}\boldsymbol{Q}_{\delta}\mathbf{q}\) is called a discrepancy parameter, where $\textbf{q}$ is unknown, \(\mathbf{D}_{\delta}(\boldsymbol{\theta})^{-1} = \text{blkdiag}(\mathbf{I}_{2n}, \mathbf{D}_{\beta}(\boldsymbol{\theta})^{-1},\mathbf{D}_{\eta}(\boldsymbol{\theta})^{-1}, \frac{1}{\sigma_{\xi}}\mathbf{I}_{2n})\),  \(\sigma_{\xi} \in \boldsymbol{\theta}\), and the ${2}n$-dimensional vector function $\bm{\tau}_{\delta,y} = (\tau_{y,k,i}: \delta_{i} = 1,k = 1, 2)^{\prime}$. When $\delta_{j} = 0$ we define $\tau_{y,k,j} = 0$. This \((4n + p + r) \times 2n\) matrix \(\mathbf{Q}_{\delta}\) represents the eigenvectors of the orthogonal complement of the \((4n + p + r) \times (2n + p + r)\) matrix \(\mathbf{H}_{\delta}\) associated with non-zero eigenvalues, where $\textbf{H}_{\delta} = (\mathbf{I}_{2n} \hspace{4pt} \mathbf{X}_{\delta} \hspace{4pt} \mathbf{G}_{\delta}: \bm{0}_{p,2n} \hspace{4pt}  \mathbf{I}_p \hspace{4pt} \bm{0}_{p,r}:  \bm{0}_{r,2n+p}\hspace{4pt}\textbf{I}_{r}:\mathbf{I}_{2n} \hspace{2pt} \bm{0}_{2n,p+r} )$, the ${2}n\times p$ matrix $\textbf{X}_{\delta} = (\textbf{x}_{k,i}^{\prime}: \delta_{i} = 1, k = {1,2})^{\prime}$, and the ${2}n\times r$ matrix $\textbf{G}_{\delta} = (\textbf{g}_{k,i}^{\prime}: \delta_{i} = 1, k = {1,2})^{\prime}$. Similar to the motivation in \citet{bradley2024generating} this specification of \(\boldsymbol{\tau}_{\bm{\delta}}\) implies that $\textbf{q}$ lies in a column space orthogonal to the column space associated with \((\boldsymbol{\xi}_{\delta}^{\prime},\boldsymbol{\beta}^{\prime}, \boldsymbol{\eta}^{\prime})^{\prime}\) (see Theorem \ref{theorem2} for verification), which again avoids collinearity issues between $(\boldsymbol{\xi}_{\delta}^{\prime},\boldsymbol{\beta}^{\prime}, \boldsymbol{\eta}^{\prime})^{\prime}$ and $\bm{\tau}_{\delta}$, where the ${2}n$-dimensional vector $\boldsymbol{\xi}_{\delta}^{\prime} = (\xi_{k,i}: \delta_{i} = 1, {k =1,2})^{\prime}$.

The statement of our SM-EPR model is specific for our assumptions of logit-beta and Weibull distributed data. However, this hierarchical model can be easily modified for other data models in the exponential family. See Supplementary Appendix \ref{appen:details} for modifications of the hierarchical model when the data are distributed according to other members of the exponential family and for further details on hyperprior specifications (i.e., \(\pi(\boldsymbol{\theta})\)).

\vspace{-15pt}
\subsection{The Posterior Distribution}\label{technical.dev}
 Our derivation of the posterior distribution is similar to results in \citet{bradley2024generating} with three important key differences. First, we allow for multivariate spatial data, whereas \citet{bradley2024generating} only allowed for univariate spatial data. Second, the two different multivariate spatial data vectors can belong to two entirely different data-types (e.g., logit-beta distributed data and Weibull distributed data). Third, the subsampling indicator \(\boldsymbol{\delta}\) is incorporated leading to scalable inference for effectively any \(KN\). Multiplying the densities in (\ref{hier.mod.eq}) and marginalizing over \(\boldsymbol{\theta}\) leads to the GCM posterior distribution stated in Theorem \ref{theorem1}.

\begin{theorem}
    Assume the hierarchical model in (\ref{hier.mod.eq}). Then 
    \vspace{-15pt}
    \begin{align}
        f(\boldsymbol{\delta}\vert \mathbf{z}, n) &= f(\boldsymbol{\delta} \vert n) \notag\\
        (\boldsymbol{\xi}_{\delta}',\boldsymbol{\beta}', \boldsymbol{\eta}', \mathbf{q}')' \vert \mathbf{z}, \boldsymbol{\delta} &\sim \text{GCM}(\boldsymbol{\alpha}_M, \boldsymbol{\kappa}_M, \bm{0}_{{4}n + p + r, 1}, \mathbf{V}, \pi, \mathbf{D}_{\delta}; \boldsymbol{\psi}),
    \end{align}

    \vspace{-15pt}
    \noindent where $\bm{\xi}_{\delta} = (\xi_{k,i}: \delta_{i} = 1,{k = 1,2})^{\prime}$, \(\mathbf{V}^{-1} = (\mathbf{H}_{\delta}, \mathbf{Q}_{\delta})\), { \(\boldsymbol{\psi}(\mathbf{h})\) \(= (\psi_{1}(h_1), \dots,\psi_{1}(h_{n}),\psi_2(h_{n+1}),\dots,\) \(\psi_{2}(h_{2n}), \psi_{*}(h_1^*),\) \(\dots, \psi_{*}(h_{{2n} + p + r}^*))\)} is the \(({4n} + p + r)\)-dimensional unit-log partition function for \(\mathbf{h}\) \(= (h_1, \dots, h_{{2}n},\) \(h_1^*, \dots, h_{{2}n + p + r}^*)' \in \mathbb{R}^{{4}n + p + r}\). {Let \(\psi_{1}(\cdot) = \text{log}(1 +  \text{exp}(\cdot))\), \(\psi_{2}(\cdot) = \text{exp}(\cdot)\), and \(\psi_{*}(\cdot) = (\cdot)^2 \).} The term \(\boldsymbol{\alpha}_M = (\alpha_{k,1}, \dots \alpha_{k,{2}n}, \bm{0}_{1, {2}n + p + r})'\) and the term \(\boldsymbol{\kappa}_M = (\kappa_{k,1}, \dots, \kappa_{k,{2}n}, \frac{1}{2}\mathbf{1}_{1, {2}n + p + r})'\), {where \(\alpha_{1,i} = -\alpha_z \sigma_i^2\), \(\kappa_{1,i} = \kappa_z\), \(\alpha_{2,i} = 1\) and \(\kappa_{2,i} = Z_{2,i}^{\rho_{z}}\).}
    \label{theorem1}
\end{theorem}
\begin{proof}
    See Supplementary Appendix \ref{appen:tech.dev}.
\end{proof}
\noindent  We focus on the $K = 2$ case because our motivating dataset is bivariate. However, Theorem \ref{theorem1} can easily be generalized to allow for several other classes of data models (e.g., members of the exponential family or DY distributions) with $K>2$. For more discussion on generalizing Theorem \ref{theorem1} to $K>2$ see {Supplementary Appendix \ref{appen:tech.dev}. }


\vspace{-10pt}
\subsection{Computational Considerations}\label{computation}
\vspace{-10pt}
An efficient composite sampler can be developed for the SM-EPR's posterior distribution that scales with $n$ instead of $N$. In particular, we sample from {$f(\boldsymbol{\delta}\vert n)$} (e.g., SRS) and then sample from the density $f(\boldsymbol{\beta}, \boldsymbol{\eta}, \boldsymbol{\xi}_{\delta}\vert \mathbf{z}, \boldsymbol{\delta})$ that only makes use of $n$ observations. From Theorem \ref{theorem1}, we have that posterior replicates of {\((\boldsymbol{\xi}_{\delta}^{\prime},\boldsymbol{\beta}^{\prime},\boldsymbol{\eta^{\prime}},\mathbf{q}^{\prime})^{\prime}\)} are equal in distribution to a GCM distribution. The expression for replicates from this specific GCM distribution is stated in Theorem 2.
\begin{theorem}
    Replicates of $\mathbf{q}$, \(\boldsymbol{\delta}\), \(\boldsymbol{\beta}\), \(\boldsymbol{\eta}\), and \(\boldsymbol{\xi}_{\delta} = (\xi_{k,i}: \delta_{i} = 1,k = {1,2})^{\prime}\) from \(f(\boldsymbol{\xi}_{\delta}, \boldsymbol{\beta}, \boldsymbol{\eta}, \mathbf{q}, \boldsymbol{\delta}|\mathbf{z})\) from Theorem \ref{theorem1} have the following property: $\boldsymbol{\delta}_{rep} \sim f(\boldsymbol{\delta}\vert n)$,
           \vspace{-15pt}
    \begin{equation*}
(\boldsymbol{\xi}_{\delta,rep}^{\prime}, \boldsymbol{\beta}_{rep}^{\prime}, \boldsymbol{\eta}_{rep}^{\prime} )^{\prime} = (\mathbf{H}_{\delta}'\mathbf{H}_{\delta})^{-1}\mathbf{H}_{\delta}'\mathbf{w}_{rep}    \vspace{-20pt},
    \end{equation*}
    and $\mathbf{q}_{rep} = \mathbf{Q}_{\delta}'\mathbf{w}_{rep}$, where the subscript ``rep'' represents a single replicate from the posterior distribution, \(\mathbf{w}_{rep} \equiv (\mathbf{y}_{\delta, rep}^{\prime}, \mathbf{w}_{\beta}^{\prime}, \mathbf{w}_{\eta}^{\prime}, \mathbf{w}_{\xi}^{\prime})^{\prime} \), \(\mathbf{y}_{\delta,rep}\) consists of independent DY random variables with \(i\)-th element corresponding to type \(k\). The terms \(\mathbf{w}_{\beta}\), \(\mathbf{w}_{\eta}\), and \(\mathbf{w}_{\xi}\) are obtained by first sampling $\bm{\theta}^{*}$ from its respective prior distribution and then \(\mathbf{w}_{\xi}\) is sampled from a mean zero normal distribution with covariance $\sigma_{\xi}^{2*}\textbf{I}_{{4}n}(\bm{\theta}^{*})$ for $\sigma_{\xi}^{2*}\in \bm{\theta}^{*}$, \(\mathbf{w}_{\beta}\) is sampled from a mean zero normal distribution with covariance $\textbf{D}_{\beta}(\bm{\theta}^{*})\textbf{D}_{\beta}(\bm{\theta}^{*})^{\prime}$, and \(\mathbf{w}_{\eta}\) is sampled from a mean zero normal distribution with covariance $\textbf{D}_{\eta}(\bm{\theta}^{*})\textbf{D}_{\eta}(\bm{\theta}^{*})^{\prime}$. 
\label{theorem2}
\end{theorem}
\begin{proof}
    See Supplementary Appendix \ref{appen:tech.dev}.
\end{proof}

\noindent In Supplementary Appendix \ref{appen:algo}, we explicitly list all the steps required for implementation using Theorem \ref{theorem2}. {While Theorem \ref{theorem2} holds for the bivariate $K = 2$ case, it can easily be extended to $K>2$. That is, this projection representation holds for other data types and for $K > 2$. See Supplementary Appendix \ref{appen:tech.dev} for the formal statement on sampling for generic \(K\).}

EPR is a special case of Theorem \ref{theorem2} when $K = 1$ and the probability of $\delta_{i} = 1$ for all $i$. This novel combination of EPR with the data subset method is extremely exciting, as it is (to our knowledge) the first time a Bayesian spatial GLMM is capable of handling datasets of \textit{any size} without discarding data, imposing additional assumptions on the dependence structure of the data, or requiring MCMC for implementation. We say ``of any size'' since the approach scales with $Kn+p^{3} + r^{3}$ as opposed to $KN+p^{3}+r^{3}$, and \(n\) can be chosen to be arbitrarily small (see Section \ref{prelim_ss} for discussions on the trade-off on this choice). Additionally, we say our model doesn't impose additional assumptions as small \(n\) leads to higher dimensional \(\mathbf{z}_{-\delta}\), which implies fewer assumptions on the data (see Section \ref{prelim_ss} for more details). 

\subsection{Statistical Properties}\label{stat.prop}
Theorems~\ref{theorem1} and \ref{theorem2} demonstrate the computational benefits of our proposed model. Namely, we are able to sample independent posterior replicates of fixed and random effects without the use of MCMC or approximate Bayesian strategies, which is crucial for our goal of analyzing our motivating air pollution dataset. To achieve this goal, the SM-EPR represents a combination of two new strategies in the literature: a discrepancy term (i.e., $\boldsymbol{\tau}$) and the data subset approach. In addition to these computational benefits, we describe two advantageous statistical properties: the ability to model cross-signal-to-noise covariances and robustness to model misspecification.  

The discrepancy term allows us to model signal-to-noise dependence (or feedback) known to be present in air pollution \citep[e.g., see][]{liu2019two}. The posterior distribution in SM-EPR gives a particular parametric form of the posterior cross-signal-to-noise covariance matrix, which is similar to existing parametric forms found in the literature \citep{bradley2023deep}. In particular, the parametric form bears similarity to that of the cross-covariance between the ordinary least squares (OLS) estimator and the OLS residuals.
\begin{theorem}
Let $\boldsymbol{\xi}_{\delta}$, $\boldsymbol{\beta}$, $\boldsymbol{\eta}$, $\textbf{q}$, and $\boldsymbol{\delta}$ follow the GCM model stated in Theorem \ref{theorem1}. Let $\textbf{y}_{\delta} = \textbf{X}_{\delta}\bm{\beta} + \textbf{G}_{\delta}\bm{\eta}+\bm{\xi}_{\delta}$, and let $\textbf{w}_{rep}\equiv (\textbf{y}_{\delta,rep}^{\prime},\textbf{w}_{\beta}^{\prime},\textbf{w}_{\eta}^{\prime},\textbf{w}_{\xi}^{\prime})^{\prime}$ as defined in Theorem \ref{theorem2}. Then, \vspace{-20pt}
\begin{equation*}
cov(\textbf{y}_{\delta},\boldsymbol{\tau}_{y}\vert \textbf{z},\boldsymbol{\delta}) = -\textbf{J}\textbf{H}_{\delta}(\textbf{H}_{\delta}^{\prime}\textbf{H}_{\delta})^{-1}\textbf{H}_{\delta}^{\prime}\text{cov}(\textbf{w}_{rep}\vert \boldsymbol{\alpha}_M,\boldsymbol{\kappa}_M)\left\lbrace\textbf{I}_{{4}n+p+r} - \textbf{H}_{\delta}(\textbf{H}_{\delta}^{\prime}\textbf{H}_{\delta})^{-1}\textbf{H}_{\delta}^{\prime}\right\rbrace \textbf{J}^{\prime},
\end{equation*}

\vspace{-15pt}
\noindent
where $\textbf{J} = (\textbf{I}_{{2}n},\bm{0}_{{2}n,p},\bm{0}_{{2}n,r},\bm{0}_{{2}n,{2}n})$ and $\boldsymbol{\alpha}_M$ and $\boldsymbol{\kappa}_M$ are the same as those defined in Theorem \ref{theorem1}.
\label{theorem3}
\end{theorem}
\begin{proof}
    See Supplementary Appendix \ref{appen:tech.dev}.
\end{proof}
\noindent
{Similar to Theorem \ref{theorem1} and Theorem \ref{theorem2}, Theorem \ref{theorem3} can be generalized to other data types and we provide the details for generic \(K\) in Supplementary Appendix \ref{appen:tech.dev}.}

In general, the parametric form in Theorem~\ref{theorem3} can be difficult to validate in practice. In our simulation studies, we compare to models that do not include a discrepancy term when the simulated data is generated without such dependence. We find little difference between prediction and regression estimation between these models suggesting that the results are robust to misspecifying that a cross-signal-to-noise-covariance matrix is present when it is not present.
 
 The data subset approach is particularly useful when the model is misspecified. Recall the SM-EPR is assumed to be misspecified, where the true data generating mechanism is actually $v(\textbf{z})$, and for a given $\bm{\delta}$ the SM-EPR correctly assumes that $\textbf{z}_{-\delta}$ is generated from the density $v(\textbf{z}_{-\delta})=\int v(\textbf{z})d\textbf{z}_{\delta}$. As such, intuition would suggest that SM-EPR's model for $(\textbf{z},\bm{\delta})$ is closer to the correct specification $v(\textbf{z})f(\bm{\delta}\vert n)$ than a parametric model assumed for the entire dataset. 
 
We formally explore this intuition through the KL divergence. In particular, consider a generic proper Full Data Model, $\prod_{k = 1}^{K}\prod_{i = 1}^{N}f(Z_{k,i}\vert \bm{\gamma})$ and a distribution $f(\bm{\gamma}\vert \bm{\delta})$, where $\bm{\gamma}\in \mathbb{R}^{\ell}$ is a generic $\ell$-dimensional real-valued vector. Under this setup, the marginal distribution of $(\textbf{z},\bm{\delta})$ from this ``Full Model'' is given by $m_{FULL}(\textbf{z},\bm{\delta}) =  \int \prod_{k = 1}^{K}\prod_{i = 1}^{N}f(Z_{k,i}\vert \bm{\gamma})f(\bm{\gamma}\vert \bm{\delta})d \bm{\gamma} f(\bm{\delta}\vert n)$, and the corresponding Data Subset Model's marginal distribution of $(\textbf{z},\bm{\delta})$ is given by, $m_{SUB}(\textbf{z},\bm{\delta}) = v(\textbf{z}_{-\delta})\int \prod_{k = 1}^{K}\prod_{\{i:\delta_{i} = 1\}}f(Z_{k,i}\vert \bm{\gamma})f(\bm{\gamma}\vert \bm{\delta})d \bm{\gamma}f(\bm{\delta}\vert n)$.
In Theorem 4, we investigate the KL divergence between $m_{FULL}$ and the true marginal distribution $m_{TRUE}(\textbf{z},\bm{\delta}) = v(\textbf{z})f(\bm{\delta}\vert n)$ and the KL divergence between $m_{TRUE}$ and $m_{SUB}$.
 \begin{theorem}
     Suppose $\textbf{z}$ is a sample from the density $v(\textbf{z})$, and let $\bm{\delta}$ be drawn from the Subset Model $f(\bm{\delta}\vert n)$ independently of $\textbf{z}$. Denote the KL divergence between generic models $f$ and $g$ with $KL\{f||g\} \equiv \sum_{\delta}\int f(\textbf{z},\bm{\delta}) log\left(\frac{f(\textbf{z},\bm{\delta})}{g(\textbf{z},\bm{\delta})}\right)d\textbf{z}$. Then we have the following properties: (a) For every fixed $n$ and $N$ with $n<N$ it follows that $KL\{m_{TRUE}(\textbf{z},\bm{\delta}) \hspace{2pt}|| \hspace{2pt}m_{FULL}(\textbf{z},\bm{\delta})\}\ge KL\{m_{TRUE}(\textbf{z},\bm{\delta})\hspace{2pt}||\hspace{2pt} m_{SUB}(\textbf{z},\bm{\delta})\}$; and (b) Suppose that for a fixed $n$ that \\$\underset{N\rightarrow \infty}{lim}KL\{v(\textbf{z})f(\bm{\delta}\vert n) || m_{FULL}(\textbf{z},\bm{\delta}\} = {0}$. Then $\underset{N\rightarrow \infty}{lim}KL\{m_{TRUE}(\textbf{z},\bm{\delta})\hspace{2pt} ||\hspace{2pt} m_{SUB}(\textbf{z},\bm{\delta})\} = {0}$.
     \label{theorem4}
 \end{theorem}
\begin{proof}
    Statement (b) follows immediately from Statement (a). See Supplementary Appendix \ref{appen:tech.dev} for the proof of Statement (a).
\end{proof}
Theorem \ref{theorem4} immediately applies to SM-EPR in (\ref{hier.mod.eq}) if we set $K = 2$, $\bm{\gamma} = \{\bm{\xi},\bm{\beta},\bm{\eta},\textbf{q},\bm{\theta}\}$, and define the Full Data Model and $f(\bm{\gamma}\vert \bm{\delta})$ according to (\ref{hier.mod.eq}).
Similarly, Theorem \ref{theorem4} immediately applies to the DSM in (\ref{equation3}) if we set $K = 1$, $\bm{\gamma} = \{\bm{\xi},\bm{\beta},\bm{\eta},\bm{\theta}\}$, and define the Full Data Model and $f(\bm{\gamma}\vert \bm{\delta})$ according to (\ref{equation3}). For any $n<N$, Theorem \ref{theorem4} suggests that it is more reasonable to assume $m_{TRUE}(\textbf{z},\bm{\delta})\approx m_{SUB}(\textbf{z},\bm{\delta})$ than $m_{TRUE}(\textbf{z},\bm{\delta})\approx m_{FULL}(\textbf{z},\bm{\delta})$. Thus, the data subset approach is more robust to departures from $m_{TRUE}(\textbf{z},\bm{\delta})$ relative to the model that does not subsample. As such, Theorem \ref{theorem4} provides support for augmenting $m_{SUB}$ with $\bm{\beta},\bm{\eta},\bm{\xi}, \textbf{q}$ and $\bm{\theta}$ leading to SM-EPR instead of augmenting $m_{FULL}$ with $\bm{\beta},\bm{\eta},\bm{\xi}, \textbf{q}$ and $\bm{\theta}$ for Bayesian inference. However, as discussed in Section~\ref{prelim_ss}, there is a trade-off between choosing \(n\) for model robustness and statistical precision. In practice, we suggest the use of an elbow plot based on a predictive metric to choose the value of \(n\).

\vspace{-20pt}
\section{Simulations}\label{sim_study}
\vspace{-10pt}
In this section, we illustrate the inferential/computational advantages of using SM-EPR, and compare to standard approaches. Additionally, we compare our method to SM-EPR without subsampling (i.e., $\delta_{i}\equiv 1$ with probability one), which will be referred to as Multivariate Exact Posterior Regression (M-EPR), to highlight the computational benefits of the data subset approach. M-EPR and SM-EPR are ``discrepancy models'' since they include $\textbf{q}$. We will compare to a spatial GLMM that sets $\textbf{q} = \bm{0}_{2N,1}$ that is fitted using INLA. 
\begin{figure}[t!]
\centering
  \centering
  \includegraphics[width=0.75\linewidth]{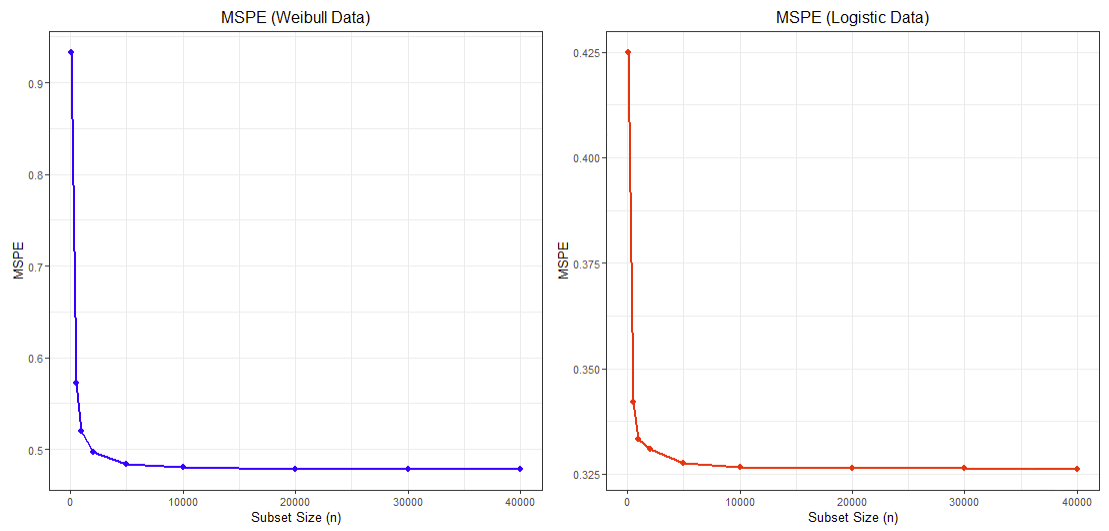}  
\caption{MSPE over various subset sizes (n) for a single simulated dataset.}
\label{large_elbow}
\end{figure}

Define the locations \(\mathbf{s} \in \{1, 2, \dots, {M}\}\) in a one-dimensional spatial domain. We observe \(N=0.8\times {M}\) randomly selected locations out of \({M}\). Both data types will be observed at the same randomly selected locations resulting in a total of \(0.8\times {2M}\) total observations and  \(0.2\times {2M}\) observations held out for testing prediction. We set $M = 60,000$ to be fairly large to assess the computational performance of SM-EPR. In Supplementary Appendix \ref{appen:sim_study_mcmc}, we compare SM-EPR to a spatial GLMM with $\textbf{q} = \bm{0}_{2N,1}$ implemented via Stan; however, this study required $N$ to be small as MCMC was too computationally demanding otherwise. We simulate data with two response variables distributed according to the logistic (a special case of the logit-beta distribution when $\alpha = \kappa/2 = 1$) and Weibull distributions to match the assumptions in Section \ref{joint_anal}. Additional univariate simulations of Gaussian, Poisson, and Bernoulli univariate spatial data are illustrated in Supplementary Appendix \ref{appen:addsims}. 

We consider two Bernoulli distributed covariates \(x_1^{{*}}(s_i)\) observed with probability \(\text{expit}(\frac{1}{{M}} s_i)\) and \(x_2^{{*}}(s_i)\) observed with probability \(\text{expit}(\frac{-0.01}{{M}} s_i)\), where $\text{expit}(x) = \text{exp}(x)/(1+\text{exp}(x))$. Let $Z_{1,i}$ be generated from the logistic distribution with mean 
\vspace{-20pt}
\begin{equation*}
1 -2x_1^{{*}}(s_i) - 2x_2^{{*}}(s_i) + \sum_{j = 1}^{15} \left[ g_{1,j}^{{*}}(s_i)\eta_{1,j}^{*} + g_{M, j}^{{*}}(s_i)\eta_{M, j}^{*} \right] + \xi_{{1}}(s_i)\vspace{-20pt},
\end{equation*}
\noindent and scale 0.6. Let $Z_{2,i}$ be distributed Weibull with scale parameter one and minus log conditional mean, $- 0.7 - 1.5x_1^{{*}}(s_i) - x_2^{{*}}(s_i)  - 0.25Z_1(s_i)+ \sum_{j = 1}^{15} \left[ g_{2,j}^{{*}}(s_i)\eta_{2, j}^{*} + g_{M, j}^{{*}}(s_i)\eta_{M, j}^{*}\right] + \xi_{{2}}(s_i)$, where $\eta_{1,j}^{{*}}\sim N(\eta_{1,j}^{{*}}\vert 0,0.81)$, $\eta_{2,j}^{{*}}\sim N(\eta_{2,j}^{{*}}\vert 0,0.04)$, $\eta_{M,j}^{{*}}\sim N(\eta_{M,j}^{{*}}\vert 0,0.81)$, $\xi_{1}(s_{i})\sim N(\xi_{1}(s_{i})\vert 0,0.15)$, and $\xi_{2}(s_{i})\sim N(\xi_{2}(s_{i})\vert 0,0.08)$. The terms \(\{g_{1,j}^{{*}}\}\),\(\{g_{2,j}^{{*}}\}\), and \(\{g_{M,j}^{{*}}\}\) are Gaussian radial basis functions with equally spaced knot points. For an example realization see Supplementary Appendix \ref{sim_fig}.

{We implement each model assuming \(\mathbf{x}_{1,i} = (1, x_{1}^{*}(s_i), x_2^*(s_i), \mathbf{0}_{1,4})^{\prime}\), \(\mathbf{x}_{2,i} = (\mathbf{0}_{1,3},1, x_1^*(s_i),\)\\\(x_2^*(s_i), Z_1(s_i))^{\prime}\),\(\mathbf{g}_{1,i} = (g_{1,1}^*(s_i), \dots g_{1,15}^*(s_i),\mathbf{0}_{1,15},g_{M,1}^*(s_i),\dots,g_{M,15}^*(s_i))^{\prime}\), and \(\mathbf{g}_{2,i} = (\mathbf{0}_{1,15},\)\\\(g_{2,1}^*(s_i), \dots g_{2,15}^*(s_i),g_{M,1}^*(s_i),\dots,g_{M,15}^*(s_i))^{\prime}\).} {In this section, the basis functions are correctly specified, although SM-EPR and EPR both incorrectly include a discrepancy term. We consider the case of misspecified basis functions in Supplementary Appendix \ref{appen:miss}. We adopt the hyperpriors stated in Appendix \ref{appen:details}. The subset size $n = 5,000$ for SM-EPR was selected using the elbow plots in Figure \ref{large_elbow} of the mean squared prediction error for the two response variables from one example simulation.} 




\begin{table}[ht!]
\caption{Evaluation metrics for models fit using SM-EPR with subset size \(n = 5,000\), M-EPR, ILNA, and Stan for each type of data for \({2}M = 120,000\).}
\label{sim_n_60000_biv}
\begin{center}
\begin{tabular}{ccccc}
\toprule
Approach & CPU & MSPE & MSE & CRPS \\ \midrule
\multicolumn{5}{c}{Weibull Response Data} \\ \midrule
SM-EPR & 12.1375 & 0.4838 &  0.2055 & 0.1504 \\ 
 & ( 11.3122, 12.9627) & (0.4719, 0.4957) & (0.1766, 0.2343) & (0.1489, 0.1519) \\ \midrule
M-EPR  & 72.1699 & 0.4777 &  0.1899 & 0.1765 \\ 
 & (70.4229 73.9169) & (0.4688, 0.4866) & (0.1473, 0.2325) & (0.1754, 0.1775) \\ \midrule
INLA & 137.2539 &  0.4958 &  18.0356 & 0.1787 \\ 
 & (125.4787, 149.0291) & (0.4814, 0.5102) & (0.0000, 37.5654) & (0.1775, 0.1799) \\ \midrule
 \multicolumn{5}{c}{Logistic Response Data} \\  \midrule
SM-EPR & 12.1375 & 0.3275 &  0.2055 & 0.3359 \\ 
 & ( 11.3122, 12.9627) & (0.3228, 0.3321) & (0.1766, 0.2343) & (0.3208, 0.3510) \\ \midrule
M-EPR  & 72.1699 &  0.3260 &  0.1899 & 0.3843 \\ 
 & (70.4229, 73.9169) & (0.3252, 0.3268) & (0.1473, 0.2325) & (0.3757, 0.3930) \\ \midrule
INLA & 137.2539 & 0.3257 &  18.0356 & 0.4280 \\ 
 & (125.4787, 149.0291) & (0.3256, 0.3258) & (0.0000, 37.5654) & (0.4277, 0.4282) \\ \bottomrule
\end{tabular}
\end{center}
\begin{flushleft}
\textit{Note}: The computational approach, the average CPU time in seconds, the mean square prediction error \(\text{MSPE} = \frac{1}{N} \sum_{i = 1}^N (Y_i - \text{E}[{Y}_i \vert \textbf{z}])^2 \), the average MSE between true and estimated coefficients, and the average CRPS. All averages are taken over 50 replicates along with plus or minus two standard deviations. The MSPE was calculated on the log-scale for the Weibull setting.
\vspace{-20pt}
\end{flushleft}
\end{table}
The central processing unit (CPU) time measured in seconds, the mean squared prediction error (MSPE) between the true latent process and the predicted latent process, the mean squared estimation error (MSE) between the true coefficients of the fixed and random effects and the estimated \(\boldsymbol{\beta}\) and \(\boldsymbol{\eta}\), and the continuous rank probability score (CRPS) were used to evaluate the different approaches. Table \ref{sim_n_60000_biv} presents the average CPU time, MSPE, MSE, and CRPS over 50 replicates with confidence intervals constructed using plus or minus two standard deviations. { SM-EPR with subset size of \(n = 5,000\) produces similar inference to M-EPR and INLA illustrated by the overlapping confidence intervals for MSPE and MSE while resulting in a lower CRPS and a significant computational advantage. Thus, these results illustrate that it is possible to use SM-EPR to obtain both faster implementation (i.e., CPU time) with similar to better inferential performance (as measured by MSPE, MSE, and CRPS).}



\vspace{-25pt}
\section[Joint Analysis of AOT and PM 2.5 ]{Joint Analysis of AOT and \(\text{PM}_{2.5}\) Data}\label{joint_anal}
\vspace{-10pt}
For our application we build on the motivating univariate analysis in Section 2 followed by our multivariate analysis of primary interest.
Covariates include average land and sea surface temperature, the Normalized Difference Vegetation Index (NDVI), cloud fraction, and logit-transformed AOT, all  downloaded from \url{https://neo.gsfc.nasa.gov/}. We use a bisquare basis function expansion \citep{cressie2008fixed}. For univariate analysis of \(\text{PM}_{2.5}\), 198 basis functions were chosen based on a sensitivity analysis (see Supplementary Appendix \ref{basis_sens}). For the bivariate analysis of \(\text{PM}_{2.5}\) and AOT, \(594\) basis functions were used with \(198\) unique to each response and \(198\) common across both responses.

To demonstrate the computational benefits of the data subset approach, we compare EPR to the univariate version of SM-EPR, referred to as Scalable Exact Posterior Regression (S-EPR), where AOT is treated as a linear covariate for \(\text{PM}_{2.5}\). Recall from Section \ref{motivate}, the bivariate analysis using M-EPR failed on the author’s laptop. {Table \ref{mot_dat_epr} presents a sensitivity analysis comparing various subset sizes for S-EPR. Notably, a subsample size of \(300,000\) yields results comparable to EPR while reducing the computation time from \(18.63\) to \(1.77\) hours.}  Table \ref{mot_dat_epr} shows the benefit of using a data subset approach over a model that does not subsample in the univariate context, and indicates that $n=300,000$ is reasonable. 

Another benefit of the data subset approach is that it gives us the capability to conduct a bivariate spatial analysis. In particular, SM-EPR using a subset size of \(n=300,000\) did not encounter memory error and had a CPU time of 18.82 hours. This CPU time is noteworthy, as it is similar to the 18.63 hours CPU time for the univariate analysis of \(\text{PM}_{2.5}\) with EPR. That is, SM-EPR can model two large spatial response variables in approximately the time it takes to model one response variable with EPR. The evaluation metrics for SM-EPR are presented in Table \ref{biv_anal_tab}, alongside those from the univariate model fitted with EPR revealing a significantly smaller prediction error on hold out data as well as smaller PMCC, CRPS, and WAIC for SM-EPR compared to EPR. This indicates that we are able to successfully leverage multivariate dependence to improve prediction.

\begin{table}[t!]
\caption{Evaluation metrics for models fit with EPR and S-EPR in univariate analysis of \(\text{PM}_{2.5}\).}
\label{mot_dat_epr}
\begin{center}
\begin{tabular}{cccccc}
\toprule
Subset Size & HOVE & PMCC & CRPS & WAIC & CPU Time \\ \midrule
5,000 & 72.8999 & 91.8573 & 2.7335  & 8.4133 & 1.51 mins \\
10,000 & 72.7569 & 79.1911 & 2.8239  & 6.7660 & 3.48 mins \\ 
50,000 & 70.8170 & 75.3292 & 2.8959  & 6.7445 & 17.86 mins  \\ 
100,000 & 70.3535 & 75.2565 & 2.8867  & 6.7432 & 35.22 mins  \\ 
200,000 & 69.9981 & 74.4716 & 2.9002 & 6.7426 & 1.21 hrs \\
300,000 & 69.8660 & 74.6802 & 2.8829  & 6.7422 & 1.77 hrs \\
\midrule
EPR & 69.5822 & 74.5495 & 2.8800 & 6.7416 & 18.63 hrs \\ \bottomrule
\end{tabular}
\end{center}
\begin{flushleft}
\textit{Note}: The first 6 rows present S-EPR results with various subset sizes and the last row of the table presents EPR results.\vspace{-20pt}
\end{flushleft}
\end{table}
Additionally, modeling the shared random effects between \(\text{PM}_{2.5}\) and AOT not only improves \(\text{PM}_{2.5}\) predictions but also provides insight into the complex relationship between the two variables in different global regions. Figure \ref{biv_anal_plot} displays predictions for both variables compared to the observed data along with the posterior mean of the shared basis functions illustrating the influence of the shared random effects on both variables in certain regions of the world. The positive effects observed in the southern region of Russia indicate that the spatial variability captured by the basis functions positively impacts both variables. These findings are valuable for scientists and policy makers because identifying and targeting areas with positive shared random effects could lead to interventions that simultaneously improve \(\text{PM}_{2.5}\) and AOT levels by leveraging their underlying spatial relationship. In the western region off the coast of Europe, the negative shared random effects suggest the spatial variability contribute to the lower AOT and \(\text{PM}_{2.5}\) measurements. Typically, high values of AOT and \(\text{PM}_{2.5}\) are associated with hazy conditions and poor air quality, while lower values indicate clear skies and healthier air conditions. 

\begin{table}[t!]
\caption{Evaluation metrics for the bivariate model for \(\text{PM}_{2.5}\) and logit(AOT) implemented with SM-EPR with subset size \(n = 300,000\) and for the univariate model for \(\text{PM}_{2.5}\) implemented with EPR.}
\label{biv_anal_tab}
\begin{center}
\begin{tabular}{cccccc}
\toprule
Variable & HOVE & PMCC & CRPS & WAIC & CPU \\  \midrule
\multicolumn{6}{c}{EPR Univariate Model}\\ \midrule
 \(\text{PM}_{2.5}\) & 69.5822 & 74.5495 & 2.8800 & 6.7416 & 18.63 hrs\\  \midrule
\multicolumn{6}{c}{SM-EPR Bivariate Model} \\  \midrule
 \(\text{PM}_{2.5}\) & 57.8830 & 61.9588 & 2.3826 & 6.7124 & 18.82 hrs \\ 
logit(AOT) & 0.3314 & 0.3366 & 0.3529 & 1.7941 & 18.82 hrs  \\ 
 \bottomrule
\end{tabular}
\vspace{-10pt}
\end{center}
\end{table}
In addition to the insight into the non-linear relationship between the two variables, we can also make inference on the variables used as predictors in our model. We identify significant covariates for predicting \(\text{PM}_{2.5}\) and logit(AOT) by determining which credible intervals contain zero. Using SM-EPR we found that logit(AOT), cloud fraction, land surface temperature, sea surface temperature, and NVDI are significant covariates for \(\text{PM}_{2.5}\) and cloud fraction is a significant covariate for AOT. Various studies utilizing Bayesian approaches and machine learning approaches have also reported significant/important covariates used for predicting \(\text{PM}_{2.5}\) measurements. AOT, temperature, and vegetation index \citep{chen2021estimating, chen2021estimation} are some of the significant variables discussed in the literature that align with our findings.

\begin{figure}[t!]
\centering
  \centering
  \includegraphics[width=1\linewidth]{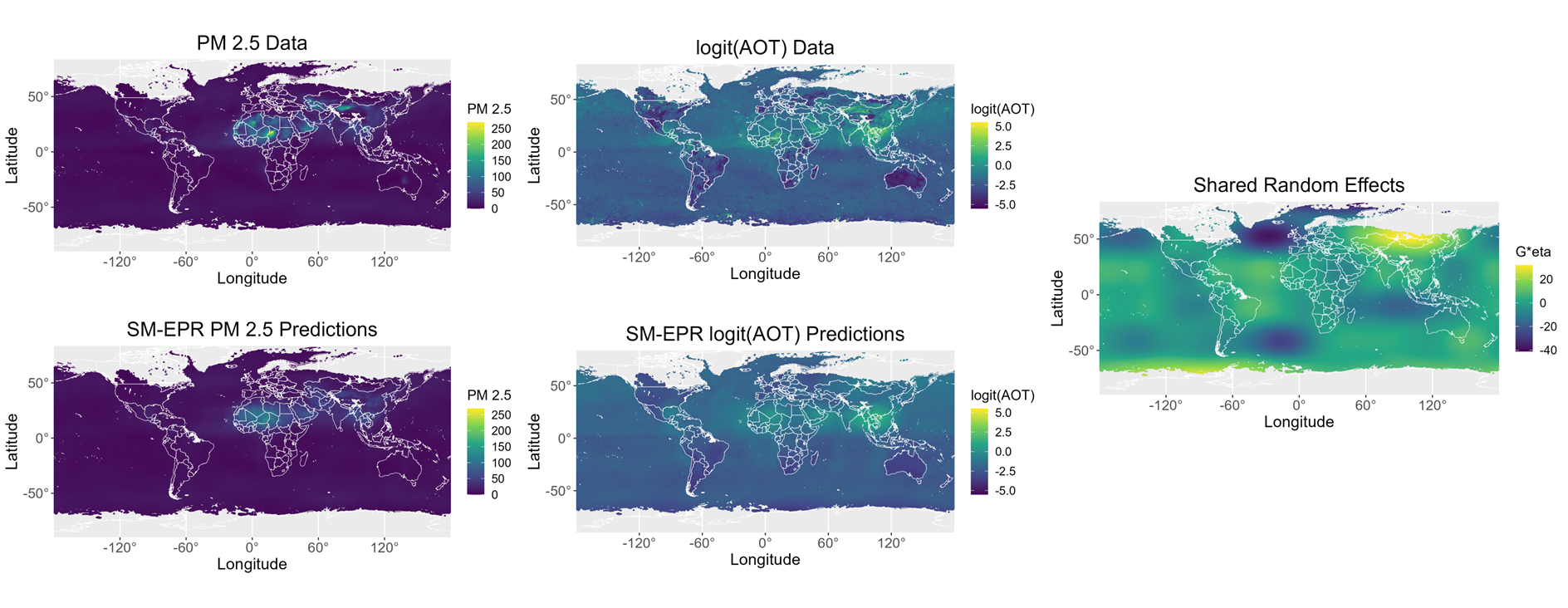}  
  \label{fig:sub-first}
\caption{Predictions of \(\text{PM}_{2.5}\) and logit(AOT) compared to the original data and plot of the posterior mean of the shared random effects \(\{\sum_{j}g_{M,j}^{*}(s_{i})\eta_{M,j}^{*}\}\) between logit(AOT) and \(\text{PM}_{2.5}\).}
\label{biv_anal_plot}
\end{figure}
\vspace{-25pt}
\section{Discussion}
\label{s:discuss}
\vspace{-15pt}
This article presents the analysis of two important climate variables and provides {six} contributions to the literature. The first main contribution is that SM-EPR allows us to simulate exact posterior replicates of the fixed and random effects in a multivariate spatial GLMM without the use of MCMC or approximate Bayesian approaches for effectively any data size. Second, our model can handle multiple response variables that may be distributed to different distributions from the exponential or DY family. A third contribution is SM-EPR includes a discrepancy term which allows us to model feedback that is known to be present in air pollution. {A fourth contribution is we present a theoretical result using KL divergence to show the data subset approach is robust to model misspecification.} A fifth contribution is the ability to jointly model AOT and \(\text{PM}_{2.5}\) on a global scale leading to improved predictions which is valuable given the challenges in obtaining global measurements for both variables. Finally, our {sixth} contribution is that the joint modeling of AOT and \(\text{PM}_{2.5}\) illustrates the possible nonlinear relationship between the two variables and elucidates important predictors for these variables. Our analysis can guide future studies that are focused on improving prediction accuracy of these climate variables which in turn can lead to more effective interventions for mitigating impacts of pollution. A natural future direction for this project involves extending this methodology to a space-time scenario, further exploring the temporal component when modeling these climate variables. 

\vspace{-20pt}
\section*{Declarations}
\vspace*{-10pt}
\subsection*{Conflicts of Interest}
\vspace{-8pt}
All authors declared that they have no conflict of interest. 

\vspace{-30pt}
\begin{singlespace}
    \bibliography{SM-EPR}

\begin{thebibliography}{}
\bibitem[Bradley et~al., 2020]{bradley2020a}
Bradley, J.~R., Wikle, C.~K., and Holan, S.~H. (2020).
\newblock Hierarchical models for spatial data with errors that are correlated with the latent process.
\newblock {\em Statistica Sinica}, 30(1):81--109.
\end{thebibliography}
\end{singlespace}

\newpage
\bigskip
\begin{center}
{\large\bf SUPPLEMENTARY MATERIAL}
\end{center}

\appendix
\renewcommand{\thesection}{\Alph{section}}

\titleformat{\section}[block]
  {\normalfont\Large\bfseries}{Appendix \thesection:}{1em}{}
\vspace{-35pt}
\section{{Notation Tables}} \label{appen:notation}
\vspace{-25pt}
\subsection{{Notation Table Section 3.1}}
\vspace{-20pt}
\begin{table}[H]
\label{param.table1}
\begin{center}
\begin{tabular}{>{\centering\arraybackslash}p{2.75cm} p{3.5cm} p{9.75cm}}
\toprule
Notation & Location of Definition in Section \ref{prelim_ss}. &  Comments \\ \midrule 
\(\mathbf{z}\), \(Z_{k,i}\) & First paragraph. & The \(NK\)-dimensional data vector and the data point at the \(i\)-th location distributed according to the \(k\)-th distribution, respectively. \(K > 1\) allows for the multi-type scenario.  \\ \midrule
\(N\) and \(K\) & First paragraph. & The total number of locations and the total number of different data types observed at each location, respectively. \\ \midrule
\(v(\cdot)\) & Second paragraph. & The true unknown non-parametric data generating pdf or pmf. \\ \midrule 
\(\delta_i\), \(\boldsymbol{\delta}\) and \(n\) & Second paragraph. & Bernoulli random variable used to include/exclude data points in training subset, the \(N\)-dimensional vector of subset indicators, and the subset size.\\ \midrule
\(\mathbf{z}_{-\delta}\) and \(\mathbf{z}_{\delta}\) & Second paragraph. & The \((N-n)K\)-dimensional vector representing the holdout data and the \(nK\)-dimensional vector representing the training data for a given \(\boldsymbol{\delta}\), respectively. \\ \midrule 
\(\mathbf{x}_{k,i}^{\prime}\boldsymbol{\beta} + \mathbf{g}_{k,i}^{\prime}\boldsymbol{\eta} + \boldsymbol{\xi}_{k,i}\) & Second paragraph. & Mixed effects model representing the natural parameter of a distribution in the exponential family. \\ \midrule 
\(\mathbf{x}_{k,i}^{\prime}\) and \(\mathbf{g}_{k,i}^{\prime}\) & Third paragraph. & The \(p\)-dimensional vector of known covariates and the \(r\)-dimensional vector of pre-specified basis functions, respectively. \\ \midrule
\(\boldsymbol{\beta}\), \(\boldsymbol{\eta}\), and \(\boldsymbol{\xi}\) & Third paragraph. & The \(p\)-dimensional regression effects vector, the \(r\)-dimensional random effects vector, and an \(NK\)-dimensional random effects vector, respectively. \\ \midrule
\(\boldsymbol{\theta}\) & Fourth paragraph. & Generic vector representing hyperparameters with proper distribution \(f(\boldsymbol{\theta})\). \\ \midrule
\(\text{EF}(Z \vert \mu, b; \psi)\) & Fourth paragraph. & A distribution from the natural exponential family. \\ \bottomrule
\end{tabular}
\end{center}
\caption{Notation for Section \ref{prelim_ss} of the main text.}
\end{table}

\subsection{{Notation Table Section 3.2}}
\begin{table}[H]
\label{param.table2}
\begin{center}
\begin{tabular}{>{\centering\arraybackslash}p{5.25cm} p{3.5cm} p{8cm}}
\toprule
Notation & Location of Definition in Section \ref{prelim_gcm}. &  Comments \\ \midrule 
\(w_{k,i}\), \(\alpha_{k,i}\), \(\kappa_{k,i}\), and \(\psi_{k,i}\) & First paragraph. & A univariate DY random variable, shape parameter,  scale parameter, and unit log partition function, respectively. \\ \midrule
\(\text{DY}(w_{k,i}, \vert \alpha_{k,i}, \kappa_{k,i}; \psi_{k,i})\) & First paragraph. & Shorthand notation for a DY distribution. \\ \midrule
\(\mathbf{h}\) & First paragraph. & A generic \(M\)-dimensional GCM random vector.  \\ \midrule 
\(\mathbf{w}\) & First paragraph. & An \(M\)-dimensional vector of independent DY random variables. \\ \midrule 
\(\boldsymbol{\mu}\) &  First paragraph. & The \(M\)-dimensional mean vector of a GCM random vector. \\ \midrule 
\(\mathbf{V}\) & First paragraph. & An \(M \times M\) invertible matrix. \\ \midrule
\(\mathbf{D}\) & First paragraph. & A \(M \times M\) matrix valued function of $\bm{\theta}$. \\ \midrule 
\(\boldsymbol{\psi}\) & First paragraph. & An \(M\)-dimensional vector valued function with elements $\psi_{k,i}(\cdot)$. \\ \midrule
\(\pi(\boldsymbol{\theta})\) & Second paragraph. & The proper hyperprior for \(\boldsymbol{\theta}\). \\ \midrule 
\(\mathscr{N}(\boldsymbol{\theta})\) & Second paragraph. & The known normalizing constant in the pdf of a GCM distribution. \\ \midrule 
\(GCM(\mathbf{h} \vert \boldsymbol{\alpha}, \boldsymbol{\kappa}, \boldsymbol{\mu}, \mathbf{V}, \pi,\mathbf{D}; \boldsymbol{\psi})\) & Second paragraph. & Shorthand notation for a GCM distribution. \\ \midrule 
\(\mathbf{h}_1\) and \(\mathbf{h}_2\) & Third paragraph. & Generic \(l\)-dimensional and \((M-l)\)-dimensional sub-vectors of GCM random vector \(\mathbf{h}\). \\ \midrule 
\(cGCM(\mathbf{h}_1 \vert \boldsymbol{\alpha}, \boldsymbol{\kappa}, \boldsymbol{\mu}^{*}, \mathbf{H}, \pi,\mathbf{D}; \boldsymbol{\psi})\) & Last paragraph. & Shorthand notation for a cGCM distribution.
\\\bottomrule
\end{tabular}
\end{center}
\caption{Notation for Section \ref{prelim_gcm} of the main text.}
\end{table}

\subsection{{Notation Table Section 3.3}}
\begin{table}[H]
\label{param.table3}
\begin{center}
\begin{tabular}{>{\centering\arraybackslash}p{4cm}  p{3.5cm} p{8.5cm}}
\toprule
Notation & Location of Definition in Section \ref{prelim_epr} &  Comments \\ \midrule 
\(\mathbf{x}_{k,i}^{\prime}\boldsymbol{\beta} + \mathbf{g}_{k,i}^{\prime}\boldsymbol{\eta} + \xi_{k,i} - \tau_{y,k,i}\) & First paragraph. & Discrepancy model. \\ \midrule 
\(N(\boldsymbol{\beta}\vert\boldsymbol{\mu}, \boldsymbol{\Sigma})\) & First Paragraph. & Shorthand notation for the multivariate normal pdf with mean \(\boldsymbol{\mu}\) and covariance \(\boldsymbol{\Sigma}\). \\ \midrule
\(\boldsymbol{\tau}_y\) & Second paragraph. & The \(N\)-dimensional discrepancy error caused by allowing for non-zero difference between the standard mixed effects model and the true latent process. \\ \midrule
\(\mathbf{X}\) and \(\mathbf{G}\) & Second paragraph. & The \(N \times p\) matrix of known covariates and the \(N \times r\) matrix of pre-specified basis functions, respectively. \\ \midrule
\(\mathbf{D}_{\beta}\boldsymbol{\tau}_{\beta}\) and  \(\boldsymbol{\Sigma}_{\beta}(\boldsymbol{\theta})\) & Second paragraph. & The mean and positive definite covariance matrix of the Gaussian prior for \(\boldsymbol{\beta}\), respectively. \\ \midrule
\(\mathbf{D}_{\eta}\boldsymbol{\tau}_{\eta}\) and  \(\boldsymbol{\Sigma}_{\eta}(\boldsymbol{\theta})\) & Second paragraph. & The mean and positive definite covariance matrix of the Gaussian prior for \(\boldsymbol{\eta}\), respectively. \\ \midrule
\(\boldsymbol{\alpha}_{\xi}\), \(\boldsymbol{\kappa}_{\xi}\), \(\boldsymbol{\tau}_{\xi}^*\), \(\mathbf{H}_{\xi}\), \(\pi_{\xi}\), \(\mathbf{D}_{\xi}\), \(\boldsymbol{\psi}_{\xi}\)  & Third paragraph. & The parameters that define the cGCM distribution for \(\boldsymbol{\xi}\). \\ \midrule
\(\boldsymbol{\tau}\) & Fourth paragraph. & The \((2N + p + r)\)-dimensional vector \(\boldsymbol{\tau} = (\boldsymbol{\tau}_y^{\prime}, \boldsymbol{\tau}_{\beta}^{\prime}, \boldsymbol{\tau}_{\eta}^{\prime}, \boldsymbol{\tau}_{\xi}^{\prime})^{\prime} = -\mathbf{D}(\boldsymbol{\theta})^{-1}\mathbf{Q}\mathbf{q}\).\\ \midrule
\(\mathbf{D}(\boldsymbol{\theta})^{-1}\) & Fourth paragraph. & The \((2N + p + r) \times (2N + p +r)\)-dimensional matrix defined as \(\text{blkdiag}(\mathbf{I}_N, \boldsymbol{D}_{\beta}(\boldsymbol{\theta})^{-1}, \mathbf{D}_{\eta}(\boldsymbol{\theta})^{-1}, \frac{1}{\sigma_{\xi}}\mathbf{I}_N)\). \\ \midrule 
\(\mathbf{q}\) & Fourth paragraph. & The unknown term referred to as the discrepancy term which is typically set equal to \(\mathbf{0}_{N,1}\) in the literature. In EPR, \(\mathbf{q}\) is given an improper prior. \\ \midrule
\(\mathbf{H}\) and \(\mathbf{Q}\) & Fourth paragraph. & Matrix valued covariance parameters in the GCM posterior distribution for \((\boldsymbol{\beta}^{\prime}, \boldsymbol{\eta}^{\prime}, \boldsymbol{\xi}^{\prime}, \mathbf{q}^{\prime})^{\prime}\). \\ \midrule
\(Y_{k,i} = \mathbf{x}_{k,i}^{\prime}\boldsymbol{\beta} + \mathbf{g}_{k,i}^{\prime}\boldsymbol{\eta}\) & Fifth paragraph. & Posterior summaries of this linear combination are used for prediction.
\\\bottomrule
\end{tabular}
\end{center}
\caption{Notation for Section \ref{prelim_epr} of the main text.}
\end{table}

\subsection{{Notation Table Section 4}}
\vspace{-20pt}
\begin{table}[H]
\label{param.table4}
\begin{center}
\begin{tabular}{>{\centering\arraybackslash}p{4cm}  p{4cm} p{8cm}}
\toprule
Notation & Location of Definition &  Comments \\ \midrule 
\(\mathbf{x}_{1,i}^{\prime}\boldsymbol{\beta} + \mathbf{g}_{1,i}^{\prime}\boldsymbol{\eta} + \xi_{1,i} - \tau_{y,1,i}\), \(\sigma_i\), \(\alpha_z\), and \(\kappa_z\) & First paragraph in Section \ref{SEPR}. & The mean, variance, and scale parameters of the logit-beta data, respectively. \\ \midrule
\(\text{exp}\{-(\mathbf{x}_{2,i}^{\prime}\boldsymbol{\beta} + \mathbf{g}_{2,i}^{\prime}\boldsymbol{\eta} + \xi_{2,i} - \tau_{y,2,i})\}\) and \(\rho_z\) & First paragraph in Section \ref{SEPR}. & The scale and shape parameter for the Weibull distributed data, respectively. \\\midrule
\(\mathbf{X}_{\delta}\) and \(\mathbf{G}_{\delta}\)& Second paragraph in Section \ref{SEPR}. & The \(2n\times p\) matrix consisting of rows of \(\mathbf{X}\) such that \(\delta_i = 1\) and the \(2n\times r\) matrix consisting of rows of \(\mathbf{G}\) such that \(\delta_i = 1\). \\ \midrule 
\(\boldsymbol{\tau}_{\delta}\) & Second paragraph in Section \ref{SEPR}. & The parameter \(\boldsymbol{\tau}_{\delta} = (\boldsymbol{\tau}_{\delta,y}^{\prime}, \boldsymbol{\tau}_{\beta}^{\prime}, \boldsymbol{\tau}_{\eta}^{\prime}, \boldsymbol{\tau}_{\delta,\xi}^{\prime})^{\prime} = \text{-blkdiag}(\mathbf{I}_{2n}, \boldsymbol{D}_{\beta}(\boldsymbol{\theta})^{-1}, \mathbf{D}_{\eta}(\boldsymbol{\theta})^{-1}, \frac{1}{\sigma_{\xi}}\mathbf{I}_{2n})\mathbf{Q}_{\delta}\mathbf{q}\).\\ \midrule
\(\boldsymbol{\tau}_{\delta, y}\) & Second paragraph in Section \ref{SEPR}. & The \(2n\)-dimensional vector function consisting of elements of \(\boldsymbol{\tau}_y\) such that \(\delta_i = 1\). \\ \midrule 
\(\boldsymbol{\tau}_{\delta, \xi}\) & Second paragraph in Section \ref{SEPR}. & The \(2n\)-dimensional vector function consisting of elements of \(\boldsymbol{\tau}_{\xi}\) such that \(\delta_i = 1\). \\ \midrule 
\(\boldsymbol{\xi}_{\delta}\) & Second paragraph in Section \ref{SEPR}. & The \(2n\)-dimensional vector consisting of elements of \(\boldsymbol{\xi}\) such that \(\delta_i = 1\). \\ \midrule 
\(\mathbf{H}_{\delta}\) and \(\mathbf{Q}_{\delta}\) & Second paragraph in Section \ref{SEPR}. & Matrix valued covariance parameters in the GCM posterior distribution for \((\boldsymbol{\beta}^{\prime}, \boldsymbol{\eta}^{\prime}, \boldsymbol{\xi}_{\delta}^{\prime}, \mathbf{q}^{\prime})^{\prime}\). \\ \midrule
\(\boldsymbol{\alpha}_M\) and \(\boldsymbol{\kappa}_M\) & Theorem \ref{theorem1} in Section \ref{technical.dev}. & The DY parameters for the GCM posterior. \\ \midrule 
Subscript ``rep'' & Theorem \ref{theorem2} in Section \ref{computation}. & Represents a single replicate from the posterior. \\ \midrule 
\(\boldsymbol{\gamma}\) & Fifth paragraph in Section \ref{stat.prop}. & Generic \(l\)-dimensional real-valued parameter vector.  \\ \midrule 
\(m_{FULL}(\mathbf{z}, \boldsymbol{\delta})\), \(m_{SUB}(\mathbf{z}, \boldsymbol{\delta})\), \(m_{TRUE}(\mathbf{z}, \boldsymbol{\delta})\) & Fifth paragraph in Section \ref{stat.prop}. & The marginal distribution of \((\textbf{z},\boldsymbol{\delta})\) from different models. \\ \midrule
\(KL\{f\vert \vert g\}\) & Theorem \ref{theorem4} in Section \ref{stat.prop}. & The KL divergence between generic models \(f\) and \(g\).
\\\bottomrule
\end{tabular}
\end{center}
\caption{Notation for Section \ref{method} of the main text.}
\end{table}
\vspace{-10pt}
\section{A General Expression of the SM-EPR}\label{appen:details}
\vspace{-10pt}
{In Section \ref{SEPR} of the main text, we provide the hierarchical model for SM-EPR for logit-beta and Weibull distributed data as they are the data types used in our application. In this appendix, we provide a general expression of the hierarchical model for SM-EPR that also allows for data distributed according to a member of the exponential family or DY family, as these data types are also of interest. The data, process, parameter, and subset models for the general hierarchical expression are defined as follows:} \vspace{-40pt}

\begin{align}
    f(\mathbf{z} &| \boldsymbol{\beta}, \boldsymbol{\eta}, \boldsymbol{\xi}, \boldsymbol{\theta}, \mathbf{q}, \boldsymbol{\delta}) = v(\mathbf{z}_{-\delta}) \prod_{k = 1}^K \prod_{\{i: \delta_i = 1\}} f_{k}(Z_{k,i}|\mathbf{x}_{k,i}^{\prime}\boldsymbol{\beta} +  \mathbf{g}_{k,i}^{\prime}\boldsymbol{\eta} +  \xi_{k,i} - \tau_{y,k,i}, b_{k,i}; \psi_{k});  \notag \\ f(\boldsymbol{\xi}|\boldsymbol{\beta},\boldsymbol{\eta},&\boldsymbol{\theta},\mathbf{q}, \boldsymbol{\delta}) = \left(\frac{1}{2 \pi \sigma_{\xi}^2}\right)^{KN/2} \text{exp}\left(-\frac{\sum_{k = 1}^K\sum_{i = 1}^N(\xi_{k,i} - \tau_{\xi,k, i})^2}{2 \sigma_{\xi}^2}\right) \notag\\
    &\times \text{exp}\left(c_{\xi}\sum_{k=1}^K\sum_{i = 1}^N (\textbf{x}_{k,i}' \boldsymbol{\beta} + \textbf{g}_{k,i}'\boldsymbol{\eta} + \xi_{k,i} - \tau_{y,k,i}) - d_{\xi} \sum_{k=1}^K\sum_{i = 1}^N \psi\left(\textbf{x}_{k,i} \boldsymbol{\beta} + \textbf{g}_{k,i}'\boldsymbol{\eta} + \xi_{k,i} - \tau_{y,k,i}\right)\right) \notag\\
    \boldsymbol{\beta}|\boldsymbol{\theta},\mathbf{q}, \boldsymbol{\delta} &\sim N(\boldsymbol{\beta}\vert\mathbf{D}_{\beta}\boldsymbol{\tau}_{\beta}(\boldsymbol{\theta}, \mathbf{q}), \mathbf{D}_{\beta}(\boldsymbol{\theta}) \mathbf{D}_{\beta}(\boldsymbol{\theta})') \notag\\
    \boldsymbol{\eta}|\boldsymbol{\theta},\mathbf{q}, \boldsymbol{\delta}  &\sim N(\boldsymbol{\eta}\vert\mathbf{D}_{\eta}\boldsymbol{\tau}_{\eta}(\boldsymbol{\theta},\mathbf{q}), \mathbf{D}_{\eta}(\boldsymbol{\theta}) \mathbf{D}_{\eta}(\boldsymbol{\theta})') \notag\\
    f(\mathbf{q}\vert \boldsymbol{\delta}) &= 1 \notag\\
    &\pi(\boldsymbol{\theta}) \notag\\
    &f(\boldsymbol{\delta}|n).
    \label{equation10}
\end{align}
{The notation \(f_{{k}}(Z\vert Y,b;\psi_{k})\) denotes the pdf or pmf of a distribution in the natural exponential family, the DY family, or Weibull distribution with parameters $b_{k,i} \in \bm{\theta}$. Table \ref{param.table} presents the parameters for the Weibull, logit-beta, Gaussian, Poisson, and binomial distributions. When \(\delta_i = 1\), we assume the natural parameter \(\mathbf{x}_{k, i}^{\prime}\boldsymbol{\beta} + \mathbf{g}_{k, i}^{\prime}\boldsymbol{\eta} + \xi_{k,i} - \tau_{y,k,i}\) is the discrepancy model from \citet{bradley2024generating}.  The terms \(\mathbf{D}_{\beta}(\boldsymbol{\theta})\), \(\mathbf{D}_{\eta}(\boldsymbol{\theta})\), \(\boldsymbol{\Sigma}_{\beta}(\boldsymbol{\theta})\), and \(\boldsymbol{\Sigma}_{\eta}(\boldsymbol{\theta})\) have the same definition as in Section \ref{prelim_epr} of the main text. The terms $v(\textbf{z}_{-\delta})$ and $f(\bm{\delta}\vert n)$ have the same definition as in Section~\ref{prelim_ss}, and it is again assumed that $\textbf{z}$ and $\bm{\delta}$ are independent.}

{
The discrepancy parameter \(\boldsymbol{\tau}_{\bm{\delta}}\), which is an \((2Kn + p + r)\)-dimensional vector that is a function of \((\boldsymbol{\theta}, \mathbf{q})\). is given by: \vspace{-20pt}
\begin{align}
    \boldsymbol{\tau}_{\bm{\delta}}(\boldsymbol{\theta}, \mathbf{q}) = \begin{pmatrix}
\boldsymbol{\tau}_{\delta, y}(\boldsymbol{\theta}, \mathbf{q})\\
\boldsymbol{\tau}_{\beta}(\boldsymbol{\theta}, \mathbf{q}) \\
 \boldsymbol{\tau}_{\eta}(\boldsymbol{\theta}, \mathbf{q}) \\
\boldsymbol{\tau}_{\delta, \xi}(\boldsymbol{\theta}, \mathbf{q})
\end{pmatrix} = -\begin{pmatrix}
 \mathbf{I}_{Kn} & \bm{0}_{Kn,p} & \bm{0}_{Kn,r} & \bm{0}_{Kn,Kn}\\
 \bm{0}_{p,Kn} &  \mathbf{D}_{\beta}(\boldsymbol{\theta})^{-1} & \bm{0}_{p,r} & \bm{0}_{p,Kn}\\
 \bm{0}_{r,Kn} & \bm{0}_{r,p} & \mathbf{D}_{\eta}(\boldsymbol{\theta})^{-1} & \bm{0}_{r,Kn} \\
 \bm{0}_{Kn,Kn} & \bm{0}_{Kn,p} &  \bm{0}_{Kn,r} & \frac{1}{\sigma_{\xi}}\bm{I}_{Kn} \\
\end{pmatrix}
 \mathbf{Q}_{\delta}\mathbf{q},
 \label{equation8.1.1}
\end{align}
}

\noindent 
where the $Kn$-dimensional vector function $\bm{\tau}_{\delta,y} = (\tau_{y,k,i}: \delta_{i} = 1, 1,\dots,K)^{\prime}$ and $\bm{\tau}_{\delta,\xi} = (\tau_{y,\xi,i}: \delta_{i} = 1, 1,\dots,K)^{\prime}$. When $\delta_{j} = 0$ we define $\tau_{y,k,j} =\tau_{y,\xi,j}= 0$. The \((2Kn + p + r) \times Kn\) matrix \(\mathbf{Q}_{\delta}\) represents the eigenvectors of the orthogonal complement of the \((2Kn + p + r) \times (Kn + p + r)\) matrix \(\mathbf{H}_{\delta}\) associated with non-zero eigenvalues, where \vspace{-15pt}
\begin{align}
    \mathbf{H}_{\delta} = \begin{pmatrix}
 \mathbf{I}_{Kn} & \mathbf{X}_{\delta} & \mathbf{G}_{\delta}\\
 \bm{0}_{p,Kn} &  \mathbf{I}_p & \bm{0}_{p,r} \\
 \bm{0}_{r,Kn} & \bm{0}_{r,p} & \mathbf{I}_r \\
 \mathbf{I}_{Kn} & \bm{0}_{Kn,p} &  \bm{0}_{Kn,r} \\
\end{pmatrix},
\label{equation9.1.1}
\end{align}

\noindent where ${K}n\times p$ matrix $\textbf{X}_{\delta} = (\textbf{x}_{k,i}^{\prime}: \delta_{i} = 1, k = {1,\ldots, K})^{\prime}$, and the ${K}n\times r$ matrix $\textbf{G}_{\delta} = (\textbf{g}_{k,i}^{\prime}: \delta_{i} = 1, k = {1,\ldots, K})^{\prime}$. Similar to the motivation in \citet{bradley2024generating} this specification of \(\boldsymbol{\tau}_{\bm{\delta}}\) given in (\ref{equation8.1.1}) is in the orthogonal column space of \((\boldsymbol{\xi}_{\delta}^{\prime},\boldsymbol{\beta}^{\prime}, \boldsymbol{\eta}^{\prime})^{\prime}\) (see Theorem \ref{theorem2} for verification), which again avoids collinearity issues between $(\boldsymbol{\xi}_{\delta}^{\prime},\boldsymbol{\beta}^{\prime}, \boldsymbol{\eta}^{\prime})^{\prime}$ and $\bm{\tau}_{\delta}$, where the $Kn$-dimensional vector $\boldsymbol{\xi}_{\delta}^{\prime} = (\xi_{k,i}: \delta_{i} = 1, k =1,\dots, K)^{\prime}$.

When $K = 2$, we have that the general statement in (\ref{equation10}) is equivalent to the SM-EPR stated in (\ref{hier.mod.eq}) in the main text. In this supplementary appendix we work with the general statement as these data models are of independent interest.

The parameter vector \(\boldsymbol{\theta}  = (\sigma_1^2, \dots, \sigma_n^2, \sigma_{\xi}^2, \sigma_{\beta}^2, \sigma_{\eta}^2, \rho_{\beta}, \rho_{\eta},\rho_{\xi},\rho_z)\) has prior distribution \(\pi(\boldsymbol{\theta}) =\) \(\pi(\sigma_{\xi}^2|\rho_{\xi}) \pi(\rho_{\xi}) \pi(\sigma_{\beta}^2|\rho_{\beta}) \pi(\rho_{\beta})  \pi(\sigma_{\eta}^2|\rho_{\eta}) \pi(\rho_{\eta})\pi(\rho_{z})\prod_{\{i:\delta_i = 1\}} \pi(\sigma_i^2)\). The matrix square roots of the prior covariance matrices for \(\boldsymbol{\beta}\) and \(\boldsymbol{\eta}\) are defined as \(\mathbf{D}_{\beta}(\boldsymbol{\theta}) \equiv \sigma_{\beta}\mathbf{I}_{p}\) and \(\mathbf{D}_{\eta}(\boldsymbol{\theta}) \equiv \sigma_{\eta}\mathbf{I}_{r}\), respectively. We assume the following priors for the components of \(\pi(\boldsymbol{\theta})\): \(\sigma_{\xi}^2|\rho_{\xi} \sim \text{IG}(1, \rho_{\xi})\),  \(\sigma_{\beta}^2|\rho_{\beta} \sim \text{IG}(1, \rho_{\beta})\),  \(\sigma_{\eta}^2|\rho_{\eta} \sim \text{IG}(1, \rho_{\eta})\), \(\rho_{\xi} \sim \text{Gamma}(1,1)\),  \(\rho_{\beta} \sim \text{Gamma}(1,1)\),  \(\rho_{\eta} \sim \text{Gamma}(1,1)\), \(\rho_z \sim \text{IG}(1,1)\) and \(\sigma_i^2 \sim \text{IG}(1, 1.5)\), where IG denotes the inverse gamma distribution.

\begin{table}[H]
\caption{SM-EPR parameters for data distributed according to logit-beta, Weibull, Gaussian, Poisson, and binomial data.}
\label{param.table}
\begin{center}
\begin{tabular}{c p{13cm}}
\toprule
\textbf{Family} & \textbf{Parameters of \(f_k\), \(f(\boldsymbol{\xi}\vert \boldsymbol{\beta}, \boldsymbol{\eta}, \boldsymbol{\theta}, \boldsymbol{q}, \boldsymbol{\delta})\), and \(f(\boldsymbol{\beta}, \boldsymbol{\eta}, \boldsymbol{\xi}, \boldsymbol{q}\vert \mathbf{z},\boldsymbol{\delta})\)} \\ \midrule 
Logit-Beta & The log-partition function \(b_{1,i}\psi_{1}\) is given by \(b_{1,i} = \kappa_z\) and \(\psi_{1}(Y_i) = \text{log}\{1 + \text{exp}(Y_i)\}\). The shape parameters are denoted with \(\alpha_z\) and \(\kappa_z\) and the mean is \(Y_i\). The constants \(c_{\xi} = 0\) and \(d_{\xi}=0\). The parameters of the GCM posterior distribution are  \(\alpha_{1,i} = -\alpha_z \sigma_i^2\) and \(\kappa_{1,i} = \kappa_z\). \\ \midrule
Weibull & The shape and scale parameters are denoted with \(b_{k,i}\equiv \rho_z\) and \(\text{exp}(-\textbf{x}_{k,i}^{\prime}\bm{\beta}-\textbf{g}_{k,i}^{\prime}\bm{\eta}-\xi_{k,i}+\tau_{y,k,i})\), respectively. Let $\psi_{2}(Y) = \text{exp}(-Y)$, where we note that the two parameter Weibull distribution is not a natural exponential family member, and hence, $b_{2,i}\psi_{2}$ is not interpreted as the log-partition function in this case. The constants \(c_{\xi} = 0\) and \(d_{\xi}= 0\). The parameters of the GCM posterior distribution are \(\alpha_{2,i} = 1\) and \(\kappa_{2,i} = Z_i^{\rho_{z}}\). \\ \midrule
Gaussian & The log-partition function \(b_{3,i}\psi_{3}\) is given by \(b_{3,i} = \frac{1}{2\sigma_i^2}\) with \(\sigma_i^2 > 0\) and \(\psi_{3}(Y) = Y^2\). The constants \(c_{\xi} = 0\) and \(d_{\xi}=0\). The parameters of the GCM posterior distribution are \(\alpha_{3,i} = \frac{Z_i}{\sigma_i^2}\) and \(\kappa_{3,i} = \frac{1}{2\sigma_i^2}\).\\ \midrule 
Poisson & The log-partition function \(b_{4,i}\psi_{4}\) is given by \(b_{4,i} = 1\) and \(\psi_{4}(Y_i) = \text{exp}(Y_i)\). The mean is \(\text{exp}(Y_i)\). The constants \(c_{\xi} = \alpha_{\xi}\) and \(d_{\xi}=0\). The parameters of the GCM posterior distribution are \(\alpha_{4,i} = Z_i + \alpha_{\xi}\) and \(\kappa_{4,i} = 1\).\\ \midrule
Binomial & The log-partition function \(b_{5,i}\psi_{5}\) is given by \(b_{5,i} = m_i\) with \(m_i \in \mathbb{Z}^{+}\) and \(\psi_{5}(Y_i) = \text{log}\{1+\text{exp}(Y_i)\}\). The sample size and probability of success are denoted with \(m_i\) and \(\text{exp}(Y_i)/\{1 + \text{exp}(Y_i)\}\), respectively. When \(m_i = 1\), \(f_{k}\) is the Bernoulli pmf, as it is a special case of the binomial distribution. The constants \(c_{\xi} = \alpha_{\xi}\) and \(d_{\xi}=2\alpha_{\xi}\). The parameters of the GCM posterior distribution are \(\alpha_{5,i} = Z_i + \alpha_{\xi}\) and \(\alpha_{5,i} = m_i + 2\alpha_{\xi}\).\\ \bottomrule
\end{tabular}
\end{center}
\end{table}

\vspace{-10pt}

\vspace{-20pt}
\section{{Technical Results}}\label{appen:tech.dev}
\renewcommand{\qedsymbol}{}

In this section we list the general statements of Theorems 1 $\--$ 3. Theorems 1$\--$3 from the main text follow immediately from General Theorems 1 $\--$ 3 when $K = 2$. Additionally, we provide a proof of Theorem 4.

\begin{namedtheorem}[General Theorem 1]
    Assume the hierarchical model in (\ref{equation10}). Then 
    \vspace{-15pt}
    \begin{align}
        f(\boldsymbol{\delta}\vert \mathbf{z}, n) &= f(\boldsymbol{\delta} \vert n) \notag\\
        (\boldsymbol{\xi}_{\delta}',\boldsymbol{\beta}', \boldsymbol{\eta}', \mathbf{q}')' \vert \mathbf{z}, \boldsymbol{\delta} &\sim \text{GCM}(\boldsymbol{\alpha}_M, \boldsymbol{\kappa}_M, \bm{0}_{2Kn + p + r, 1}, \mathbf{V}, \pi, \mathbf{D}_{\delta}; \boldsymbol{\psi}),
    \end{align}

    \vspace{-15pt}
    \noindent where $\bm{\xi}_{\delta} = (\xi_{k,i}: \delta_{i} = 1,k = 1,\ldots, K)^{\prime}$ \(\mathbf{V}^{-1} = (\mathbf{H}_{\delta}, \mathbf{Q}_{\delta})\) is defined in (\ref{equation9.1.1}), \(\boldsymbol{\psi}(\mathbf{h})\) \(= (\psi_{k}(h_1), \dots,\) \(\psi_{k}(h_{Kn}), \psi_{*}(h_1^*),\) \(\dots, \psi_{*}(h_{Kn + p + r}^*))\) is the \((2Kn + p + r)\)-dimensional unit-log partition function for \(\mathbf{h}\) \(= (h_1, \dots, h_{Kn},\) \(h_1^*, \dots, h_{Kn + p + r}^*)' \in \mathbb{R}^{2Kn + p + r}\). Let \(\psi_{k},k =1,\dots K\) denote the unit log partition function corresponding to the \(k\)-th family and \(\psi_{*}(\mathbf{h}^*) = (\mathbf{h}^{*})^2 \). The term \(\boldsymbol{\alpha}_M = (\alpha_{k,1}, \dots \alpha_{k,Kn}, \bm{0}_{1, Kn + p + r})'\) and the term \(\boldsymbol{\kappa}_M = (\kappa_{k,1}, \dots, \kappa_{k,Kn}, \frac{1}{2}\mathbf{1}_{1, Kn + p + r})'\). The terms \(\alpha_{k,i}\) and \(\kappa_{k,i}\) are the DY parameters corresponding to the \(k\)-th family and are defined in Table \ref{param.table}. 
    \label{sgcm.posterior}
\end{namedtheorem}

\begin{proof}
We assume \(\mathbf{z}\) and \(\boldsymbol{\delta}\) are independent and as a result, the expression simplifies to \(f(\boldsymbol{\delta}\vert\mathbf{z}, n) = f(\boldsymbol{\delta}\vert n)\). Thus this leaves us to derive \(f(\boldsymbol{\xi}, \boldsymbol{\beta}, \boldsymbol{\eta}, \mathbf{q}\vert \mathbf{z}, \boldsymbol{\delta})\) since \(f(\boldsymbol{\xi}, \boldsymbol{\beta}, \boldsymbol{\eta}, \mathbf{q}, \boldsymbol{\delta}\vert \mathbf{z}, n) = f(\boldsymbol{\xi}, \boldsymbol{\beta}, \boldsymbol{\eta}, \mathbf{q} \vert \mathbf{z}, \boldsymbol{\delta})f(\boldsymbol{\delta}\vert \mathbf{z}, n)\). Our strategy is to show that \(f(\boldsymbol{\xi}, \boldsymbol{\beta}, \boldsymbol{\eta}, 
\mathbf{q} |\mathbf{z}, \boldsymbol{\delta})\propto \) \(\int_{\Omega} \pi(\boldsymbol{\theta})f(\boldsymbol{\xi}, \boldsymbol{\beta}, \boldsymbol{\eta}, \mathbf{q}, \mathbf{z}, \boldsymbol{\delta} | \boldsymbol{\theta}) d \boldsymbol{\theta}\) is the GCM stated in General Theorem \ref{theorem1}. The data model can be written as: 
\begin{align*}
 f(\mathbf{z}|\boldsymbol{\xi}_{\delta}, \boldsymbol{\beta}, \boldsymbol{\eta}, \boldsymbol{\theta}, \mathbf{q}, \boldsymbol{\delta}) &= N \text{exp}\left[\mathbf{a}'\left\{\begin{pmatrix}\mathbf{I}_{{K}n} & \mathbf{X}_{\delta} & \mathbf{G}_{\delta}\end{pmatrix} \begin{pmatrix}
    \boldsymbol{\xi}_{{\delta}} \\
    \boldsymbol{\beta}\\
    \boldsymbol{\eta}
\end{pmatrix} - \boldsymbol{\tau}_{{\delta},y}\right\} \right. \notag \\ 
&\left.-\mathbf{b}'\boldsymbol{\psi}_{D}\left\{\begin{pmatrix}\mathbf{I}_{{K}n} & \mathbf{X}_{\delta} & \mathbf{G}_{\delta}\end{pmatrix}\begin{pmatrix}
    \boldsymbol{\xi}_{{\delta}} \\
    \boldsymbol{\beta}\\
    \boldsymbol{\eta}
\end{pmatrix} - \boldsymbol{\tau}_{{\delta},y}\right \} \right].
\end{align*}

\noindent The term \(N = v(\mathbf{z}_{-\delta}){\prod_{k =1}^K}\prod_{\{i:\delta_i =1\}}N_{k,i}\) and let \(a_i\) and \(b_i\) represent the \(i\)-th component of \(\mathbf{a}\) and \(\mathbf{b}\) respectively. When the \(i\)-th data point is logit-beta distributed, \(N_{1,i} = \frac{\Gamma(\kappa_i)}{\Gamma(\alpha_i)\Gamma(\kappa_i - \alpha_i)}\), \(a_i = -\alpha_z \sigma_i^2\), and \(b_i = \kappa_z\). When the \(i\)-th data point is Weibull distributed, \(N_{2,i} = \rho_z Z_{2,i}^{\rho_z - 1}\), \(a_i = 1\), and \(b_i = Z_{2,i}^{\rho_z}\). When the \(i\)-th data point is Gaussian distributed, \(N_{3,i} = \frac{\text{exp}(-Z_{3,i}^2/2\sigma_i^2)}{\sigma_i^2}\), \(a_i = \frac{Z_{3,i}}{\sigma_i^2}\), and \(b_i = \frac{1}{2\sigma_i^2}\). When the \(i\)-th data point is Poisson distributed \(N_{4,i} =  \frac{1}{Z_{4,i}!} \), \(a_i = Z_{4,i}\), and \(b_i = 1\). When the \(i\)-th data point is Binomial distributed \(N_{5,i} = \binom{m_i}{Z_{5,i}}\), \(a_i = Z_{5,i}\), and \(b_i = m_i\). The term \(\boldsymbol{\tau}_{{\delta}, y}\) is a function of \(\mathbf{q}\) and \(\boldsymbol{\theta}\), and \(\boldsymbol{\psi}_D(\cdot) = (\psi_{k,i}(\cdot): \delta_{i} = 1)^{\prime}\). The density \(f(\boldsymbol{\xi}, \boldsymbol{\beta}, \boldsymbol{\eta}, \mathbf{q}, \boldsymbol{\theta} | \mathbf{z}, \boldsymbol{\delta})\) is proportional to the product \vspace{-10pt}
\begin{align*}
f(\mathbf{z}|\boldsymbol{\xi}_{\delta}, \boldsymbol{\beta}, \boldsymbol{\eta}, \boldsymbol{\theta}, \mathbf{q}, \boldsymbol{\delta}) f(\boldsymbol{\xi}_{\delta} | \boldsymbol{\beta}, \boldsymbol{\eta}, \boldsymbol{\theta}, \mathbf{q}, \boldsymbol{\delta}) f(\boldsymbol{\beta}| \boldsymbol{\theta}, \mathbf{q},  {\boldsymbol{\delta}}) f(\boldsymbol{\eta}| \boldsymbol{\theta}, \mathbf{q}, {\boldsymbol{\delta}}) f(\mathbf{q}\vert \boldsymbol{\delta}) \pi(\boldsymbol{\theta}),
\end{align*}

\vspace{-10pt}
\noindent where \(\boldsymbol{\tau}_{\delta} = (\boldsymbol{\tau}_{{\delta},y}', \boldsymbol{\tau}_{\beta}', \boldsymbol{\tau}_{\eta}',\boldsymbol{\tau}_{{\delta},\xi}')' = -\mathbf{D}_{\delta}(\boldsymbol{\theta})^{-1}\boldsymbol{Q}_{\delta}\mathbf{q}\). {We have that the distribution of \(f(\boldsymbol{\xi}_{\delta}|\boldsymbol{\beta}, \boldsymbol{\eta}, \boldsymbol{\theta}, \mathbf{q},\boldsymbol{\delta})\) is proportional to the following cGCM,} 
\begin{align*}
    f(\boldsymbol{\xi}_{\delta}|\boldsymbol{\beta}, \boldsymbol{\eta}, \boldsymbol{\theta}, \mathbf{q},\boldsymbol{\delta}) &\propto \left(\frac{1}{2 \pi \sigma_{\xi}^2} 
 \right)^{{K}n/2} \text{exp} \left[\boldsymbol{\alpha}_{\xi}'\left\{ 
 \begin{pmatrix}
     \mathbf{I}_{{K}n} & \mathbf{X}_{\delta} & \mathbf{G}_{\delta}\\
     \frac{1}{\sigma_{\xi}}\mathbf{I}_{{K}n} & \mathbf{0}_{{K}n,p} & \mathbf{0}_{{K}n,r}
 \end{pmatrix} \begin{pmatrix}
     \boldsymbol{\xi}_{{\delta}} \\ \boldsymbol{\beta} \\ \boldsymbol{\eta} 
 \end{pmatrix} - \begin{pmatrix}
     \boldsymbol{\tau}_{{\delta}, y} \\ \boldsymbol{\tau}_{{\delta}, \xi}
 \end{pmatrix} \right\} \right.\\[1em]
 & \left. - \boldsymbol{\kappa}_{\xi}' \boldsymbol{\psi}_{D,\xi} \left\{ 
 \begin{pmatrix}
     \mathbf{I}_{{K}n} & \mathbf{X}_{\delta} & \mathbf{G}_{\delta}\\
     \frac{1}{\sigma_{\xi}}\mathbf{I}_{{K}n} & \mathbf{0}_{{K}n,p} & \mathbf{0}_{{K}n,r}
 \end{pmatrix} \begin{pmatrix}
     \boldsymbol{\xi} \\ \boldsymbol{\beta} \\ \boldsymbol{\eta} 
 \end{pmatrix} - \begin{pmatrix}
     \boldsymbol{\tau}_{{\delta},y} \\ \boldsymbol{\tau}_{{\delta},\xi}
 \end{pmatrix} \right\} \right].
\end{align*}
\noindent The term \(\boldsymbol{\alpha}_{\xi} = (\alpha_{\xi,1}, \dots \alpha_{\xi, {K}n}, \bm{0}_{1,{K}n})'\) and the term \(\boldsymbol{\kappa}_{\xi} = (\kappa_{\xi,1}, \dots, \kappa_{\xi,{K}n}, \frac{1}{2}\mathbf{1}_{1, {K}n})'\). When the \(i\)-th datum is logit-beta, Weibull, or Gaussian distributed \(\alpha_{\xi,i} = 0\) and \(\kappa_{\xi,i} = 0\). When the \(i\)-th datum is Poisson distributed \(\alpha_{\xi,i} = \alpha_{\xi}\) and \(\kappa_{\xi,i} = 0\). When the \(i\)-th datum is Binomial distributed \(\alpha_{\xi,i} = \alpha_{\xi}\) and \(\kappa_{\xi,i} = 2\alpha_{\xi}\). The terms \(\sigma_{\xi}^2\) and \(\alpha_{\xi}\) are known, and \(\boldsymbol{\tau}_{{\delta}, y}\) and \(\boldsymbol{\tau}_{{\delta}, \xi}\) are functions of \(\mathbf{q}\) and \(\boldsymbol{\theta}\). The function \(\boldsymbol{\psi}_{D, \xi}(\mathbf{h}_{D,\xi})= (\psi_{D}(\mathbf{h})', \psi_*(\mathbf{h}^*)')'\), where \(\mathbf{h}\) and \(\mathbf{h}^*\) are both \({K}n\)-dimensional vectors, \(\psi_D(h_i) = \{\psi_{k,i}(h_i): i:\delta_i =1; k = 1, \dots, K\}\), \(\psi_{*}(\mathbf{h}^*) = (\mathbf{h}^{*})^2 \), and \(\mathbf{h}_{D, \xi} = (\mathbf{h}', \mathbf{h}^{*}{'})'\) is a \(2{K}n\)-dimensional vector. Now multiplying the data model and the distribution for \(\boldsymbol{\xi}_{\delta}\) results in the following product: \vspace{-15pt}
 \begin{align*}
     f(\mathbf{z}|\boldsymbol{\xi}, \boldsymbol{\beta}, \boldsymbol{\eta}, &\boldsymbol{\theta}, \mathbf{q}, \boldsymbol{\delta})f(\boldsymbol{\xi}_{\delta}|\boldsymbol{\beta}, \boldsymbol{\eta}, \boldsymbol{\theta}, \mathbf{q},\boldsymbol{\delta}) \propto \\ &N \left(\frac{1}{2 \pi \sigma_{\xi}^2} 
 \right)^{{K}n/2} \text{exp} \left[\mathbf{a}_{Z}'\left\{ 
 \begin{pmatrix}
     \mathbf{I}_{{K}n} & \mathbf{X}_{\delta} & \mathbf{G}_{\delta}\\
     \frac{1}{\sigma_{\xi}}\mathbf{I}_{{K}n} & \mathbf{0}_{{K}n,p} & \mathbf{0}_{{K}n,r}
 \end{pmatrix} \begin{pmatrix}
     \boldsymbol{\xi}_{{\delta}} \\ \boldsymbol{\beta} \\ \boldsymbol{\eta} 
 \end{pmatrix} - \begin{pmatrix}
     \boldsymbol{\tau}_{{\delta},y} \\ \boldsymbol{\tau}_{{\delta},\xi}
 \end{pmatrix} \right\} \right.\\[1em]
 & \left. - \mathbf{b}_{Z}' \boldsymbol{\psi}_{D,\xi} \left\{ 
 \begin{pmatrix}
     \mathbf{I}_{{K}n} & \mathbf{X}_{\delta} & \mathbf{G}_{\delta}\\
     \frac{1}{\sigma_{\xi}}\mathbf{I}_{{K}n} & \mathbf{0}_{{K}n,p} & \mathbf{0}_{{K}n,r}
 \end{pmatrix} \begin{pmatrix}
     \boldsymbol{\xi}_{{\delta}} \\ \boldsymbol{\beta} \\ \boldsymbol{\eta} 
 \end{pmatrix} - \begin{pmatrix}
     \boldsymbol{\tau}_{{\delta}, y} \\ \boldsymbol{\tau}_{{\delta}, \xi}
 \end{pmatrix} \right\} \right].
 \end{align*}

\vspace{-10pt}
\noindent The term \(\mathbf{a}_{Z} = (a_{Z,1}, \dots, a_{Z,{K}n}, \mathbf{0}_{1,{K}n})^{\prime}\) and the term \(\mathbf{b}_{Z} = (b_{Z,1},\dots,  \mathbf{b}_{Z,{K}n}, \frac{1}{2}\mathbf{1}_{1,{K}n})^{\prime}\). When the \(i\)-th data point is logit-beta distributed, \(a_{Z,i} = -\alpha_z\sigma_i^2\) and \(b_{Z,i} = \kappa_z\). When the \(i\)-th data point is Weibull distributed,  \(a_{Z,i} = 1\) and  \(b_{Z,i} = Z_{2,i}^{\rho_z}\). When the \(i\)-th data point is Gaussian distributed \(a_{Z,i} = \frac{Z_{3,i}}{\sigma_i^2}\) and \(b_{Z,i} = \frac{1}{2\sigma_i^2}\). When the \(i\)-th data point is Poisson distributed \(a_{Z,i} = Z_{4,i} + \alpha_{\xi}\) and \(b_{Z,i} = 1\). When the \(i\)-th data point is binomial distributed,  \(a_{Z,i} = Z_{5,i} + \alpha_{\xi}\) and \(b_{Z,i} = m_i + 2\alpha_{\xi}\).

 Now multiplying the product of the data model and the distribution of {\(\boldsymbol{\xi}_{\delta}\)} by \\ \(f(\boldsymbol{\beta}| \boldsymbol{\theta}, \mathbf{q}, {\bm{\delta}}) f(\boldsymbol{\eta}| \boldsymbol{\theta}, \mathbf{q}, {\bm{\delta}}) f(\mathbf{q}\vert \boldsymbol{\bm{\delta}}) \pi(\boldsymbol{\theta})\) and stacking vector and matrices, results in the follow distribution:  \vspace{-20pt}
\begin{align*}
     f&({\boldsymbol{\xi}_{\delta}}, \boldsymbol{\beta}, \boldsymbol{\eta}, 
\mathbf{q} |\mathbf{z}, \boldsymbol{\delta}) \propto \\
& \frac{\pi(\boldsymbol{\theta})}{\text{det}\{\mathbf{D}_{\delta}(\boldsymbol{\theta})\}}\text{exp} \left[\boldsymbol{\alpha}_M'\left\{ 
 \begin{pmatrix}
     \mathbf{I}_{{K}n} & \mathbf{X}_{\delta} & \mathbf{G}_{\delta}\\
     \bm{0}_{p,{K}n} & \mathbf{D}_{\beta}(\boldsymbol{\theta})^{-1} & \bm{0}_{p,r} \\
     \bm{0}_{r,{K}n} & \bm{0}_{r,p} & \mathbf{D}_{\eta}(\boldsymbol{\theta})^{-1} \\
     \frac{1}{\sigma_{\xi}}\mathbf{I}_{{K}n} & \mathbf{0}_{{K}n,p} & \mathbf{0}_{{K}n,r}
 \end{pmatrix} \begin{pmatrix}
     \boldsymbol{\xi}_{{\xi}} \\ \boldsymbol{\beta} \\ \boldsymbol{\eta} 
 \end{pmatrix} - \begin{pmatrix}
      \boldsymbol{\tau}_{{\delta},y} \\ \boldsymbol{\tau}_{\beta} \\ \boldsymbol{\tau}_{\eta} \\\boldsymbol{\tau}_{{\delta}, \xi} 
 \end{pmatrix} \right\} \right.\\[1em]
 & \left. - \boldsymbol{\kappa}_M' \boldsymbol{\psi} \left\{ 
 \begin{pmatrix}
     \mathbf{I}_{{K}n} & \mathbf{X}_{\delta} & \mathbf{G}_{\delta}\\
     \bm{0}_{p,{K}n} & \mathbf{D}_{\beta}(\boldsymbol{\theta})^{-1} & \bm{0}_{p,r} \\
     \bm{0}_{r,{K}n} & \bm{0}_{r,p} & \mathbf{D}_{\eta}(\boldsymbol{\theta})^{-1} \\
     \frac{1}{\sigma_{\xi}}\mathbf{I}_{{K}n} & \mathbf{0}_{{K}n,p} & \mathbf{0}_{{K}n,r}
 \end{pmatrix} \begin{pmatrix}
     \boldsymbol{\xi}_{{\delta}} \\ \boldsymbol{\beta} \\ \boldsymbol{\eta} 
 \end{pmatrix} - \begin{pmatrix}
     \boldsymbol{\tau}_{{\delta},y} \\ \boldsymbol{\tau}_{\beta} \\ \boldsymbol{\tau}_{\eta} \\\boldsymbol{\tau}_{{\delta}, \xi} 
 \end{pmatrix} \right\} \right].
 \end{align*}
 Now, substituting \(\boldsymbol{\tau}_{{\delta}} = -\mathbf{D}_{\delta}(\boldsymbol{\theta})^{-1}\mathbf{Q}_{\delta}\mathbf{q}\), substituting \(\boldsymbol{\zeta} = ({\boldsymbol{\xi}_{\delta}'}, \boldsymbol{\beta}', \boldsymbol{\eta}')'\) and  integrating with respect to \(\boldsymbol{\theta}\) gives: 

\vspace{-40pt}
\begin{align*}
    f&({\boldsymbol{\xi}_{\delta}}, \boldsymbol{\beta}, \boldsymbol{\eta}, 
\mathbf{q} |\mathbf{z}, \boldsymbol{\delta}) \propto \\
&\int_{\Omega} \frac{\pi(\boldsymbol{\theta})}{\text{det}\{\mathbf{D}_{\delta}(\boldsymbol{\theta})\}} \text{exp}\left[\boldsymbol{\alpha}_M' \left\{ \mathbf{D}_{\delta}(\boldsymbol{\theta})^{-1}\begin{pmatrix}
    \mathbf{H}_{\delta} & \mathbf{Q}_{\delta} 
\end{pmatrix} \begin{pmatrix}
    \boldsymbol{\zeta} \\ \mathbf{q}
\end{pmatrix} \right \} - \boldsymbol{\kappa}_M'\boldsymbol{\psi} \left\{\mathbf{D}_{\delta}(\boldsymbol{\theta})^{-1}\begin{pmatrix}
    \mathbf{H}_{\delta} & \mathbf{Q}_{\delta} 
\end{pmatrix} \begin{pmatrix}
    \boldsymbol{\zeta} \\ \mathbf{q}
\end{pmatrix}  \right\}  \right] d \boldsymbol{\theta} \\
&\propto \text{GCM}(\boldsymbol{\alpha}_M, \boldsymbol{\kappa}_M, \bm{0}_{2{K}n + p + r, 1}, \mathbf{V}, \pi, \mathbf{D}_{\delta}; \boldsymbol{\psi}),
\end{align*}

\vspace{-20pt}
\noindent as defined in Section \ref{prelim_gcm} of the main text. The term \(\boldsymbol{\alpha}_M = (\alpha_{k,1}, \dots \alpha_{k,{K}n}, \bm{0}_{1, {K}n + p + r})'\) and the term \(\boldsymbol{\kappa}_M = (\kappa_{k,1}, \dots, \kappa_{k,{K}n}, \frac{1}{2}\mathbf{1}_{1, {K}n + p + r})'\). The terms \(\alpha_{k,i}\) and \(\kappa_{k,i}\) are defined in Table~\ref{param.table}. This completes the proof.
\end{proof} 


\begin{namedtheorem}[General Theorem 2]
    Replicates of $\mathbf{q}$, \(\boldsymbol{\delta}\), \(\boldsymbol{\beta}\), \(\boldsymbol{\eta}\), and \(\boldsymbol{\xi}_{\delta} = (\xi_{k,i}: \delta_{i} = 1,k = 1,\ldots, K)^{\prime}\) from \(f(\boldsymbol{\xi}_{\delta}, \boldsymbol{\beta}, \boldsymbol{\eta}, \mathbf{q}, \boldsymbol{\delta}|\mathbf{z})\) from {General Theorem 1} have the following property. \vspace{-20pt}
\begin{align}
\boldsymbol{\delta}_{rep} &\sim f(\boldsymbol{\delta}\vert n) \notag \\
 (\boldsymbol{\xi}_{\delta,rep}^{\prime}, \boldsymbol{\beta}_{rep}^{\prime}, \boldsymbol{\eta}_{rep}^{\prime} )^{\prime} &= (\mathbf{H}_{\delta}'\mathbf{H}_{\delta})^{-1}\mathbf{H}_{\delta}'\mathbf{w}_{rep}, \notag \\
    \mathbf{q}_{rep} &= \mathbf{Q}_{\delta}'\mathbf{w}_{rep},\label{equation11gen}
\end{align}
where the subscript ``rep'' represents a single replicate from the posterior distribution, \(\mathbf{w}_{rep} \equiv (\mathbf{y}_{\delta, rep}^{\prime}, \mathbf{w}_{\beta}^{\prime}, \mathbf{w}_{\eta}^{\prime}, \mathbf{w}_{\xi}^{\prime})^{\prime} \), \(\mathbf{y}_{\delta,rep}\) consists of independent DY random variables with \(i\)-th element corresponding to type \(k\) and \(\mathbf{H}_{\delta}\) is defined in (\ref{equation9.1.1}). The terms \(\mathbf{w}_{\beta}\), \(\mathbf{w}_{\eta}\), and \(\mathbf{w}_{\xi}\) are obtained by first sampling $\bm{\theta}^{*}$ from its respective prior distribution and then \(\mathbf{w}_{\xi}\) is sampled from a mean zero normal distribution with covariance $\sigma_{\xi}^{2*}\textbf{I}_{2Kn}(\bm{\theta}^{*})$ for $\sigma_{\xi}^{2*}\in \bm{\theta}^{*}$, \(\mathbf{w}_{\beta}\) is sampled from a mean zero normal distribution with covariance $\textbf{D}_{\beta}(\bm{\theta}^{*})\textbf{D}_{\beta}(\bm{\theta}^{*})^{\prime}$, and \(\mathbf{w}_{\eta}\) is sampled from a mean zero normal distribution with covariance $\textbf{D}_{\eta}(\bm{\theta}^{*})\textbf{D}_{\eta}(\bm{\theta}^{*})^{\prime}$. 
\label{sgcm.sampling}
\end{namedtheorem}
\begin{proof}
From Theorem 2.1 of \citep{bradley2024generating}, a GCM random vector is simulated via the transformation: \vspace{-10pt}
\[{\mathbf{h}} = \boldsymbol{\mu} + \boldsymbol{V}\mathbf{D}_{\delta}(\boldsymbol{\theta})\mathbf{w},\]

\vspace{-10pt}
\noindent where \(\mathbf{w}\) consists of DY random variables, \(\mu = \mathbf{0}_{2{K}n+p+r, 1}\), \(\mathbf{D}_{\delta}(\boldsymbol{\theta})\) is the block diagonal matrix defined in Theorem \ref{theorem1}, \({\mathbf{w}_{rep}} = \mathbf{D}_{\delta}(\boldsymbol{\theta})\mathbf{w}\), and {\(\boldsymbol{\alpha}_M\) and \(\boldsymbol{\kappa}_M\)} are the shape parameters defined in Theorem \ref{theorem1}. Recall Theorem \ref{theorem1} defines \(\mathbf{V}^{-1} = (\mathbf{H}_{\delta}, \mathbf{Q}_{\delta})\) and it is clear to see the term \(\mathbf{V}\) is given by: \vspace{-15pt}
\begin{align*}
 \mathbf{V} = (\mathbf{H}_{\delta}, \mathbf{Q}_\delta)^{-1}=
     \begin{pmatrix}
        (\mathbf{H}_{\delta}'\mathbf{H}_{\delta})^{-1}\mathbf{H}_{\delta}^{\prime} \\ \mathbf{Q}_{\delta}^{\prime}
    \end{pmatrix}.
\end{align*}

\vspace{-10pt}
\noindent Now, from the transformation, the posterior replicates of \(\boldsymbol{\xi}, \boldsymbol{\beta},\boldsymbol{\eta}\) have the form: 
\begin{align*}
    \begin{pmatrix}
        \boldsymbol{\xi}_{rep} \\
        \boldsymbol{\beta}_{rep} \\
        \boldsymbol{\eta}_{rep}
    \end{pmatrix} &= (\mathbf{H}_{\delta}'\mathbf{H}_{\delta})^{-1}\mathbf{H}_{\delta}^{\prime}{\mathbf{w}_{rep}},
\end{align*}
and $\textbf{q}_{rep} = \textbf{Q}_{\delta}^{\prime}{\textbf{w}_{rep}}$ with \({\mathbf{w}_{rep}} = (\mathbf{y}_{\delta, rep}^{\prime}, \mathbf{w}_{\beta}^{\prime}, \mathbf{w}_{\eta}^{\prime}, \mathbf{w}_{\xi}^{\prime})\).
\end{proof} 

\begin{namedtheorem}[General Theorem 3]
Let $\boldsymbol{\xi}_{\delta}$, $\boldsymbol{\beta}$, $\boldsymbol{\eta}$, $\textbf{q}$, and $\boldsymbol{\delta}$ follow the GCM model stated in {General Theorem 1}. Let $\textbf{y}_{\delta} = \textbf{X}_{\delta}\bm{\beta} + \textbf{G}_{\delta}\bm{\eta}+\bm{\xi}_{\delta}$, and let $\textbf{w}_{rep}\equiv (\textbf{y}_{\delta,rep}^{\prime},\textbf{w}_{\beta}^{\prime},\textbf{w}_{\eta}^{\prime},\textbf{w}_{\xi}^{\prime})^{\prime}$ as defined in {General Theorem 2}. Then, \vspace{-20pt}
\begin{equation*}
cov(\textbf{y}_{\delta},\boldsymbol{\tau}_{y}\vert \textbf{z},\boldsymbol{\delta}) = -\textbf{J}\textbf{H}_{\delta}(\textbf{H}_{\delta}^{\prime}\textbf{H}_{\delta})^{-1}\textbf{H}_{\delta}^{\prime}\text{cov}(\textbf{w}_{rep}\vert {\boldsymbol{\alpha}_M,\boldsymbol{\kappa}_M})\left\lbrace\textbf{I}_{2{K}n+p+r} - \textbf{H}_{\delta}(\textbf{H}_{\delta}^{\prime}\textbf{H}_{\delta})^{-1}\textbf{H}_{\delta}^{\prime}\right\rbrace \textbf{J}^{\prime},
\end{equation*}

\vspace{-15pt}
\noindent
where $\textbf{J} = (\textbf{I}_{{K}n},\bm{0}_{{K}n,p},\bm{0}_{Kn,r},\bm{0}_{{K}n,{K}n})$ and $\boldsymbol{\alpha}_M$ and $\boldsymbol{\kappa}_M$ are the same as those defined in {General Theorem 1}.
\label{sgcm.crosscov}
\end{namedtheorem}
\begin{proof}
Let $\boldsymbol{\zeta} = (\boldsymbol{\xi}_{{\delta}}^{\prime},\boldsymbol{\beta}^{\prime},\boldsymbol{\eta}^{\prime})^{\prime}$. We have that
\vspace{-20pt}
\begin{align*}
    \mathbf{y} &= \mathbf{J}\mathbf{H}_{\delta}\boldsymbol{\zeta}\\
    \boldsymbol{\tau}_{y}&=-\mathbf{J}\mathbf{Q}_{\delta}\mathbf{q}.
\end{align*}

\vspace{-15pt}
\noindent Then from General Theorem 2, and noting that \(\mathbf{Q}_{\delta}\mathbf{Q}_{\delta}^{\prime} = \mathbf{I}_{2{K}n+p+r} - \mathbf{H}_{\delta}(\mathbf{H}_{\delta}^{\prime}\mathbf{H}_{\delta})^{-1}\mathbf{H}_{\delta}^{\prime}\), we have the result.
\end{proof}

\noindent
\textbf{Proof of Theorem 4 stated in the main text:}\\
We start by showing that $m_{SUB}(\textbf{z},\bm{\delta}) = v(\textbf{z}_{-\delta})m_{FULL}(\textbf{z}_{\delta},\bm{\delta})$. We have $\prod_{k = 1}^{K}\prod_{i = 1}^{N}f(Z_{k,i}\vert \bm{\gamma})$ can be factorized as 
\begin{align}\label{eq:factorizes}
   \prod_{k = 1}^{K}\prod_{i = 1}^{N}f(Z_{k,i}\vert \bm{\gamma}) = \prod_{k = 1}^{K}\prod_{\{i: \delta_{i} = 1\}}f(Z_{k,i}\vert \bm{\gamma}) \prod_{k = 1}^{K}\prod_{\{i: \delta_{i} = 0\}}f(Z_{k,i}\vert \bm{\gamma}) ,
\end{align}
for any given $\bm{\delta}$. Thus, from Equation (\ref{eq:factorizes}) we have that,
\begin{align*}
    & m_{SUB}(\textbf{z},\bm{\delta})= v(\textbf{z}_{-\delta})\int  \prod_{k = 1}^{K}\prod_{\{i: \delta_{i} = 1\}}f(Z_{k,i}\vert \bm{\gamma})f(\bm{\gamma}\vert \bm{\delta})d\bm{\gamma} f(\bm{\delta}\vert n)\\
    &=v(\textbf{z}_{-\delta})\int \prod_{k = 1}^{K}\prod_{\{i: \delta_{i} = 0\}}f(Z_{k,i}\vert \bm{\gamma})d\textbf{z}_{-\delta}\int  \prod_{k = 1}^{K}\prod_{\{i: \delta_{i} = 1\}}f(Z_{k,i}\vert \bm{\gamma})f(\bm{\gamma}\vert \bm{\delta})d\bm{\gamma} f(\bm{\delta}\vert n)\\
    & = v(\textbf{z}_{-\delta})\int\int  \prod_{k = 1}^{K}\prod_{i = 1}^{N}f(Z_{k,i}\vert \bm{\gamma})f(\bm{\gamma}\vert \bm{\delta})d\bm{\gamma}d\textbf{z}_{-\delta} f(\bm{\delta}\vert n) \\
    & = v(\textbf{z}_{-\delta}) \int m_{FULL}(\textbf{z},\bm{\delta}) d\textbf{z}_{-\delta}\\
    & = v(\textbf{z}_{-\delta}) m_{FULL}(\textbf{z}_{\delta},\bm{\delta}),
\end{align*}
\noindent
which follows from the fact that $\int \prod_{k = 1}^{K}\prod_{\{i: \delta_{i} = 0\}}f(Z_{k,i}\vert \bm{\gamma})d\textbf{z}_{-\delta}=1$ and (\ref{eq:factorizes}).

We now prove Theorem 4(a). Let 
\begin{equation*}
    m(\textbf{z}\vert\bm{\delta}) = \int \prod_{k = 1}^{K}\prod_{i = 1}^{N}f(Z_{k,i}\vert \bm{\gamma})f(\bm{\gamma}\vert \bm{\delta})d\bm{\gamma},
\end{equation*}
\noindent so that $m_{FULL}(\textbf{z},\bm{\delta}) = m(\textbf{z}\vert\bm{\delta})f(\bm{\delta}\vert n)$ and $m_{SUB}(\textbf{z},\bm{\delta}) = v(\textbf{z}_{-\delta})m(\textbf{z}_{\delta}\vert \bm{\delta})f(\bm{\delta}\vert n)$ for $m(\textbf{z}_{\delta}\vert \bm{\delta}) = \int m(\textbf{z}\vert \bm{\delta}) d\textbf{z}_{-\delta}$. We have that
\begin{align*}
&KL\{m_{TRUE}(\textbf{z},\bm{\delta}) \hspace{2pt}||\hspace{2pt}m_{SUB}(\textbf{z},\bm{\delta})\}\\
&= \sum_{\delta}\int v(\textbf{z})f(\bm{\delta}\vert n) \text{log}\left(\frac{v(\textbf{z})}{v(\textbf{z}_{-\delta})m(\textbf{z}_{\delta}\vert \bm{\delta})}\right)d\textbf{z}\\
&= \sum_{\delta}\int v(\textbf{z})f(\bm{\delta}\vert n) \text{log}\left(\frac{v(\textbf{z}_{-\delta})v(\textbf{z}_{\delta}\vert \textbf{z}_{-\delta})}{v(\textbf{z}_{-\delta})m(\textbf{z}_{\delta}\vert \bm{\delta})}\right)d\textbf{z}\\
& = \sum_{\delta}\int v(\textbf{z})f(\bm{\delta}\vert n) \text{log}\left(\frac{v(\textbf{z}_{\delta}\vert \textbf{z}_{-\delta})}{m(\textbf{z}_{\delta}\vert \bm{\delta})}\right)d\textbf{z},
\end{align*}
\noindent where $v(\textbf{z}_{\delta}\vert \textbf{z}_{-\delta}) = v(\textbf{z})/v(\textbf{z}_{-\delta})$. We also have that,
\begin{align*}
&KL\{m_{TRUE}(\textbf{z},\bm{\delta}) \hspace{2pt}||\hspace{2pt}m_{FULL}(\textbf{z},\bm{\delta})\}\\
&= \sum_{\delta}\int v(\textbf{z})f(\bm{\delta}\vert n) \text{log}\left(\frac{v(\textbf{z})}{m(\textbf{z}\vert \bm{\delta})}\right)d\textbf{z}\\
&= \sum_{\delta}\int v(\textbf{z})f(\bm{\delta}\vert n) \text{log}\left(\frac{v(\textbf{z}_{-\delta})v(\textbf{z}_{\delta}\vert \textbf{z}_{-\delta})}{m(\textbf{z}_{\delta}\vert \bm{\delta})m(\textbf{z}_{-\delta}\vert \textbf{z}_{\delta},\bm{\delta})}\right)d\textbf{z}\\
& = \sum_{\delta}\int v(\textbf{z})f(\bm{\delta}\vert n) \text{log}\left(\frac{v(\textbf{z}_{-\delta})}{m(\textbf{z}_{-\delta}\vert \textbf{z}_{\delta},\bm{\delta})}\right)d\textbf{z}+ \sum_{\delta}\int v(\textbf{z})f(\bm{\delta}\vert n) \text{log}\left(\frac{v(\textbf{z}_{\delta}\vert \textbf{z}_{-\delta})}{m(\textbf{z}_{\delta}\vert \bm{\delta})}\right)d\textbf{z}\\
&=\sum_{\delta}\int v(\textbf{z})f(\bm{\delta}\vert n) \text{log}\left(\frac{v(\textbf{z}_{-\delta})m(\textbf{z}_{\delta}\vert \bm{\delta})}{m(\textbf{z}\vert \bm{\delta})}\right)d\textbf{z}+KL\{m_{TRUE}(\textbf{z},\bm{\delta}) \hspace{2pt}||\hspace{2pt}m_{SUB}(\textbf{z},\bm{\delta})\},
\end{align*}
\noindent
where $m(\textbf{z}_{-\delta}\vert \textbf{z}_{\delta},\bm{\delta}) = m(\textbf{z}\vert \bm{\delta})/m(\textbf{z}_{\delta}\vert \bm{\delta})$. It follows that \\
$KL\{m_{TRUE}(\textbf{z},\bm{\delta}) \hspace{2pt}||\hspace{2pt}m_{FULL}(\textbf{z},\bm{\delta})\}\ge KL\{m_{TRUE}(\textbf{z},\bm{\delta}) \hspace{2pt}||\hspace{2pt}m_{SUB}(\textbf{z},\bm{\delta})\}$ provided that
\begin{equation}
\label{eq:difficultpiece}
\sum_{\delta}\int v(\textbf{z})f(\bm{\delta}\vert n) \text{log}\left(\frac{v(\textbf{z}_{-\delta})m(\textbf{z}_{\delta}\vert \bm{\delta})}{m(\textbf{z}\vert \bm{\delta})}\right)d\textbf{z}\ge 0.
\end{equation}
\noindent
For a given $\bm{\delta}$, let ${A}_{\delta} = \{\textbf{z}:v(\textbf{z})\ge v(\textbf{z}_{-\delta})m(\textbf{z}_{\delta}\vert \bm{\delta})\}$ with set complement ${A}_{\delta}^{c} = \left\lbrace\textbf{z}:\frac{v(\textbf{z}_{-\delta})m(\textbf{z}_{\delta}\vert \bm{\delta})}{v(\textbf{z})}> 1\right\rbrace$, where the probability of these sets are denoted $P(A_{\delta})$ and $P(A_{\delta}^{c})$, respectively. Also, let $I(\textbf{z} \in A_{\delta})$ equal 1 when $\textbf{z} \in A_{\delta}$ and zero otherwise. Assume $P(A_{\delta}) \in (0,1)$. Writing (\ref{eq:difficultpiece}) as the sum of two indefinite integrals,
\begin{align*}
&\sum_{\delta}\int v(\textbf{z})f(\bm{\delta}\vert n) \text{log}\left(\frac{v(\textbf{z}_{-\delta})m(\textbf{z}_{\delta}\vert\bm{\delta})}{m(\textbf{z}\vert\bm{\delta})}\right)d\textbf{z}\\
&=\sum_{\delta}\int_{A_{\delta}} v(\textbf{z})f(\bm{\delta}\vert n) \text{log}\left(\frac{v(\textbf{z}_{-\delta})m(\textbf{z}_{\delta}\vert\bm{\delta})}{m(\textbf{z}\vert\bm{\delta})}\right)d\textbf{z}+\sum_{\delta}\int_{A_{\delta}^{c}} v(\textbf{z})f(\bm{\delta}\vert n) \text{log}\left(\frac{v(\textbf{z}_{-\delta})m(\textbf{z}_{\delta}\vert\bm{\delta})}{m(\textbf{z}\vert\bm{\delta})}\right)d\textbf{z}\\
&=\sum_{\delta}P(A_{\delta})\int \frac{v(\textbf{z})I(\textbf{z}\in A_{\delta})}{P(A_{\delta})}f(\bm{\delta}\vert n) \text{log}\left(\frac{v(\textbf{z}_{-\delta})m(\textbf{z}_{\delta}\vert\bm{\delta})}{m(\textbf{z}\vert\bm{\delta})}\right)d\textbf{z}\\
&\hspace{50pt}+\sum_{\delta}\int_{A_{\delta}^{c}} v(\textbf{z})f(\bm{\delta}\vert n) \text{log}\left(\frac{v(\textbf{z}_{-\delta})m(\textbf{z}_{\delta}\vert\bm{\delta})}{m(\textbf{z}\vert\bm{\delta})}\right)d\textbf{z}\\
&\ge \sum_{\delta}P(A_{\delta})\int \frac{v(\textbf{z}_{-\delta})m(\textbf{z}_{\delta}\vert\bm{\delta})I(\textbf{z}\in A_{\delta})}{P(A_{\delta})}f(\bm{\delta}\vert n) \text{log}\left(\frac{v(\textbf{z}_{-\delta})m(\textbf{z}_{\delta}\vert\bm{\delta})}{m(\textbf{z}\vert\bm{\delta})}\right)d\textbf{z}\\
&\hspace{50pt}+\sum_{\delta}\int_{A_{\delta}^{c}} v(\textbf{z})f(\bm{\delta}\vert n) \text{log}\left(\frac{v(\textbf{z}_{-\delta})m(\textbf{z}_{\delta}\vert\bm{\delta})}{m(\textbf{z}\vert\bm{\delta})}\right)d\textbf{z},
\end{align*}
\noindent where the inequality holds by definition of $A_{\delta}$. It follows that 
\begin{align}
\nonumber
&\sum_{\delta}P(A_{\delta})\int \frac{v(\textbf{z}_{-\delta})m(\textbf{z}_{\delta}\vert\bm{\delta})I(\textbf{z}\in A_{\delta})}{P(A_{\delta})}f(\bm{\delta}\vert n) \text{log}\left(\frac{v(\textbf{z}_{-\delta})m(\textbf{z}_{\delta}\vert\bm{\delta})}{m(\textbf{z}\vert\bm{\delta})}\right)d\textbf{z}\\
\nonumber
&\hspace{50pt}+\sum_{\delta}\int_{A_{\delta}^{c}} v(\textbf{z})f(\bm{\delta}\vert n) \text{log}\left(\frac{v(\textbf{z}_{-\delta})m(\textbf{z}_{\delta}\vert\bm{\delta})}{m(\textbf{z}\vert\bm{\delta})}\right)d\textbf{z}\\
\label{eq:onecase}
&=\sum_{\delta}P(A_{\delta})f(\bm{\delta}\vert n)KL\left\{\frac{v(\textbf{z}_{-\delta})m(\textbf{z}_{\delta}\vert\bm{\delta})I(\textbf{z}\in A_{\delta})}{P(A_{\delta})}\hspace{2pt}||\hspace{2pt}\frac{m(\textbf{z}\vert\bm{\delta})I(\textbf{z}\in A_{\delta})}{P(A_{\delta})}\right\}\\
\nonumber
&\hspace{50pt}+\sum_{\delta}\int_{A_{\delta}^{c}} v(\textbf{z})f(\bm{\delta}\vert n) \text{log}\left(\frac{v(\textbf{z}_{-\delta})m(\textbf{z}_{\delta}\vert\bm{\delta})}{m(\textbf{z}\vert\bm{\delta})}\right)d\textbf{z},
\end{align}
and since the KL is bounded below by zero $\sum_{\delta}P(A_{\delta})f(\bm{\delta}\vert n)KL\left\{\frac{v(\textbf{z}_{-\delta})m(\textbf{z}_{\delta}\vert\bm{\delta})I(\textbf{z}\in A_{\delta})}{P(A_{\delta})}\hspace{2pt}||\hspace{2pt}\frac{m(\textbf{z}\vert\bm{\delta})I(\textbf{z}\in A_{\delta})}{P(A_{\delta})}\right\}\ge 0$. We are left to show that $\sum_{\delta}\int_{A_{\delta}^{c}} v(\textbf{z})f(\bm{\delta}\vert n) \text{log}\left(\frac{v(\textbf{z}_{-\delta})m(\textbf{z}_{\delta}\vert\bm{\delta})}{m(\textbf{z}\vert\bm{\delta})}\right)d\textbf{z}\ge 0$. Then,
\begin{align}
\nonumber
&\sum_{\delta}\int_{A_{\delta}^{c}} v(\textbf{z})f(\bm{\delta}\vert n) \text{log}\left(\frac{v(\textbf{z}_{-\delta})m(\textbf{z}_{\delta}\vert\bm{\delta})}{m(\textbf{z}\vert\bm{\delta})}\right)d\textbf{z}= \sum_{\delta}\int_{A_{\delta}^{c}} v(\textbf{z})f(\bm{\delta}\vert n) \text{log}\left(\frac{v(\textbf{z}_{-\delta})m(\textbf{z}_{\delta}\vert\bm{\delta})/v(\textbf{z})}{m(\textbf{z}\vert\bm{\delta})/v(\textbf{z})}\right)d\textbf{z}\\
\nonumber
&=\sum_{\delta}\int_{A_{\delta}^{c}} v(\textbf{z})f(\bm{\delta}\vert n) \text{log}\left(\frac{v(\textbf{z}_{-\delta})m(\textbf{z}_{\delta}\vert\bm{\delta})}{v(\textbf{z})}\right)d\textbf{z} + \sum_{\delta}\int_{A_{\delta}^{c}} v(\textbf{z})f(\bm{\delta}\vert n) \text{log}\left(\frac{v(\textbf{z})}{m(\textbf{z}\vert\bm{\delta})}\right)d\textbf{z}\\
\label{eq:zerocase}
&\ge \sum_{\delta}\int_{A_{\delta}^{c}} v(\textbf{z})f(\bm{\delta}\vert n) \text{log}\left(\frac{v(\textbf{z}_{-\delta})m(\textbf{z}_{\delta}\vert\bm{\delta})}{v(\textbf{z})}\right)d\textbf{z} + \sum_{\delta}P(A_{\delta}^{c})f(\bm{\delta}\vert n)KL\{\frac{v(\textbf{z})I(\textbf{z}\in A_{\delta}^{c})}{P(A_{\delta}^{c})}\hspace{2pt}||\hspace{2pt}\frac{m(\textbf{z}\vert\bm{\delta})I(\textbf{z}\in A_{\delta}^{c})}{P(A_{\delta}^{c})}\},
\end{align}
where for every $\textbf{z}\in A_{\delta}^{c}$, we have $\text{log}\left(\frac{v(\textbf{z}_{-\delta})m(\textbf{z}_{\delta}\vert\bm{\delta})}{v(\textbf{z})}\right)>0$ and hence the first term on the right-hand-side of the inequality is strictly positive as the integrand is strictly positive. The second term is also greater than or equal to zero, since $KL\{\frac{v(\textbf{z})I(\textbf{z}\in A_{\delta}^{c})}{P(A_{\delta}^{c})}\hspace{2pt}||\hspace{2pt}\frac{m(\textbf{z}\vert\bm{\delta})I(\textbf{z}\in A_{\delta}^{c})}{P(A_{\delta}^{c})}\}\ge 0$. 

When $P(A_{\delta}) = 0$, we have that 
\begin{equation*}
\sum_{\delta}\int v(\textbf{z})f(\bm{\delta}\vert n) \text{log}\left(\frac{v(\textbf{z}_{-\delta})m(\textbf{z}_{\delta}\vert\bm{\delta})}{m(\textbf{z}\vert\bm{\delta})}\right)d\textbf{z}=\sum_{\delta}\int_{A_{\delta}^{c}} v(\textbf{z})f(\bm{\delta}\vert n) \text{log}\left(\frac{v(\textbf{z}_{-\delta})m(\textbf{z}_{\delta}\vert\bm{\delta})}{m(\textbf{z}\vert\bm{\delta})}\right)d\textbf{z},
\end{equation*}
and the result follows from (\ref{eq:zerocase}). When  $P(A_{\delta}) = 1$, similar to (\ref{eq:onecase}),
\begin{align*}
&\sum_{\delta}\int v(\textbf{z})f(\bm{\delta}\vert n) \text{log}\left(\frac{v(\textbf{z}_{-\delta})m(\textbf{z}_{\delta}\vert\bm{\delta})}{m(\textbf{z}\vert\bm{\delta})}\right)d\textbf{z}=\sum_{\delta}\int_{A_{\delta}} v(\textbf{z})f(\bm{\delta}\vert n) \text{log}\left(\frac{v(\textbf{z}_{-\delta})m(\textbf{z}_{\delta}\vert\bm{\delta})}{m(\textbf{z}\vert\bm{\delta})}\right)d\textbf{z}\\
&\ge\sum_{\delta}P(A_{\delta})f(\bm{\delta}\vert n)KL\{\frac{v(\textbf{z}_{-\delta})m(\textbf{z}_{\delta}\vert\bm{\delta})I(\textbf{z}\in A_{\delta})}{P(A_{\delta})}\hspace{2pt}||\hspace{2pt}\frac{m(\textbf{z}\vert\bm{\delta})I(\textbf{z}\in A_{\delta})}{P(A_{\delta})}\},
\end{align*}
\noindent and the result again follows from the fact that KL divergence is non-negative.
\vspace{-20pt}
\section{SM-EPR Algorithm}\label{appen:algo}
\noindent One can obtain independent replicates from the posterior distribution of \(\boldsymbol{\xi}, \boldsymbol{\beta}, \boldsymbol{\eta}, \boldsymbol{\delta} \vert \mathbf{z}\) using General Theorem 2 and a composite sampler. 

\begin{algorithm}[H]
\caption{}
\begin{algorithmic}[1] 
\setcounter{ALG@line}{-1}
\State Set \(t = 1\) and sample \(\bm{\delta}^{[t]} = \left(\delta_1^{[t]}, \dots, \delta_N^{[t]}\right)'\) from \(f(\bm{\delta} | n)\).
\State  Sample \(\mathbf{y}_{\delta, rep}^{[t]} = \left(y_{\delta, rep,1}^{[t]}, \dots y_{\delta, rep, n}^{[t]}\right)'\).
    \begin{itemize}
    \item When \(k=1\), \(\alpha_{y} = 1\) and \(\kappa_{y} = 2\alpha_{y}\). Simulate the \(n\)-dimensional vector \(\mathbf{r}_{y}^{[t]}\) from \(\text{Beta}(\alpha_{y}, \kappa_{y}- \alpha_{y})\) and \(\boldsymbol{\sigma}^{2*[t]}\) from \(\pi(\boldsymbol{\sigma}^{2})\). Then compute \(\mathbf{y}_{\delta, rep}^{[t]} = \mathbf{z}_{\delta} + \boldsymbol{\sigma}^* \text{log}\left(\frac{\mathbf{r}_{y}^{[t]}}{\bm{1}_{n} -\mathbf{r}_{y}^{[t]}}\right)\).
    \item When \(k=2\), \(\alpha_y = 1\) and \(\kappa_y = \mathbf{z}_{\delta}^{\rho_z[t]}\). Simulate \(\rho_z^{[t]}\) from \(\pi(\rho_z)\) and then simulate \(n\)-dimensional vector \(\mathbf{r}_{y}^{[t]}\) from \(\text{Gamma}(\alpha_{y}, \kappa_{y})\). Then compute \(\mathbf{y}_{\delta, rep}^{[t]} = \text{log}\left(r_{y}^{[t]}\right)\).
    \item When \(k=3\),  \(\mathbf{y}_{\delta,rep}^{[t]} \sim \text{Normal} (\mathbf{z}_{\delta}^{[t]}, \boldsymbol{\sigma}^{2*[t]})\), where \(\boldsymbol{\sigma}^{2*[t]}\) is a \(n\)-dimensional vector sampled from \(\pi(\boldsymbol{\sigma}^2)\).
    \item When \(k=4\), \(\alpha_{y} = \mathbf{z}_{\delta}^{[t]} + \alpha_{\xi}\) and \(\kappa_{y} = \bm{1}_{n}\). Simulate the \(n\)-dimensional vector \(\mathbf{r}_{y}^{[t]}\) from \(\text{Gamma}(\alpha_{y}, \kappa_{y})\) and then compute \(\mathbf{y}_{\delta, rep}^{[t]} = \text{log}\left(r_{y}^{[t]}\right)\).
    \item When \(k = 5\), \(\alpha_{y} = \mathbf{z}_{\delta}^{[t]} + \alpha_{\xi}\) and \(\kappa_{y} = \bm{1}_n + 2 \alpha_{\xi}\). Simulate the \(n\)-dimensional vector \(\mathbf{r}_{y}^{[t]}\) from \(\text{Beta}(\alpha_{y}, \kappa_{y} - \alpha_{y})\) and then compute \(\mathbf{y}_{\delta,rep}^{[t]} = \text{log}\left(\frac{\mathbf{r}_{y}^{[t]}}{\bm{1}_{n} -\mathbf{r}_{y}^{[t]}}\right)\).
    \end{itemize}
\State Sample \(\boldsymbol{\theta}^{*[t]}\) from \(\pi(\boldsymbol{\theta})\), where \(\sigma_{\xi}^{2*} > 0 \) is an element of \(\boldsymbol{\theta}^{*}\).
\State Sample \(\mathbf{w}_{\beta}^{[t]}\) from \(N(\bm{0}_{p,1}, \boldsymbol{D}_{\beta}(\boldsymbol{\theta}^{*[t]})\boldsymbol{D}_{\beta}(\boldsymbol{\theta}^{*[t]})')\).
\State Sample \(\mathbf{w}_{\eta}^{[t]}\) from \(N(\bm{0}_{r,1}, \boldsymbol{D}_{\eta}(\boldsymbol{\theta}^{*[t]})\boldsymbol{D}_{\eta}(\boldsymbol{\theta}^{*[t]})')\).
\State Sample \(\mathbf{w}_{\xi}^{[t]}\) from \(N(\bm{0}_{n,1}, \sigma_{\xi}^{2*[t]}\mathbf{I}_n).\)
\State Compute \(\boldsymbol{\xi}_{ rep}^{[t]}\), \(\boldsymbol{\beta}_{rep}^{[t]}\), and \(\boldsymbol{\eta}_{rep}^{[t]}\) using \(\mathbf{y}_{\delta,rep}^{[t]}\), \(\mathbf{w}_{\beta}^{[t]}\), \(\mathbf{w}_{\eta}^{[t]}\), and \(\mathbf{w}_{\xi}^{[t]}\) efficiently via Equation (\ref{equation11gen}) from General Theorem 2 and block matrix inversion formulas. 
\State Set \(t = t + 1\).
\State Repeat steps 1-7 until \(t = T\).
\State Let \(\boldsymbol{\beta}^{[1:T]} =  \{\boldsymbol{\beta}^{[t]}: t = 1,\dots,T \}\), \(\boldsymbol{\eta}^{[1:T]} =  \{\boldsymbol{\eta}^{[t]}: t = 1,\dots,T \}\), \(\boldsymbol{\xi}^{[1:T]} =  \{\boldsymbol{\xi}^{[t]}: t = 1,\dots,T \}\),  \(\boldsymbol{\delta}^{[1:T]} =  \{\boldsymbol{\delta}^{[t]}: t = 1,\dots,T \}\), \({\mathbf{Y}}^{[1:T]} = \mathbf{X}\boldsymbol{\beta}^{[1:T]} + \mathbf{G}\boldsymbol{\eta}^{[1:T]} \).
\State Compute summaries of \(\boldsymbol{\beta}^{[1:T]}\), \(\boldsymbol{\eta}^{[1:T]}\), and  \({\mathbf{Y}}^{[1:T]}\) (e.g. row means, variances, quantiles, etc.).
\end{algorithmic}
\label{algorithm2}
\end{algorithm}
This algorithm provides \(T\) independent replicates of \(\boldsymbol{\xi}\), \(\boldsymbol{\beta}\), \(\boldsymbol{\eta}\), \(\boldsymbol{\delta}\), and \({y}_i \equiv \mathbf{x}_i'\boldsymbol{\beta} + \mathbf{g}_i'\boldsymbol{\eta}\), where the \(j\)-th sampled replicates are defined by the columns of the matrices \(\boldsymbol{\xi}^{[1:T]}\), \(\boldsymbol{\beta}^{[1:T]}\), \(\boldsymbol{\eta}^{[1:T]}\), \(\boldsymbol{\delta}^{[1:T]}\), \({\mathbf{Y}}^{[1:T]}\) defined in Step 9. One can use summaries of \({\mathbf{Y}}^{[1:T]}\), \(\boldsymbol{\beta}^{[1:T]}\), and \(\boldsymbol{\eta}^{[1:T]}\), to make predictions at unobserved locations, inference on fixed effects, and inference on random effects respectively.  

\section{Simulation Study: Comparing to MCMC in a Small Sample Case}\label{appen:sim_study_mcmc}
{The results for the bivariate model fit with SM-EPR with subset size \(n = 400\), M-EPR, INLA, and Stan are presented in Table \ref{sim_n_1000_biv}. The subset size was chosen based on the elbow plots of the MSPE for each variable in Figure \ref{small_elbow}.} The bivariate model fit with SM-EPR and M-EPR had a smaller CPU time than the model fit with INLA and Stan. The CPU time for SM-EPR was slightly larger than the CPU time of M-EPR for this small subset size \(n = 400\) due to SM-EPR requiring repeated matrix inversions (a pattern that didn't arise in the large \(N\) case). This difference in CPU time between SM-EPR and M-EPR suggests that subsampling is unnecessary for small datasets. For the MSPE and the MSE, all four computational approaches performed similar, as indicated by the overlapping confidence intervals. SM-EPR and M-EPR outperformed INLA and Stan in terms of MSE. A notable difference is the significantly larger CPU time for Stan compared to the other three approaches. As the size of the data becomes much larger, Stan will become computationally burdensome, and consequently we exclude the comparison to Stan in the higher-dimensional study in Section \ref{sim_study} of the main text. 
\begin{table}[H]
\caption{Evaluation metrics for models fit using SM-EPR with subset size \(n = 400\), M-EPR, INLA, and Stan for each type of data for \(M = 1,000\).}
\label{sim_n_1000_biv}
\begin{center}
\begin{tabular}{ccccc}
\toprule
Approach & CPU & MSPE & MSE & CRPS \\ \midrule
\multicolumn{5}{c}{Weibull Response Data} \\ \midrule
SM-EPR  & 0.9970 & 0.1212 & 0.1765 & 0.1876 \\ 
 & (0.8466, 1.1473) & (0.0955, 0.1469) & (0.1503, 0.2027) & (0.1596, 0.2156) \\ \midrule
M-EPR  & 0.8372 & 0.1207 &  0.1796 & 0.1971 \\ 
 & (0.6612, 1.0131) & (0.0951, 0.1463) & ( 0.1452, 0.2140) & (0.1636, 0.2307) \\ \midrule
 INLA   & 41.2266 & 0.1128 & 36.7173 & 0.1905 \\ 
 & (38.8735 43.5796) & (0.0973, 0.1283) & (3.4560, 69.9786) & (0.1703, 0.2107) \\ \midrule  
 MCMC  & 1095.444 & 0.1140 &  1.0788 & 0.1925 \\ 
 & (1055.349, 1135.540) & (0.0945, 0.1335) & (0.8486, 1.3089) & (0.1642, 0.2208) \\ \midrule 
 \multicolumn{5}{c}{Logistic Response Data} \\  \midrule
 SM-EPR & 0.9970 & 0.1718 &  0.1765 & 0.2892 \\ 
 & (0.8466, 1.1473) & ( 0.1506, 0.1930) & (0.1503, 0.2027) & (0.2383, 0.3402) \\ \midrule
 M-EPR   & 0.8372 & 0.1652 &  0.1796 & 0.2573 \\ 
 & (0.6612, 1.0131) & (0.1509, 0.1795) & (0.1452, 0.2140) & (0.2089, 0.3058) \\ \midrule
 INLA  & 41.2266 & 0.1546 & 36.7173 & 0.2538 \\ 
 & (38.8735 43.5796) & (0.1480, 0.1613) & (3.4560, 69.9786) & (0.2435, 0.2641) \\ \midrule
 MCMC  & 1095.444 &  0.1569 &  1.0788 & 0.2460 \\ 
 & (1055.349, 1135.540) & (0.1483, 0.1655) & (0.8486, 1.3089) & (0.2279, 0.2642) \\
 \bottomrule   
\end{tabular}
\end{center}
\begin{flushleft}
\textit{Note}: The first column presents the name of the computational approach, the second column displays the average CPU time measured in seconds, the third column presents the mean square prediction error, \(\text{MSPE} = \frac{1}{N} \sum_{i = 1}^N (Y_i - \text{E}[{Y}_i \vert \textbf{z}])^2 \), the fourth column contains the average MSE and the fifth column contains the average CRPS. All averages are taken over 50 replicates along with plus or minus two standard deviations. The MSPE was calculated on the log-scale for the Weibull setting. 
\end{flushleft}
\end{table}

\begin{figure}[H]
\centering
  \centering
  \includegraphics[width=0.75\linewidth]{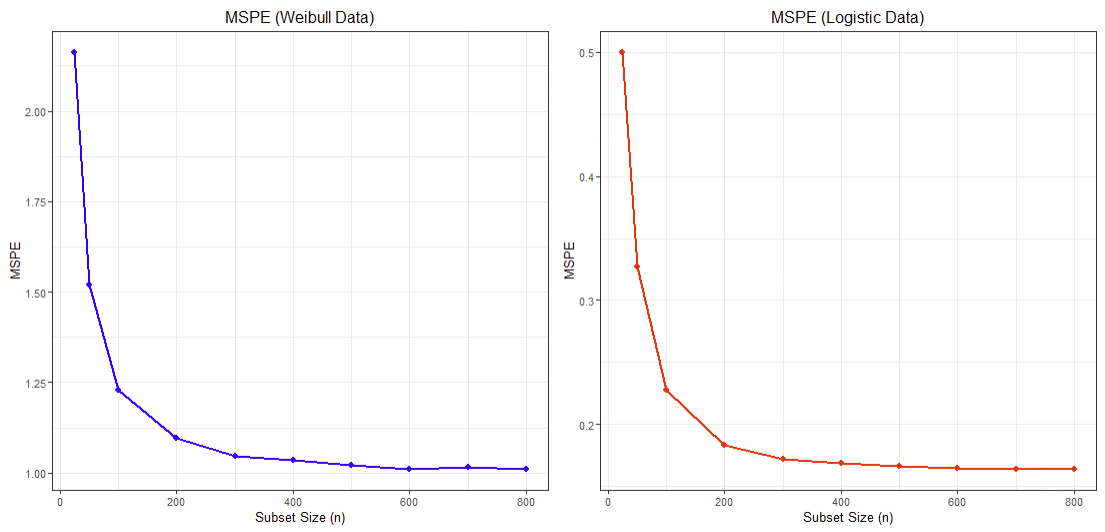}  
\caption{MSPE over various subset sizes (n).}
\label{small_elbow}
\end{figure}

\section{Additional Univariate Spatial Simulations}\label{appen:addsims}
\subsection{Univariate Simulation Setup}
In the main text, we provide two simulation studies for a multi-type response scenario with Weibull distributed data and logistic distributed data. Other distributions of interest are Gaussian, Poisson, and Bernoulli. To illustrate the computational advantages of our method in these more standard cases, we present additional simulation results. The first analysis will be the univariate spatial scenario (i.e., K = 1) for each data type with \(M =1,000\). We assume that 20$\%$ is missing. This dataset size was used to compare the fit of a correctly specified model using INLA and Stan to the fit of a discrepancy model using EPR and S-EPR. Recall from the main text, in the univariate spatial scenario, we refer to SM-EPR as Scalable Exact Posterior Regression (S-EPR). The second analysis will be the univariate spatial scenario for each data type with \(M = 60,000\). This dataset size was used to compare the fit of a correctly specified model using INLA with the discrepancy model fit using EPR and S-EPR (Stan excluded for computational reasons). The univariate data for both analyses will be simulated using the following distributions:
\begin{align*}
    &Z_3(s_i)|\eta_1, \dots \eta_{30}, \xi(s_i) \sim \text{Normal}\left(2.5 - \frac{1}{2}x_1(s_i) - 2x_2(s_i) + \sum_{j = 1}^{30}g_j(s_i)\eta_{j} + \xi(s_i), 0.25 \right), \\[1em]
    &Z_4(s_i)|\eta_1, \dots \eta_{30}, \xi(s_i) \sim \text{Poisson}\left(\text{exp}\left\{-1 - 0.4x_1(s_i) - 1.2x_2(s_i) + \sum_{j = 1}^{30}g_j(s_i)\eta_{j} + \xi(s_i) \right\} \right),  \\[1em]
    &Z_5(s_i)|\eta_1, \dots \eta_{30}, \xi(s_i) \sim \text{Bernoulli} \left(\frac{\text{exp}\left\{-5 + x_1(s_i) - x_2(s_i) + \sum_{j = 1}^{30}g_j(s_i)\eta_{j} + \xi(s_i) \right\}}{1 + \text{exp}\left\{-5 + x_1(s_i) - x_2(s_i) + \sum_{j = 1}^{30}g_j(s_i)\eta_{j} + \xi(s_i) \right\}} \right),
\end{align*}
\noindent for \(i = 1, \dots, M\) and \(\mathbf{s} \in \{1, 2, \dots, M\}\) is the one-dimensional spatial domain. We observe \(N = 0.8\times M\) randomly selected locations out of \(M\). We simulate the Bernoulli distributed covariate \(x_1(s_i)\) with probability \(\text{exp}(\frac{1}{M} s_i)/ (1 + \text{exp}(\frac{1}{M} s_i))\) and the Bernoulli distributed covariate \(x_2(s_i)\) with probability \(\text{exp}(\frac{-0.01}{M} s_i)/ (1 + \text{exp}(\frac{-0.01}{M} s_i))\). When \(k = 3\), \(\{\eta_j\}\) are simulated from a normal distribution with mean zero and variance \(0.81\) and \(\xi(s_i)\) are simulated from a normal distribution with mean zero and variance \(0.07\). When \(k = 4\), \(\{\eta_j\}\) are simulated from a normal distribution with mean \(0.2\) and variance \(0.04\) and \(\xi(s_i)\) are simulated from a normal distribution with mean zero and variance \(0.01\). When \(k = 5\), \(\{\eta_j\}\) are simulated from a normal distribution with mean \(0.2\) and variance \(0.04\) and \(\xi(s_i)\) are simulated from a normal distribution with mean zero and variance \(0.01\).

\subsection{A Simulated Univariate Example}
Using \(M = 60,000\), we compare EPR with S-EPR across each data type. In Figure \ref{sim_epr_sepr_plot}, we present plots of the posterior mean of \({\mathbf{Y}}^{[1:T]}\) for each data type, for both S-EPR and EPR. The subset size for S-EPR was selected based on the elbow plots of the mean square prediction error between the true latent process and the predicted latent process and both observed and unobserved locations. Across all three data types, we observed very similar predictions for both methods. The advantage of using S-EPR is that it offers predictions comparable to those of EPR but with the additional benefit of reduced computation time. Specifically, S-EPR required approximately 6 seconds of CPU time, whereas EPR required approximately 16 seconds. As the size of the data increases, this difference in CPU time becomes more pronounced. 

\begin{figure}[H]
\centering
\begin{subfigure}{.5\textwidth}
  \centering
  \includegraphics[width=.9\linewidth]{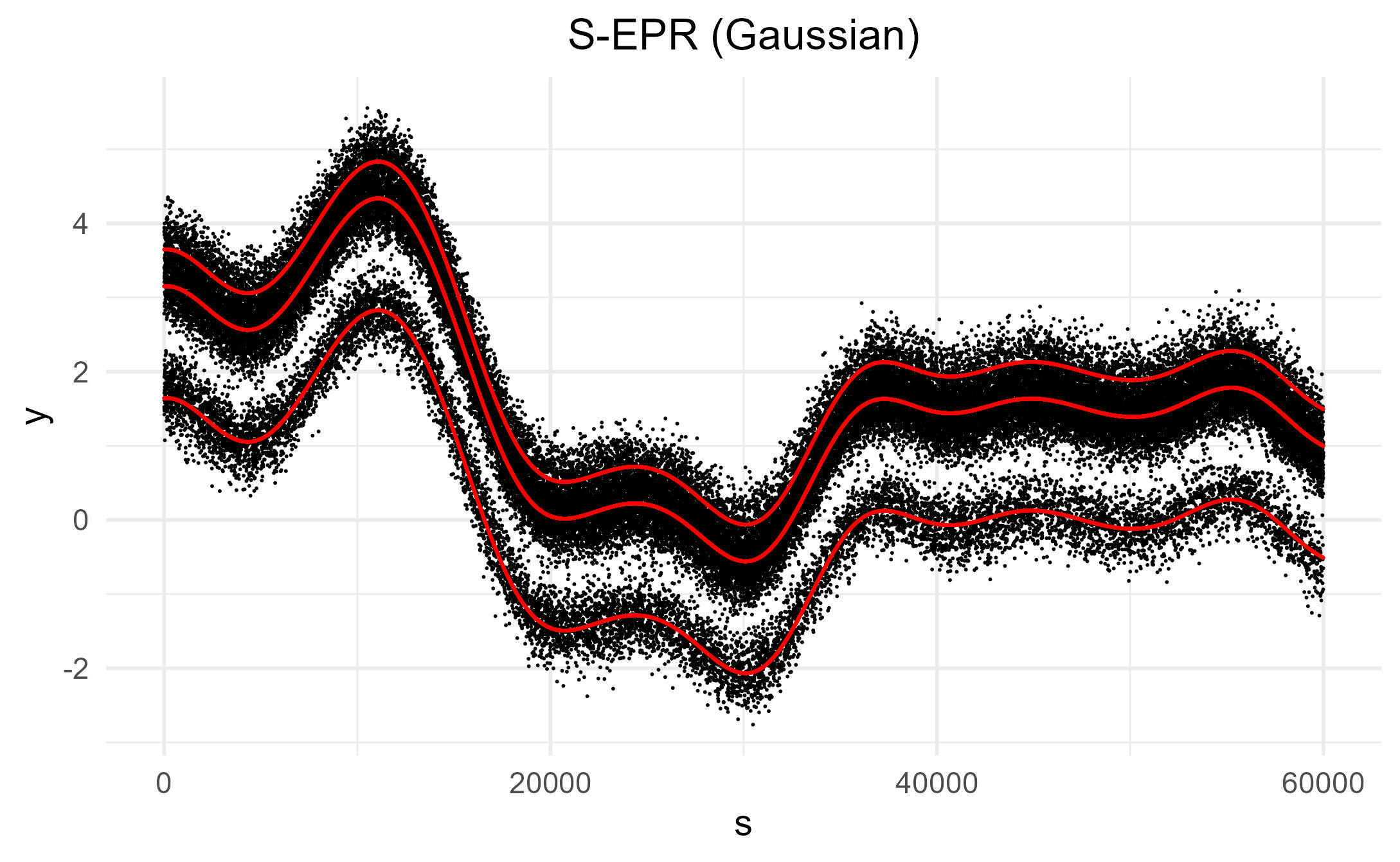}  
\end{subfigure}%
\begin{subfigure}{.5\textwidth}
  \centering
  \includegraphics[width=.9\linewidth]{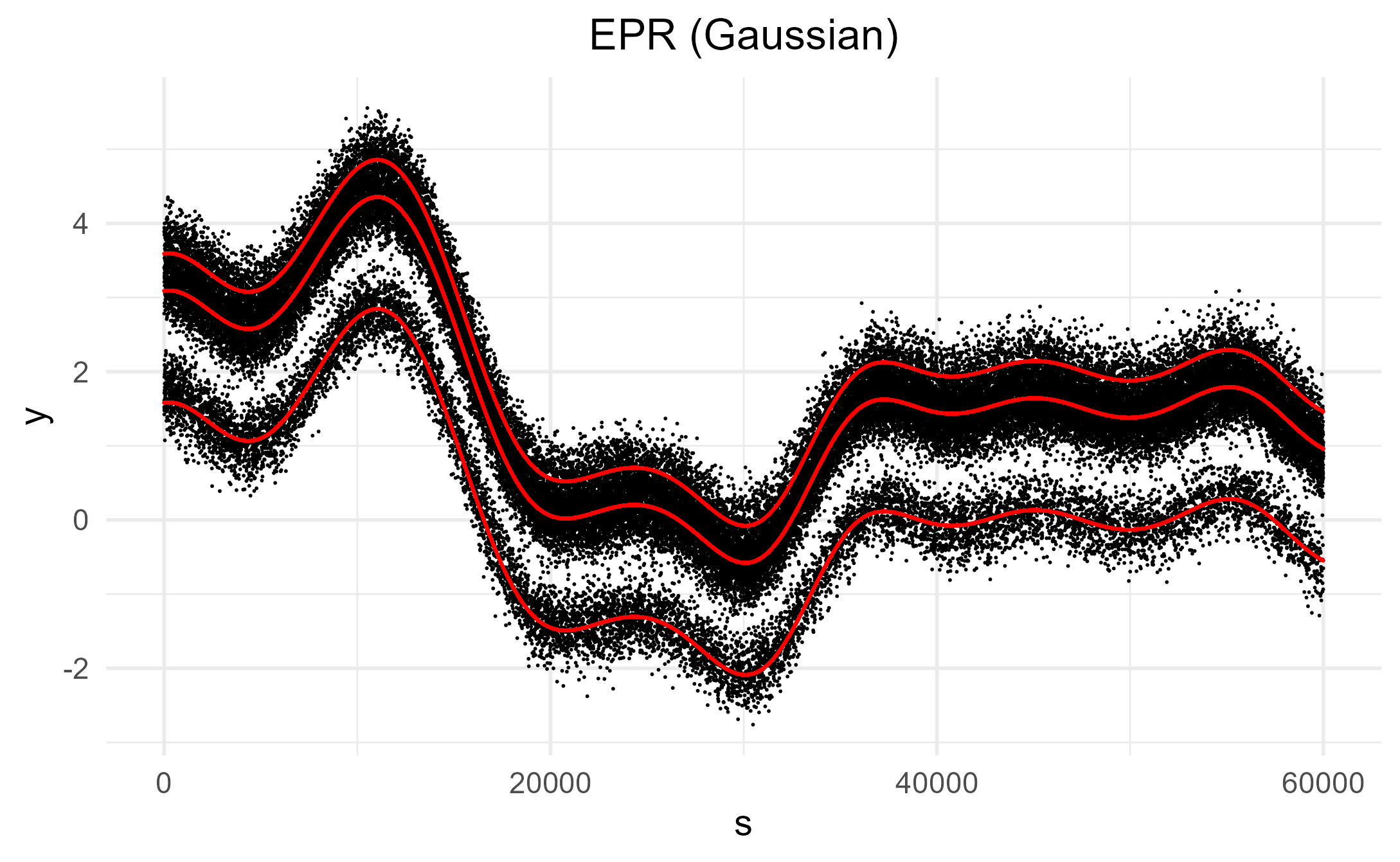}
\end{subfigure}
\begin{subfigure}{.5\textwidth}
  \centering
  \includegraphics[width=.9\linewidth]{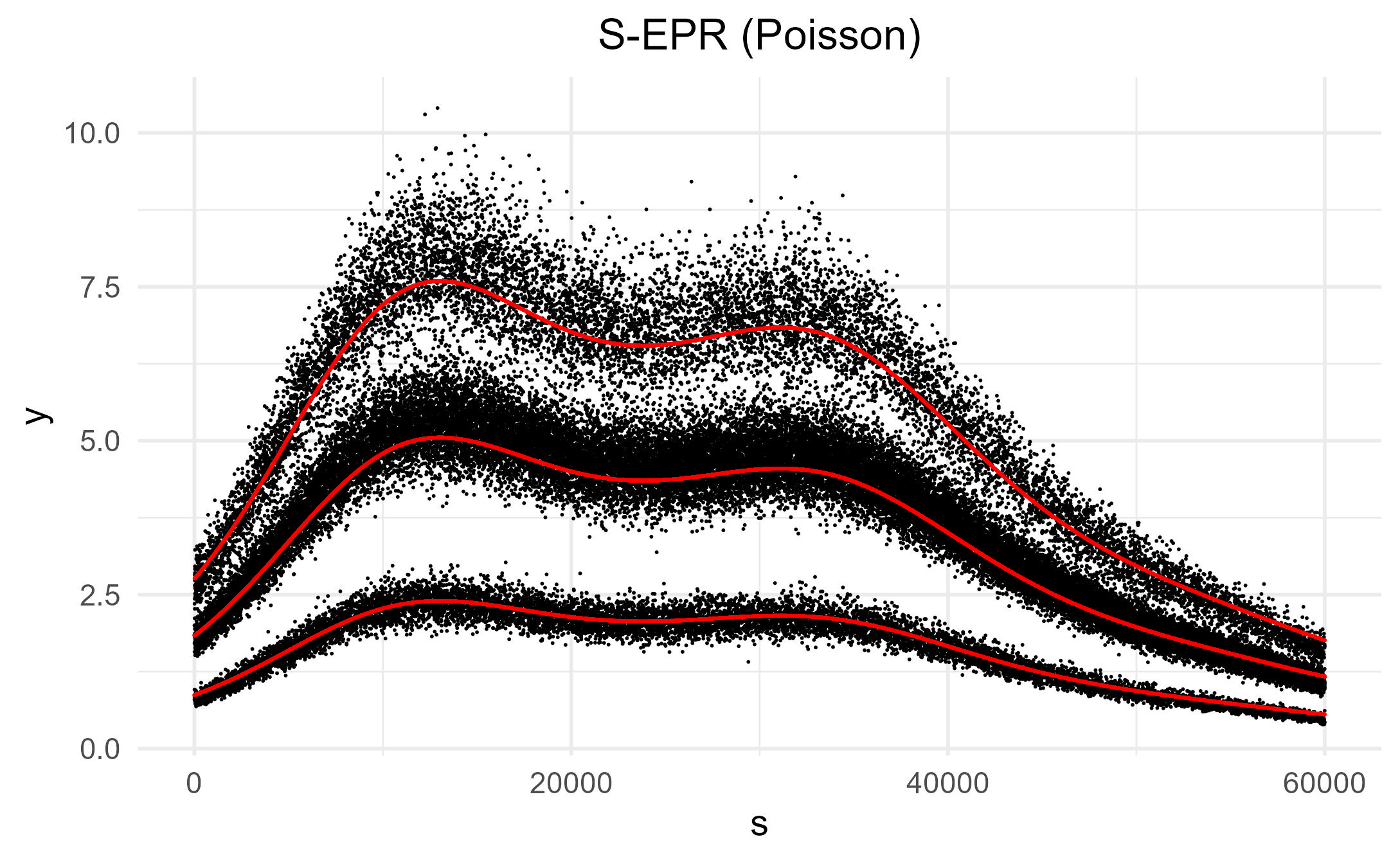}
\end{subfigure}%
\begin{subfigure}{.5\textwidth}
  \centering
  \includegraphics[width=.9\linewidth]{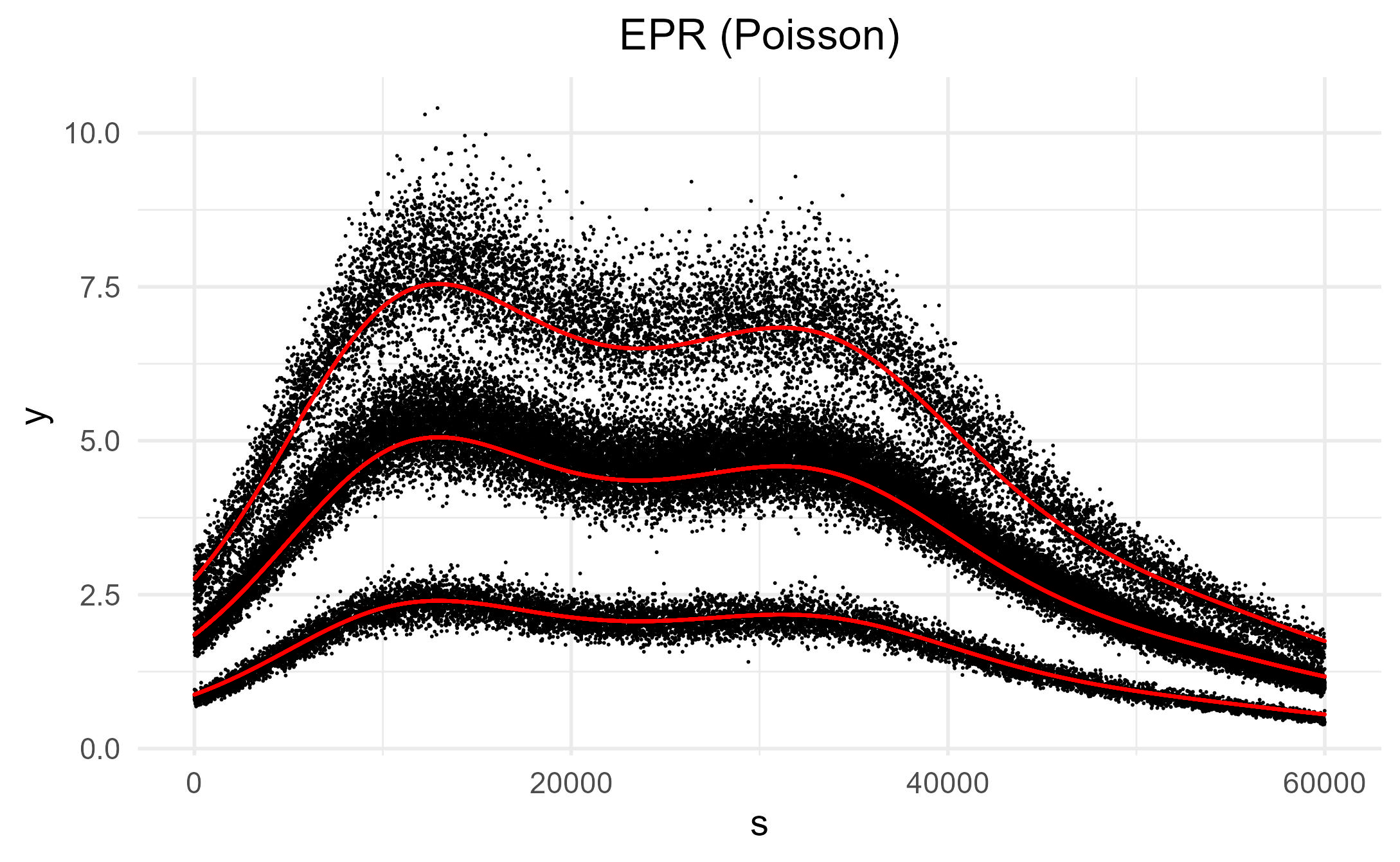}
\end{subfigure}
\begin{subfigure}{.5\textwidth}
  \centering
  \includegraphics[width=.9\linewidth]{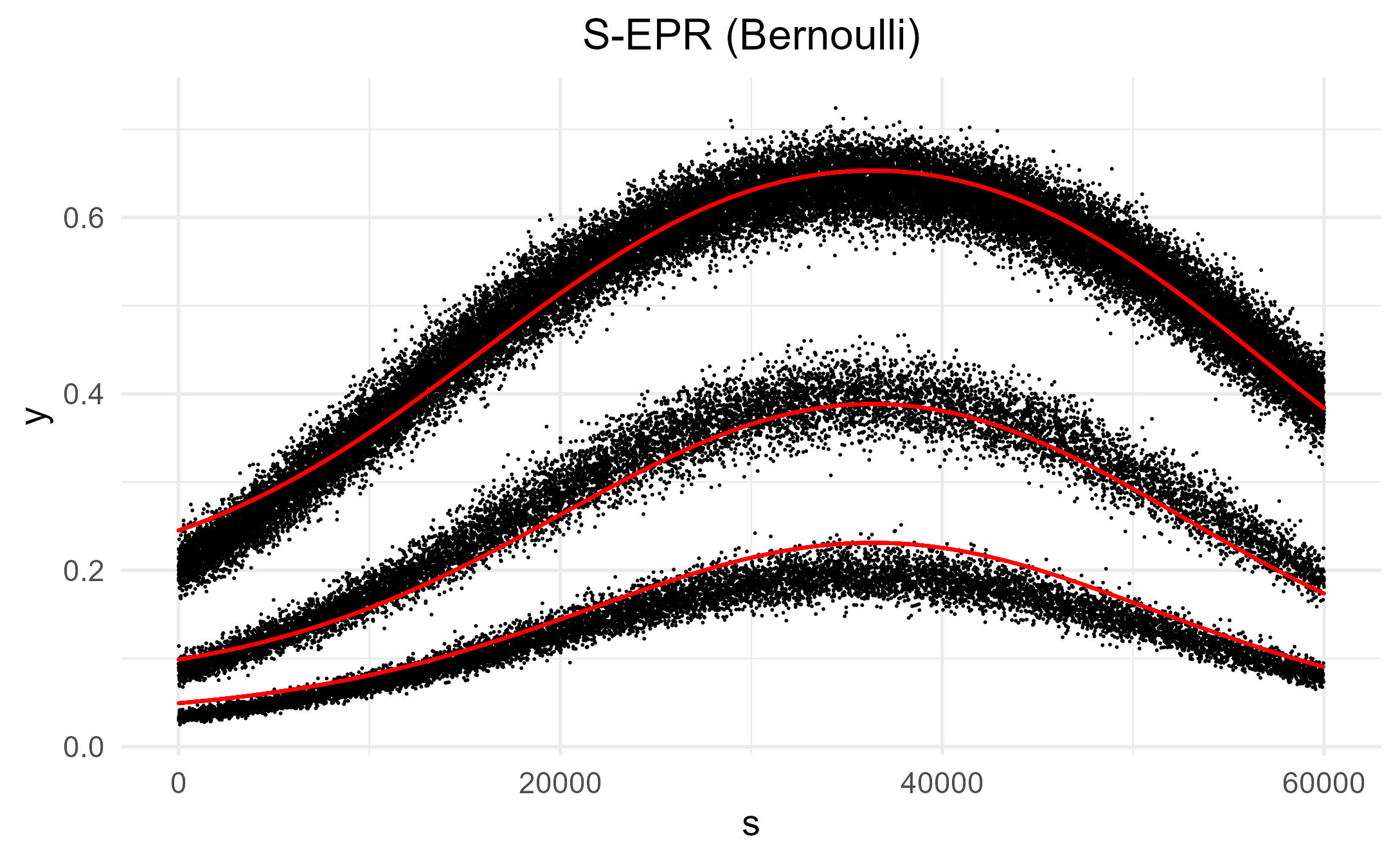}
\end{subfigure}%
\begin{subfigure}{.5\textwidth}
  \centering
  \includegraphics[width=.9\linewidth]{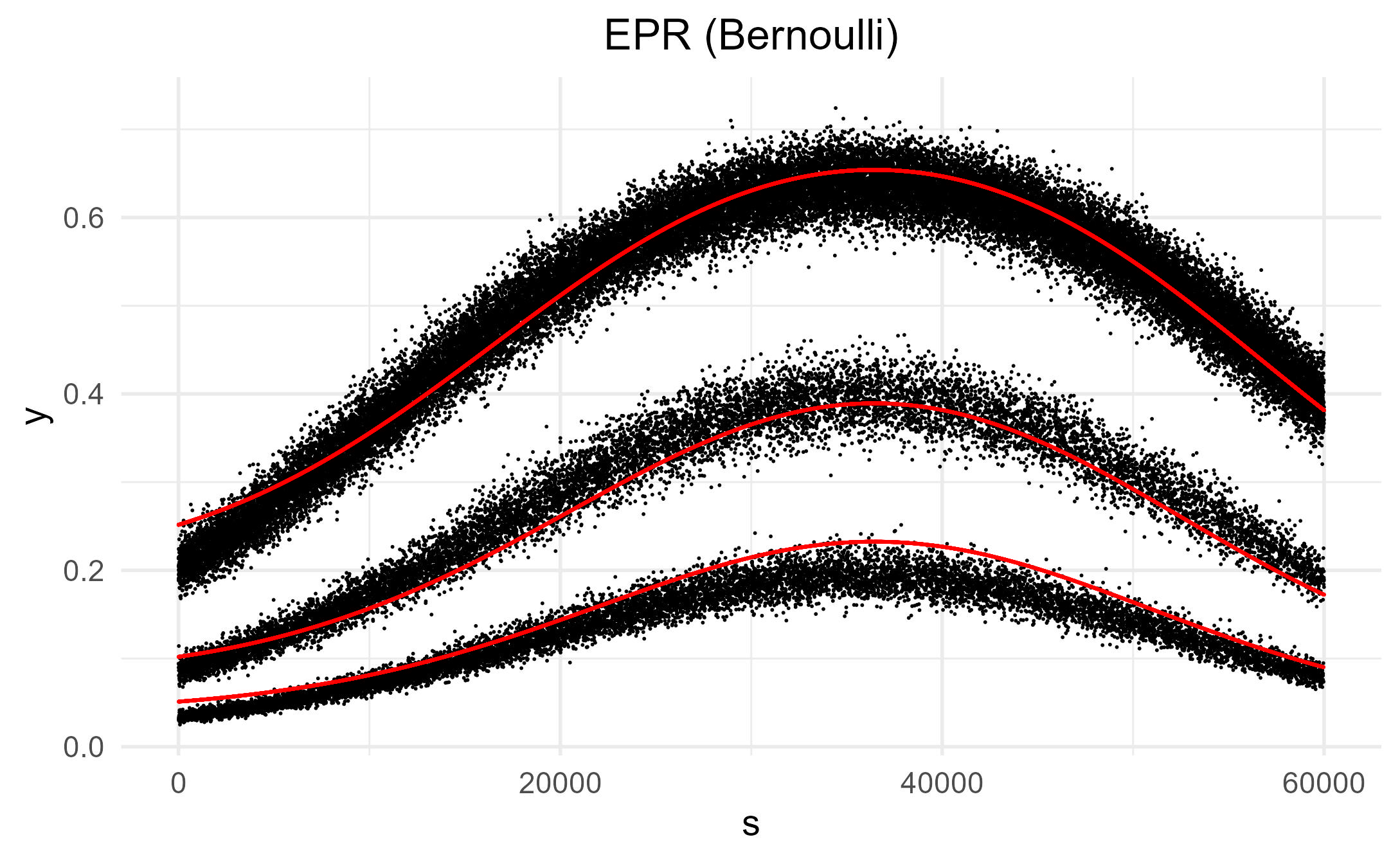}
\end{subfigure}
\caption{Illustration of SM-EPR and EPR predictions. The first row displays predictions for the Gaussian spatial data scenario, the second row displays predictions for the Poisson spatial data scenario, and the third row displays predictions for the Bernoulli spatial data scenario. The black points represent the true latent process \(\mathbf{y}\) and the red lines represent the posterior mean of \({\mathbf{Y}}^{[1:T]} = \mathbf{X}\boldsymbol{\beta}^{[1:T]} + \mathbf{G}\boldsymbol{\eta}^{[1:T]}\).}
\label{sim_epr_sepr_plot}
\end{figure}

\subsection{A Univariate Simulation Study}
The first simulation study involves univariate spatial datasets of size \(M = 1,000\). This dataset size was used to compare the fit of a correctly specified model using INLA and Stan to the fit of a discrepancy model using EPR and S-EPR. The central processing unit (CPU) time measured in seconds, the mean square prediction error (MSPE) between the true latent process and the predicted latent process, the mean square estimation error (MSE) between the true coefficients of the fixed and random effects and the predicted \(\boldsymbol{\beta}\) and \(\boldsymbol{\eta}\), and the continuous rank probability score (CRPS) were used to evaluate the different approaches. Table \ref{sim_n_1000_biv_uni} displays averages of the evaluation metrics calculated over 50 replicates along with plus or minus two standard deviations.

\begin{table}[H]
\caption{Evaluation metrics for models fit using S-EPR with subset size \(n = 400\), EPR, ILNA, and Stan for each type of data for \(M = 1,000\).}
\label{sim_n_1000_biv_uni}
\begin{center}
\begin{tabular}{ccccc}
\hline
Approach & CPU & MSPE & MSE & CRPS \\ \hline
\multicolumn{5}{c}{Gaussian Data} \\ \hline
S-EPR & 0.308 & 0.125 &  0.438 & 0.249 \\ 
 & (0.225, 0.391) & (0.083, 0.168) & (0.307, 0.569) & ( 0.176, 0.321) \\ \hline 
EPR  & 0.375 & 0.110 & 0.424 & 0.209 \\ 
 & (0.212, 0.537) & (0.086, 0.135) & (0.125, 0.722) & (0.146, 0.272) \\  \hline
INLA & 24.198 & 0.117 & 35.829 & 0.199 \\ 
 & (20.271, 28.124) & (0.089, 0.145) & (0.000, 78.889) & (0.163, 0.236) \\ \hline 
MCMC & 259.210 & 0.281 & 0.872 & 0.305\\ 
 & (207.071, 311.349) & (0.000, 0.694) & (0.173, 1.570) & (0.251, 0.359) \\ \hline 
\multicolumn{5}{c}{Poisson Data} \\  \hline
S-EPR & 0.279 & 0.014 & 0.103 &  0.066 \\ 
 & (0.229, 0.329) & (0.005, 0.022) & (0.092, 0.115) & (0.046, 0.085) \\ \hline
EPR  & 0.369 & 0.014 & 0.103 & 0.094 \\ 
 & (0.242, 0.496) & (0.004, 0.023) & (0.083, 0.123) & (0.063, 0.125)\\ \hline 
INLA & 15.285 & 0.012 & 26.809 & 0.062 \\ 
 & (6.345, 24.225) & (0.004, 0.019) & (0.000, 84.299) & (0.043, 0.082)  \\ \hline
MCMC& 165.483 & 0.011 & 0.097 & 0.141 \\ 
 & (132.389, 198.577) & (0.003, 0.018) & (0.057, 0.137) & (0.129, 0.153) \\
 \hline
 \multicolumn{5}{c}{Bernoulli Data}\\ \hline
S-EPR & 0.266 & 0.0019 & 0.669 & 0.149 \\ 
 &(0.182, 0.349) & ( 0.0002, 0.0031) & (0.438, 0.883) & (0.028, 0.222)\\ \hline 
 EPR & 0.362 & 0.0017 & 0.661 & 0.125 \\ 
 & (0.222, 0.502) & (0.0002, 0.0031) & (0.438, 0.883) & (0.028, 0.222) \\  \hline 
 INLA& 19.199 & 0.0021 & 5.844 & 0.132 \\ 
 & (15.924, 22.475) & (0.0002, 0.0039) & (0.000, 18.981) & (0.064, 0.199) \\ \hline
 MCMC & 220.530 & 0.0040 & 0.634 & 0.194 \\ 
 & (205.145, 235.916) & (0.0020, 0.0059) & (0.598, 0.670) & (0.162, 0.226) \\  \hline 
 \hline
\end{tabular}
\end{center}
\begin{flushleft}
\textit{Note}: The first column presents the name of the computational approach, the second column displays the average CPU time measured in seconds, the third column presents the mean square prediction error, \(MSPE = \frac{1}{N} \sum_{i = 1}^N (Y_i - E[{Y}_i \vert \textbf{z}])^2 \), the fourth column contains the average MSE and the fifth column contains the average CRPS. All averages are taken over 50 replicates along with plus or minus two standard deviations. The MSPE was calculated on the log-scale for the Poisson setting. 
\end{flushleft}
\end{table}

For all three types of data S-EPR and EPR had a smaller CPU time. The difference in CPU time between S-EPR and EPR for subset size \(n=400\) was small, suggesting that subsampling is unnecessary for small datasets. For the other three evaluation metrics, the four computational approaches performed similarly, as indicated by the overlapping confidence intervals. The main difference is the significantly larger CPU time for Stan compared to the other three approaches. As the size of the data becomes much larger, Stan will become computationally burdensome, so we exclude the comparison to Stan for the second stage of our simulation study. 

The second simulation study involves univariate spatial datasets of size \(M = 60,000\). This dataset size was used to compare the fit of a correctly specified model using INLA with the discrepancy model fit using EPR and S-EPR. Table \ref{sim_n_60000_uni} presents the results for these simulations for the three types of data, including the average CPU time, MSPE, MSE, and CRPS over 50 replicates with confidence intervals constructed using plus or minus two standard deviations. S-EPR consistently required the smallest CPU time for all types of data. The MSPE, MSE, and CRPS values for S-EPR, with a subset size of \(n=10,000\), were very similar to those for EPR. This demonstrates that S-EPR achieves inference comparable to EPR but in less time. Across all three approaches, the performance for Gaussian data was similar. For the Poisson and Bernoulli settings, INLA had slightly better performance than EPR and S-EPR. The increase in data size may explain INLA's improvement, as INLA relies on Laplace approximations, which use the normal distribution to approximate the marginal posterior distributions. 

\begin{table}[H]
\caption{Evaluation metrics for models fit using S-EPR with subset size \(n =10,000\), EPR, ILNA, and Stan for each type of data for \(M = 60,000\).}
\label{sim_n_60000_uni}
\begin{center}
\begin{tabular}{ccccc}
\hline
Approach & CPU & MSPE & MSE & CRPS \\ \hline
\multicolumn{5}{c}{Gaussian Data} \\ \hline
S-EPR & 6.104 & 0.069 & 0.377 & 0.155 \\ 
 & (5.787, 6.421) & ( 0.067, 0.071) & ( 0.280, 0.475) & (0.136, 0.174) \\  \hline 
EPR  & 15.305 & 0.068 & 0.323 & 0.171 \\ 
 & (14.503, 16.107) & (0.067, 0.069) & (0.173, 0.474) & (0.163, 0.178)  \\ \hline
INLA & 33.547 & 0.068 & 24.044 & 0.194 \\ 
 & (32.278, 34.816) & (0.0680, 0.0688) & (0.000, 48.240) & (0.192, 0.195) \\ \hline 
\multicolumn{5}{c}{Poisson Data} \\  \hline
S-EPR & 6.070 & 0.0086 & 0.115 &  0.054 \\ 
 & (5.832, 6.309) & (0.0079, 0.0094) & (0.038, 0.192) & (0.049, 0.059) \\ \hline
EPR  & 15.611 & 0.0084 &  0.114 &  0.058 \\ 
 & (14.459, 16.764) & (0.0079, 0.0090) & (0.000, 0.287) & (0.053, 0.064)\\ \hline 
INLA & 22.854 & 0.0070 & 13.171 &  0.062 \\ 
 & (21.436, 24.272) & (0.0069, 0.0071) & (0.000, 36.016) & (0.061, 0.063) \\ \hline 
 \multicolumn{5}{c}{Bernoulli Data}\\ \hline
S-EPR & 6.478 & 0.0006 & 0.810 & 0.079 \\ 
 &(6.130, 6.826) & (0.0005, 0.0007) & (0.000,2.916) & (0.069, 0.089)\\  \hline 
EPR & 16.602 & 0.0006 & 0.735 & 0.092 \\ 
 & (15.578,17.626) & (0.0005, 0.0007) & (0.305, 1.164) & (0.079, 0.106) \\  \hline 
INLA& 25.848 &  0.0004 & 1.067 & 0.069 \\ 
 & (25.019, 26.677) & (0.00042, 0.00048) & (0.000, 2.644) & (0.066, 0.072)  \\ \hline
\end{tabular}
\end{center}
\begin{flushleft}
\textit{Note}: The first column presents the name of the computational approach, the second column displays the average CPU time measured in seconds, the third column presents the mean square prediction error, \(MSPE = \frac{1}{N} \sum_{i = 1}^N (Y_i - E[{Y}_i \vert \textbf{z}])^2 \), the fourth column contains the average MSE and the fifth column contains the average CRPS. All averages are taken over 50 replicates along with plus or minus two standard deviations. The MSPE was calculated on the log-scale for the Poisson setting. 
\end{flushleft}
\end{table}

\section{Additional Bivariate Spatial Simulations}\label{appen:sim_plots_figs}
\subsection{A Simulated Example}\label{sim_fig}
Using \(M = 60,000\), we compare SM-EPR to M-EPR and INLA for multiple-type response data distributed as Weibull and logistic. In Figure \ref{sim_epr_sepr_inla_plot}, we present plots of the posterior mean of \(\text{exp}\left(\frac{-{\mathbf{Y}}^{[1:T]}}{{\rho}^{[1:T]}}\right)\Gamma\left(1 + \frac{1}{{\rho}^{[1:T]}}\right)\) and \({\mathbf{Y}}^{[1:T]}\) for the Weibull and logistic scenarios, respectively. The subset size for SM-EPR was selected based on the elbow plots of the mean square prediction error between the true latent process and the predicted latent process at unobserved locations that were held out for model evaluation. For both data types in this multivariate spatial model, we observe very similar predictions for each of the three approaches used. The advantage of using SM-EPR is that it offers predictions comparable to those of M-EPR and INLA but with the additional benefit of reduced computation time. Specifically, SM-EPR with a subset size of \(n = 5,000\) required approximately 12 seconds of central processing unit (CPU) time, whereas 72 and 137 seconds were required for M-EPR and INLA, respectively. As the total size of the data \(N\) increases, this difference in CPU time becomes more pronounced.
\begin{figure}[H]
\centering
\begin{subfigure}{.5\textwidth}
  \centering
  \includegraphics[width=.9\linewidth]{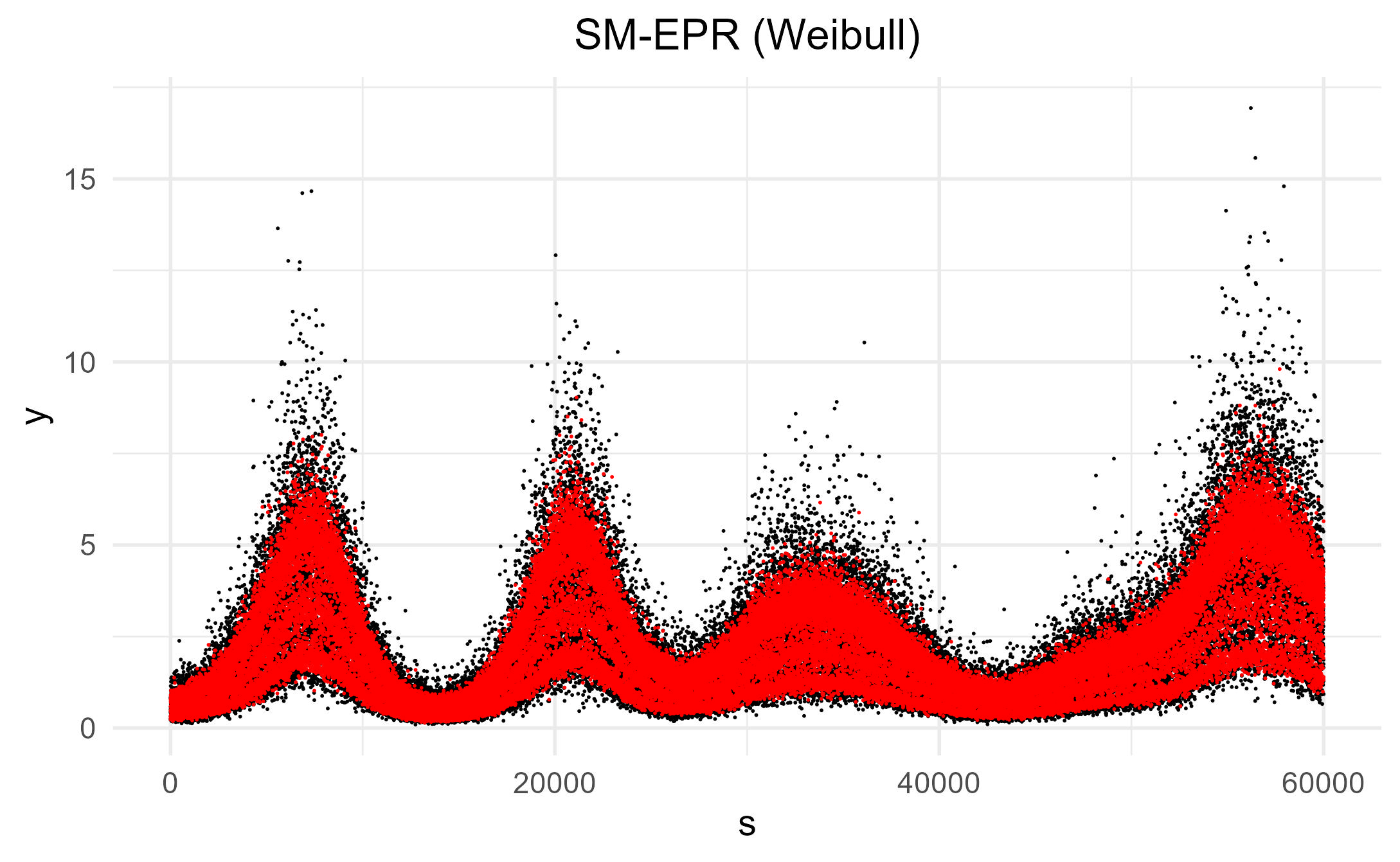} 
\end{subfigure}%
\begin{subfigure}{.5\textwidth}
  \centering
  \includegraphics[width=.9\linewidth]{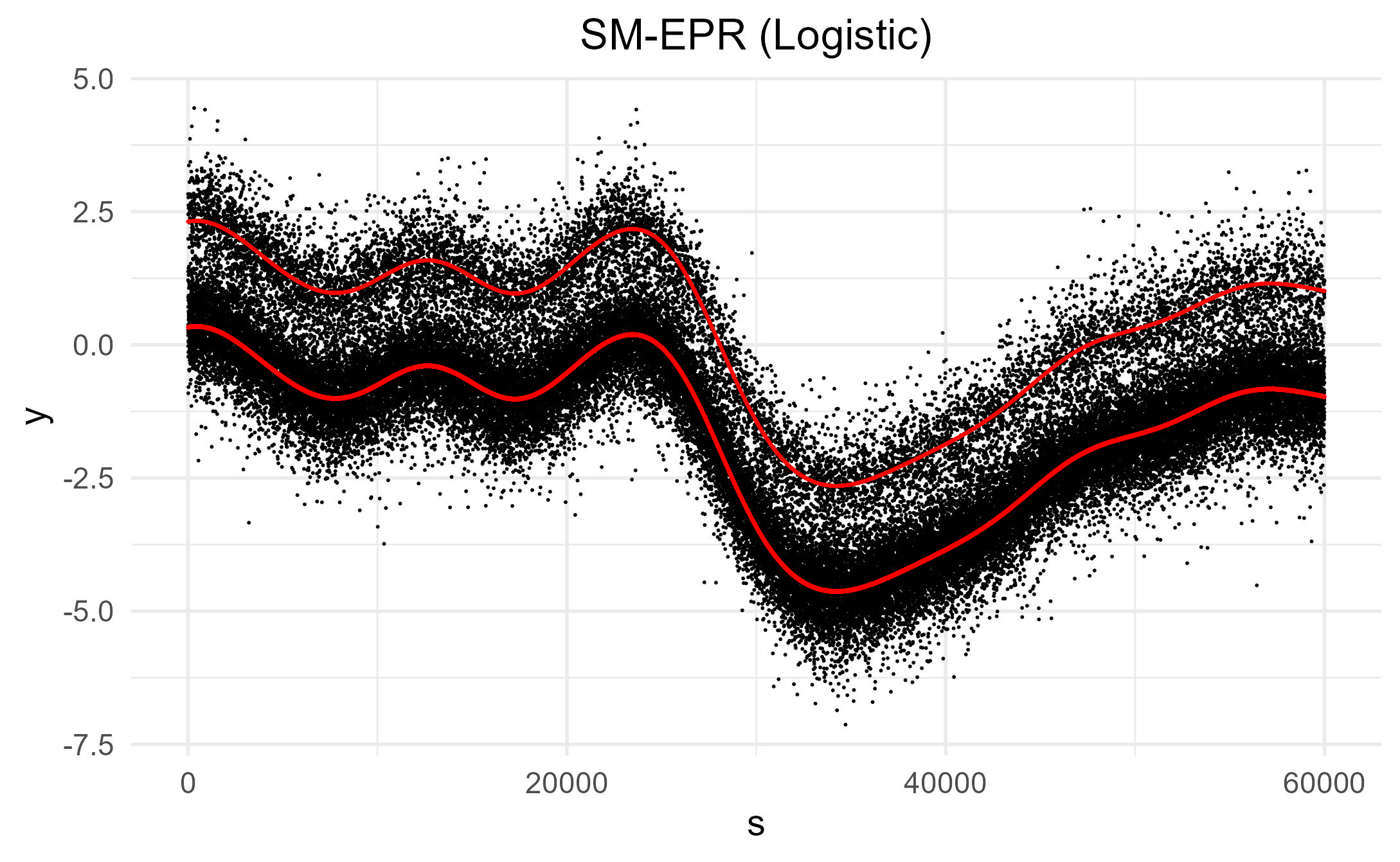}
\end{subfigure}
\begin{subfigure}{.5\textwidth}
  \centering
  \includegraphics[width=.9\linewidth]{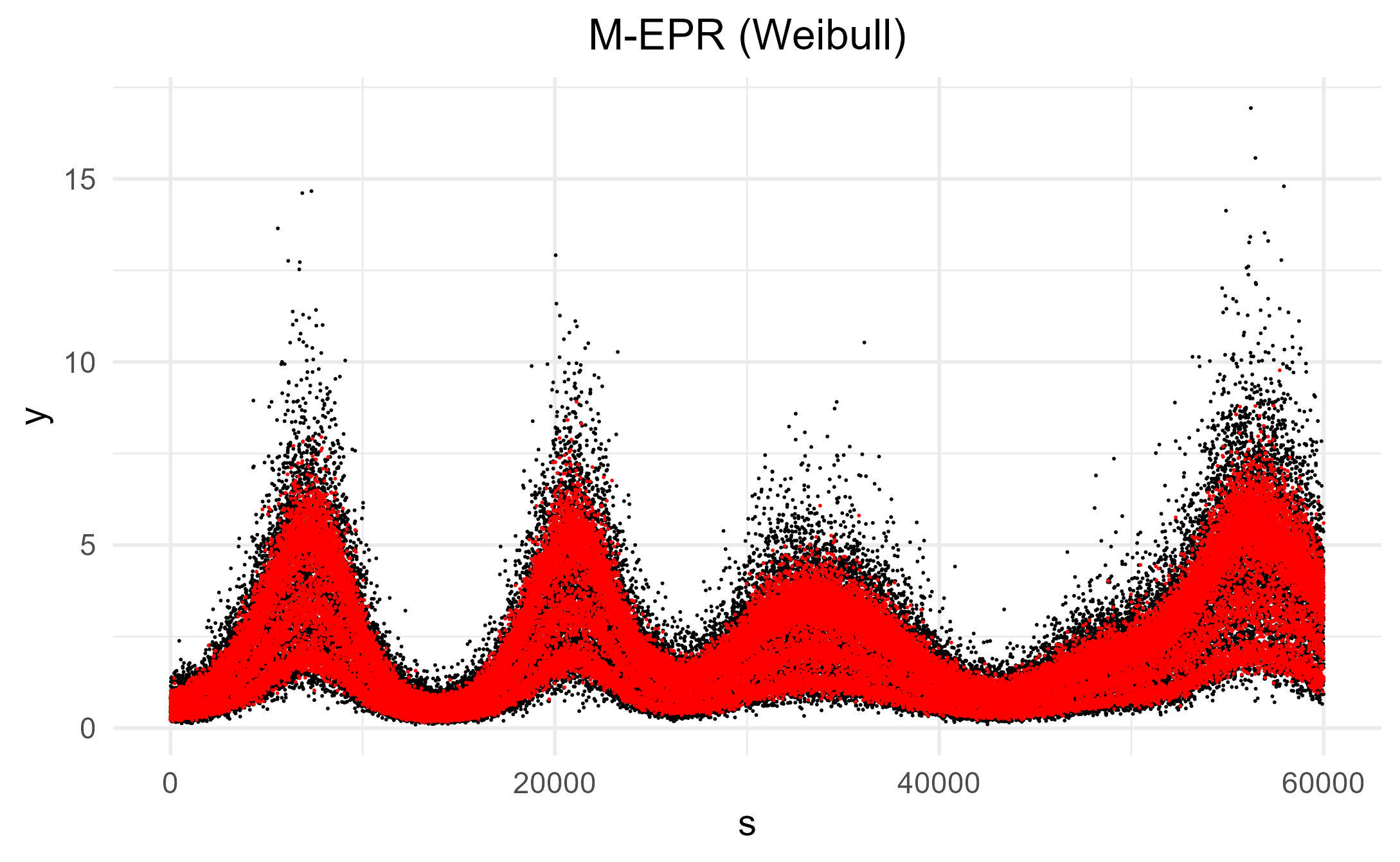}
\end{subfigure}%
\begin{subfigure}{.5\textwidth}
  \centering
  \includegraphics[width=.9\linewidth]{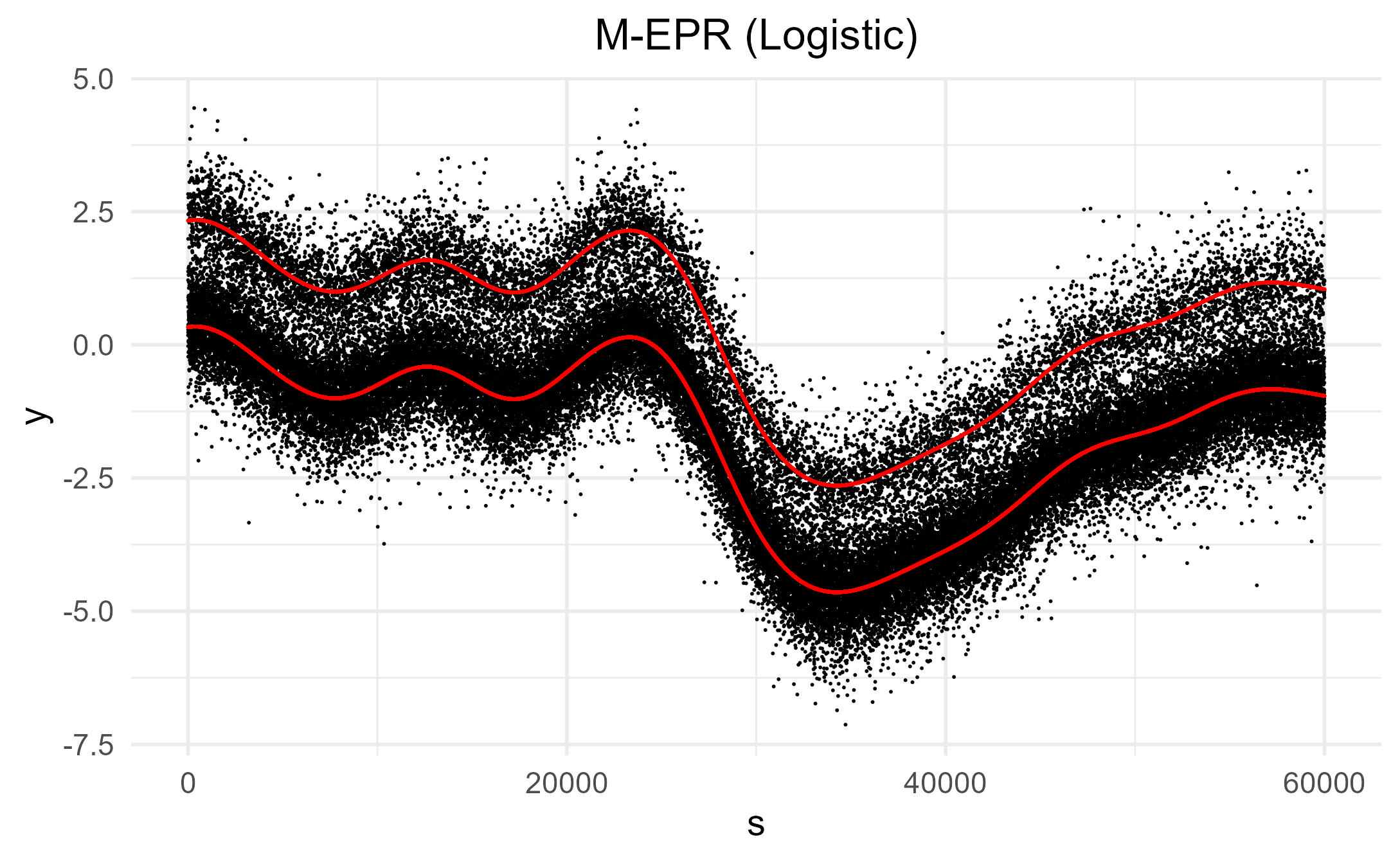}
\end{subfigure}
\begin{subfigure}{.5\textwidth}
  \centering
  \includegraphics[width=.9\linewidth]{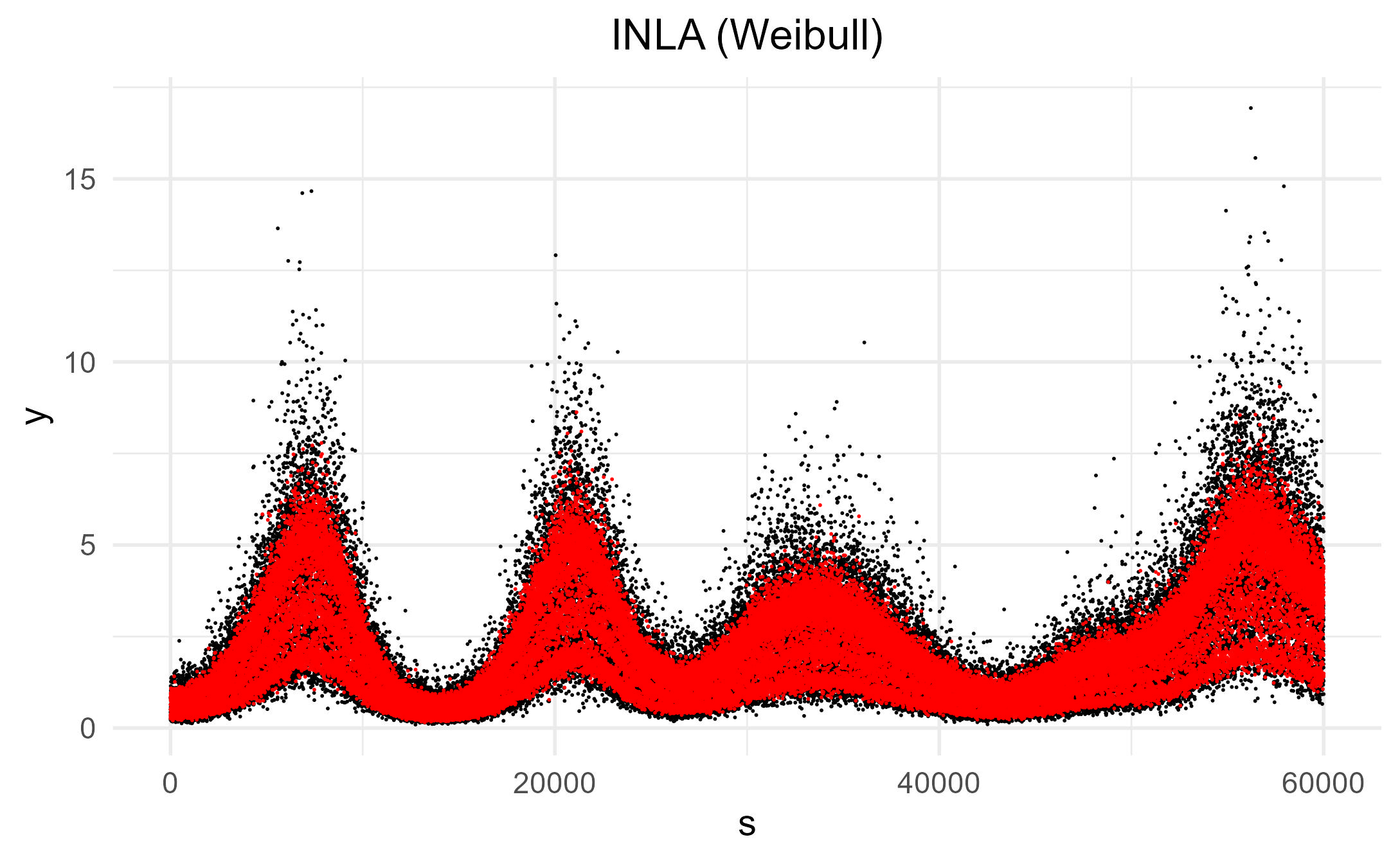}
\end{subfigure}%
\begin{subfigure}{.5\textwidth}
  \centering
  \includegraphics[width=.9\linewidth]{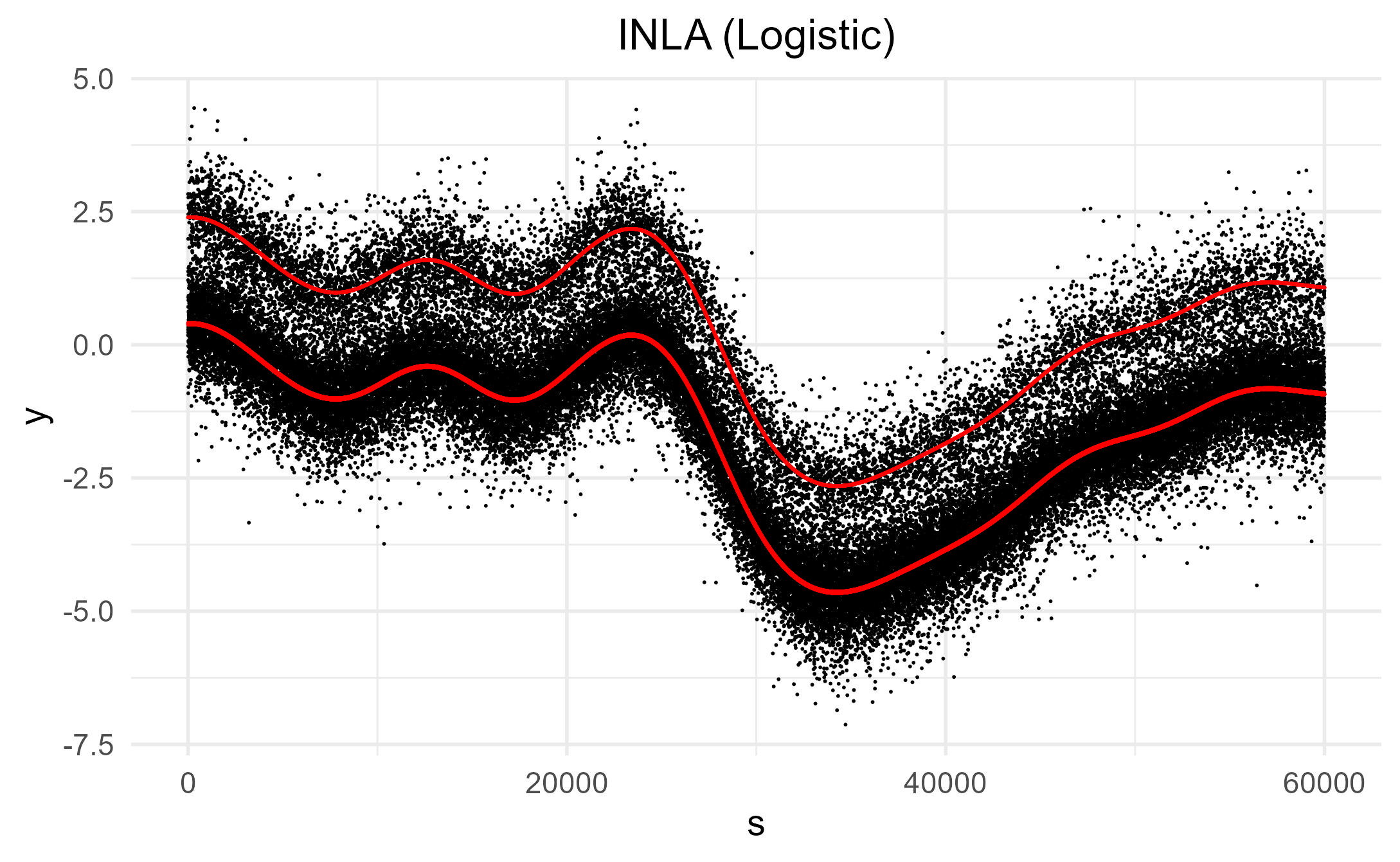}
\end{subfigure}
\caption{Illustration of SM-EPR, M-EPR, and INLA predictions. The first column displays the predictions for the Weibull spatial response data and the second column displays the predictions for the logistic spatial response data. The black points represent the true latent process \(\mathbf{y}\). The red points/lines represent the posterior mean of \(\text{exp}\left(\frac{-{\mathbf{Y}}^{[1:T]}}{{\rho}^{[1:T]}}\right)\Gamma\left(1 + \frac{1}{{\rho}^{[1:T]}}\right)\) and \({\mathbf{Y}}^{[1:T]}\) for the Weibull and logistic response, respectively. Where \({\mathbf{Y}}^{[1:T]} = \mathbf{X}\boldsymbol{\beta}^{[1:T]} + \mathbf{G}\boldsymbol{\eta}^{[1:T]}\).}
\label{sim_epr_sepr_inla_plot}
\end{figure}

\subsection{Basis Function Misspecification}\label{appen:miss}
In this simulation study, we compare the performance of a misspecified implementation of SM-EPR and M-EPR. In particular, we aim to evaluate the prediction performance of the full M-EPR model  (\(\boldsymbol{\delta} = \mathbf{1}_N\)) compared to the data subset approach with SM-EPR (\(\boldsymbol{\delta} \sim f(\boldsymbol{\delta}\vert n)\)) in the scenario of model misspecification. The same simulation model from Section \ref{sim_study} is used to generated data, however we misspecify both SM-EPR and M-EPR and use a total of 50 (25 individual and 25 shared) radial basis functions for each data type, where recall 15 individual and shared basis functions were used to generate data. Posterior sampling is done using \(0.8\times 60,000\) of the data and \(0.2\times 60,000\) of the data is held out for testing prediction for each response.

\begin{figure}[H]
\centering
  \centering
  \includegraphics[width=0.75\linewidth]{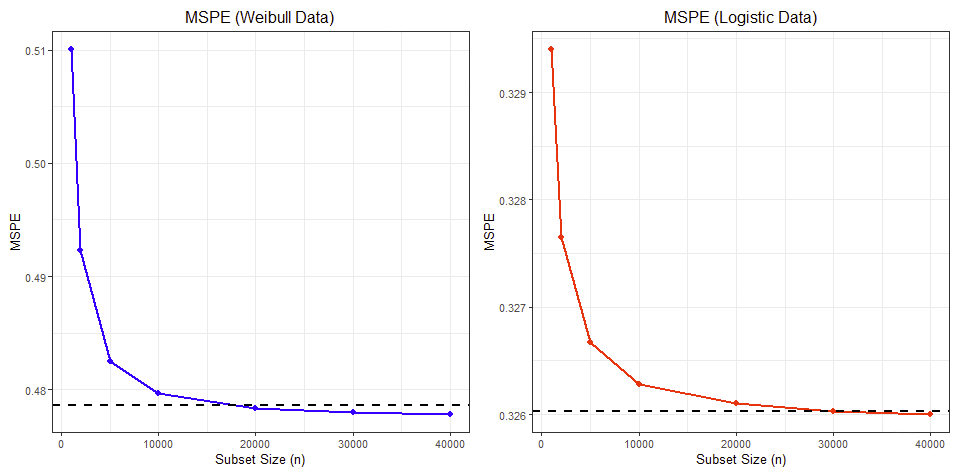}  
\caption{The red and blue lines the the mean MSPE using SM-EPR over various subset sizes. The black dashed line is the mean MSPE using M-EPR.}
\label{misspecify_elbow}
\end{figure}

In Figure \ref{misspecify_elbow} we compare the mean MSPE for SM-EPR with various subset sizes to M-EPR across 50 replicates. The black dashed line indicates the MSPE of M-EPR. The elbow plot indicates that SM-EPR performs better than M-EPR in terms of prediction for subset sizes 20,000, 30,000, and 40,000 in the Weibull scenario and 40,000 for the logistic scenario. The smaller MSPE in these cases is because SM-EPR removes parametric assumptions from the spatial basis function expansion. This suggests that in the presence of model misspecification, SM-EPR can perform better, albeit marginally better, in terms of prediction than M-EPR for select subset sizes. This is consistent with results in \citet{saha2023incorporating}, and provides empirical support for our Theorem \ref{theorem4}, where we see that the data subset approach can aid with model robustness.

\section{Basis Function Sensitivity Analysis}\label{basis_sens}
\begin{table}[H]
\caption{Sensitivity analysis for selecting number of bisquare basis functions to use for modeling small-scale variability in mixed effects model for univariate analysis of \(\text{PM}_{2.5}\).}
\begin{center}
\begin{tabular}{cccccc}
\toprule
Basis Functions & HOVE & PMCC & CRPS & WAIC & CPU Time \\ \midrule
48 & 111.8636 & 116.1026 & 3.5010  & 6.8065 & 24.71 sec \\
92 & 91.3138 & 96.1632 & 3.3415 & 6.7735 & 59.28 sec \\
145 & 98.6512 & 107.4218 & 3.3391 & 6.7863 & 1.99 min \\
198 & 72.7569 & 79.1911 & 2.8239  & 6.7660 & 3.48 min \\ 
259 & 73.9918 & 80.4919 & 2.8235 & 6.7525 & 6.27 min\\
280 & 73.9525 & 80.1861 & 2.8291 & 6.7521 & 7.78 min \\
\bottomrule
\end{tabular}
\end{center}
\begin{flushleft}
\textit{Note}: The first column displays the number of bisquare basis functions used to fit univariate S-EPR with subset size \(n = 10,000\). The hold out validation error (HOVE) on \(0.2N\) of the original dataset using the posterior mean of the latent process is \(\text{HOVE} = \frac{1}{0.2N} \sum_{\{i: i \in D_h\}} (Z_i - E[{Z}_{i}^{new} \vert \textbf{z}])^2 \) where \(D_h\) is the holdout dataset and ${Z}_{i}^{new}$ is predictive data. PMCC is the predictive model choice criterion over the hold out locations. WAIC is the Wantanabe-Akaike information criterion. CRPS is the continuous rank probability score. The last column displays the CPU time.
\end{flushleft}
\end{table}









\end{document}